\documentclass[12pt,oneside]{article}
\usepackage{amsfonts}
\usepackage{amscd}
\usepackage{amsthm}
\usepackage{amssymb}
\usepackage{amsmath}
\usepackage{epsfig}
\usepackage{graphicx}
\usepackage{psfrag}

\newtheorem{Lem}{Lemma}[section]
\newtheorem{Prop}[Lem]{Proposition}

\theoremstyle{plain}

\newtheorem{Cor}[Lem]{Corollary}

\long\def\symbolfootnote[#1]#2{\begingroup%
\def\thefootnote{\fnsymbol{footnote}}\footnote[#1]{#2}\endgroup}

\theoremstyle{remark}

\begin{document}
\begin{center}
\textbf{{\large A Perturbative Approach to Inelastic Collisions in a Bose-Einstein Condensate}}\\[0pt]

\vspace{12pt} R.B. Mann$^{1,2}$, M. B. Young$^{1,3}$, I. Fuentes-Schuller$^{4,5,}$\symbolfootnote[1]{Published before under
maiden name Fuentes-Guridi.}\vspace{12pt}

\smallskip \textit{$^1$ Department of Physics \& Astronomy, University of Waterloo, \\[0pt]
Waterloo, Ontario, N2L 3G1, Canada }

\smallskip \textit{$^2$ Perimeter Institute for Theoretical Physics, \\ [0pt] Waterloo, Ontario, N2L 2Y5, Canada}

\smallskip \textit{$^3$ Department of Mathematics, State University of New York at Stony Brook, \\[0pt]
Stony Brook, New York, 11794-3651, U.S.A. }

\smallskip \textit{$^5$ Institute for Theoretical Physics, Technical University of Berlin, \\[0pt]
Hardenbergstr. 36, D-10623, Berlin, Germany}

\smallskip \textit{$^4$ Institute of Physics, University of Potsdam, \\[0pt]
Am Neuen Palais 10, D-14469, Potsdam, Germany}

\bigskip

mann@avatar.uwaterloo.ca, myoung@math.sunysb.edu, ivette@physik.uni-berlin.de

\vspace{24pt}
\end{center}
\begin{abstract}
\addtolength{\baselineskip}{1.2mm} \addtolength{\baselineskip}{1.2mm}

It has recently been discovered that for certain rates of mode-exchange collisions analytic solutions can be found for a Hamiltonian describing the two-mode Bose-Einstein condensate.  We proceed to study the behaviour of the system using perturbation theory if the coupling constants only approximately match these parameter constraints.  We find that the model is robust to such perturbations.  We study the effects of degeneracy on the perturbations and find that the induced changes differ greatly from the non-degenerate case.  We also model inelastic collisions that result in particle loss or condensate decay as external perturbations and use this formalism to examine the effects of three-body recombination and background collisions.

\end{abstract}

\addtolength{\baselineskip}{1.5mm} \thispagestyle{empty}

\vspace{48pt}

\setcounter{footnote}{0}

\section{Introduction}

A Bose-Einstein condensate (BEC) is a state of matter in which a large number of bosons occupy the same quantum mechanical ground state.  As such, BECs present the opportunity to study quantum systems which display large scale (macroscopic) collective behaviour.  Recently there has been much interest in multi-component BECs because of their importance to quantum optics \cite{milbu,albie,santo,eva}.  Multi-component BECs are most often formed in a multi-well potential in which the components are spatially separated \cite{albie,eva}.  Alternatively, the multi-component formalism can be used to model a single component BEC that possesses several internal degrees of freedom, such as varying amounts of spin \cite{santo}.

In general, many-body systems are of significant importance in physics.  For example, quantum information processes require the manipulation and control of systems with large numbers of particles.  However, many-body systems are often difficult to treat exactly and are most often studied using numerical or approximate methods.  The realm of applicability of these methods is limited by the number of degrees of freedom of the system. Two common approximate models are the Bose-Hubbard model  of quantum optics \cite{eva} and closely related to it, the Lipkin-Meshkov-Glick model of nuclear physics \cite{lipkin}, both used to describe the two-body interactions of spin-$J$ systems.

A family of exactly solvable many-body systems was introduced in \cite{fuent1} and studied in greater depth in \cite{fuent2}. These models can be used to describe the physics of a two component BEC where both elastic and mode-exchange inelastic collisions are present.  Mode-exchange collisions are known as general nearest neighbor interactions in the context of multi-well BECs \cite{nearest} and as inelastic collisions in the case of a single condensate consisting of particles in two hyperfine levels \cite{julienne}. The family of models is parameterized by a positive integer $n$, and hence the specific models are called $n$-models.  This terminology is used because the $n$-model contains $1,2, \dots, n-$body interactions.  Analytic solutions can be found for the $n$-model when the strengths of the various particle interactions obey specific constraints.  While the $n$-models are of interest to many-body physics in general, in this paper we will largely be concerned with the $2$-model where only single body interactions and two-body collisions are considered.  Microscopic calculations show that mode-exchange collisions occur in BECs as a result of the interaction of a laser field with the system \cite{julienne}.  The 2-model includes the usual Josephson-type interactions, but also allows the effects of mode-exchange collisions to be studied analytically, and as such provides a more realistic framework for studying two-mode BECs than the canonical Josephson Hamiltonian \cite{corne}.

The analytic solution found in \cite{fuent1} requires that the strengths of the various interactions meet certain constraints; if these constraints are not met the solution is invalid.  While in some experimental situations the rates of inelastic and elastic mode-exchange collisions can be controlled externally \cite{roberts}, this is not always the case.  Moreover, even in if the rate of collisions can be manipulated, the constraints will likely still only be approximately satisfied.  It is thus of interest to extend the solution space of the system to the case in which these constraints are only approximately satisfied.  This naturally leads to the consideration of small parameter perturbations in the model to study its robustness.  Along with perturbations within the model it is of interest to include additional interaction terms as perturbations, such as inelastic collisions resulting in particle loss. Particle loss, often suppressed in experimental settings, is studied theoretically by considering classical rate equations \cite{corne}.  More recently, a quantum treatment of 3-body loss has been given in \cite{quantumloss}. Including external perturbations extends the predictive power of the model. Additional interactions of primary interest include three-body recombination, background collisions, spin exchange and dipole relaxation \cite{corne}. To our knowledge, this is the first study which analyzes particle loss as a perturbative effect in BECs.

Motivated by the above considerations, in this paper we carry out a full perturbative analysis of the solvable BEC 2-model proposed in \cite{fuent1}. We begin by reviewing this model and the solutions derived when the constraint equations on the parameters are satisfied \cite{fuent2}.  We then introduce perturbations to the parameters in the model and analyze the effects on the particle distribution, entanglement and the evolution of the relative number operator.  In later sections we discuss the effects of state degeneracy and the inclusion of a general loss term as a perturbation, illustrating the formalism by studying background collisions and three-body recombination.

\section{A Model for Two-Mode Bose-Einstein Condensates with Mode-exchange Collisions}

The $2$-model studied in \cite{fuent1, fuent2} is governed by the Hamiltonian
\begin{equation}
\begin{array}{lcl}
	\displaystyle H_2& = &A_0 + \omega \left(a^{\dagger}a - b^{\dagger}b \right) +\lambda\left( e^{i\phi}a^{\dagger}b +e^{-i\phi}ab^{\dagger} \right) + \mathcal{U}a^{\dagger}ab^{\dagger}b\\ \\ & & +\Lambda \left(e^{2i\phi}a^{\dagger}a^{\dagger}bb + h.c. \right) +\mu\left\{ e^{i\phi}\left(a^{\dagger}a^{\dagger}ab - a^{\dagger}b^{\dagger}bb \right) +h.c. \right\}  \label{becHam}
\end{array}
\end{equation}
where $h.c.$ denotes the Hermitian conjugate of the preceding term.  The two modes  $a$ and $b$ are independent Bose operators satisfying $[a,a^{\dagger}]=1=[b,b^{\dagger}]$ with the commutators of all remaining pairs vanishing.  The term $\omega \left(a^{\dagger}a - b^{\dagger}b \right)$ is the free energy of $\hat{n}_a=a^{\dagger}a$ particles in the $a$ mode and $\hat{n}_b=b^{\dagger}b$ particles in the $b$ mode, with frequency difference $\omega$ between modes.  This frequency difference arises because the model describes atoms in different hyperfine levels, or alternatively unsymmetric spatially separated condensates \cite{steng}.  A Josephson-type or spin flip interaction is included with strength $\lambda$ and phase $\phi$.  In practice such an interaction is induced by an external field, such as a laser \cite{myatt}.  This interaction may also be interpreted as modeling the tunneling of particles between modes with probability proportional to $\lambda$.  Also included are terms corresponding to number-preserving elastic and mode-exchange collisions, namely those terms containing four bose operators.  The interaction with strength $\mu$ is a single dispersive process.  The interaction with strength $\Lambda$ describes a collision where two particles exchange their mode.  The elastic interaction has strength $\mathcal{U}$ and models the collision of a particle from each mode in which the number of particles in each mode is conserved.  We see that the mode-exchange collisions preserve the total particle number but not the relative particle number.

The Hamiltonian given in eq. (\ref{becHam}) can be efficaciously studied by first introducing the simpler Hamiltonian \cite{fuent1}
\begin{equation}
H_0:= H_{0,2}=A_1\left(a^{\dagger}a-b^{\dagger}b\right) + A_2\left( a^{\dagger}a-b^{\dagger}b \right)^2  \label{simpHam}
\end{equation}
with real constants $A_1$ and $A_2$.  Such a Hamiltonian models a two-mode condensate with energy difference $A_1$ between modes and elastic scattering probability proportional to $A_2$. As is clear from the form of $H_0$ the relative number operator $\hat{m}=a^{\dagger}a-b^{\dagger}b$ is a commuting observable so that the number of particles in each mode is conserved.  In particular, there is no probability of spin-flip or tunneling between modes. We  also note that the total number operator $\hat{N}=a^{\dagger}a+b^{\dagger}b$ commutes with $H_0$.  We may thus take the eigenvalues $N$ and $m$ of $\hat{N}$ and $\hat{m}$ as labels of the eigenstates, for which we write $\vert N,m \rangle$.  For fixed $N\in\mathbb{N}$  the values of $m$ are limited to\footnote{See the appendix for details.} $-N, -N +2 , \dots, N-2$ and $N$.
 The energy of the state $\vert N,m \rangle$ is $E_m= A_1 m +A_2m^2$.

For a certain choice of parameters in $H_2$ the solutions $\vert N, m \rangle$ of $H_0$ can be used to obtain analytic solutions of $H_2$.  To see this, define a two-mode displacement operator by $U\left(\xi \right)= \exp \left(\xi a^{\dagger}b -\xi^* ab^{\dagger} \right)$ with displacement parameter $\xi=\frac{1}{2}\theta e^{i\phi}$.  It is clear that $U\left(\xi\right)$ is unitary.  It can then be shown that if the  parameters in $H_2$ satisfy
\begin{subequations}  \label{coeffId}
	\begin{equation}
		A_0= A_2 \left( N^2 \cos ^2 \theta + N \sin^2 \theta  \right)
	\end{equation}
	\begin{equation}
		\omega=A_1\cos \theta
	\end{equation}
	\begin{equation}
		\lambda=A_1\sin\theta
	\end{equation}
	\begin{equation}
		\mathcal{U}=2A_2\left(1-3\cos^2\theta\right)
	\end{equation}
	\begin{equation}
		\Lambda=A_2\sin^2\theta
	\end{equation}
	\begin{equation}
		\mu=2A_2\cos\theta \sin\theta
	\end{equation}
\end{subequations}
then $H_2=U^{\dagger} H_0 U$, a result shown by computing
\begin{equation}
U^{\dagger}\left(
\begin{array}{c} a\\ b \\  a^{\dagger}  \\ b^{\dagger} \\
\end{array}
\right) U= \left(
\begin{array}{cccc} \cos\frac{1}{2}\theta & e^{i\phi}\sin \frac{1}{2}\theta & 0 & 0\\ -e^{-i\phi}\sin\frac{1}{2}\theta & \cos \frac{1}{2}\theta & 0 & 0 \\  0 & 0 &\cos\frac{1}{2}\theta & e^{-i\phi}\sin \frac{1}{2}\theta  \\ 0& 0& -e^{i\phi}\sin\frac{1}{2}\theta & \cos \frac{1}{2}\theta \\
\end{array}
\right) \left(
\begin{array}{c} a\\ b \\  a^{\dagger}  \\ b^{\dagger} \\
\end{array}  \label{trans}
\right).
\end{equation}
Observe that $
U\left(
\begin{array}{cccc} a & b &  a^{\dagger}  & b^{\dagger} \\
\end{array}
\right)^T U^{\dagger}$
may be computed by setting $\theta\mapsto -\theta$ in eq. (\ref{trans}).

Since the eigenvectors of $H_0$ are of the form $\vert N,m\rangle$ the eigenvectors of $H_2$ when satisfying eqs. (\ref{coeffId}) are simply $U^{\dagger}\vert N,m\rangle$ with energy $E_m=A_1 m +A_1m^2$.  An extensive analysis and discussion of these solutions is included in \cite{fuent2}.

For fixed $N\in\mathbb{N}$ the ground state of $H_2$, labeled as $\vert \psi_{m_0} \rangle = U^{\dagger} \vert N, m_0 \rangle$, is found by minimizing the energy $E_m=A_1 m +A_2 m^2$ with respect to $m$.  First assume $A_2 <0$.  In this case, if $A_1<0$ ($A_1 >0$) the minimum occurs at $m_0=N$ ($m_0=-N$).  Secondly, say $A_2 >0$.  If $\left| \frac{A_1}{2A_2} \right| \leq N$, then $m_0$ is the closest allowable\footnote{The integer $m$ can only take the values $-N, -N+2,\dots,N-2$ and $N$; see the appendix for details.} integer to $ -\frac{A_1}{2A_2} $; otherwise $m_0$ is the closest allowable integer to $-\hbox{sgn} (A_1) N$.  In particular, a choice of either $m_0$ or $\frac{A_1}{A_2}$ determines the other.\footnote{In calculations we will assume $A_2=1$, so that a choice of $m_0$ determines $A_1$.}

It is important to recall that, in the case of a double-well BEC, the two-mode approximation must be satisfied \cite{milbu}. This means that collisions taking place in the region where the wavefunctions overlap must be less probable than collisions between particles belonging to the same well. It is possible to find an exact analytic solution to the model in this case by considering $\theta<<1$. It is interesting to observe that it is not possible to find analytical solutions to this model for an exactly symmetric double-well and satisfy, simultaneously, the two mode approximation. Fortunately, in the case of slightly asymmetric double wells (which corresponds to a more realistic situation) it is possible to find exact analytical solutions in the two-mode approximation \cite{fuent2} .

The most general $n$-model Hamiltonian discussed in \cite{fuent1} has as its Hamiltonian $H_n= U^{\dagger} H_{0,n} U$ where $H_{0,n} = \sum_{k=1}^n A_k \hat{m}^k$.  The eigenvectors of $H_n$ are again $U^{\dagger} \vert N, m \rangle$ with the same energy $E_m$.  The same process of matching coefficients  in the expansion of $U^{\dagger} H_{0,n} U$ that was used to obtain eqs. (\ref{coeffId}) can be carried out so that $U^{\dagger} \vert N, m \rangle$ is the exact solution of a Hamiltonian containing  up to $n$-body interactions.  In the general case, there are $n+2$ free variables (the $n$ $A_i$ as well as $\theta$ and $\phi$) in $U H_{0,n}U^{\dagger}$, whereas the number of terms in a general Hamiltonian describing such interactions is $(n+2)(n+4)(n+6)/48$ as shown in Appendix B.  So while the above method continues to give analytic solutions for all $n$, its range of applicability decreases as $n$ grows.  It should also be noted that because $H_{0,n}$ contains only terms with the same number of creation and annihilation Bose operators, this formalism cannot be used to study interactions involving an odd number of Bose operators, such as the interaction $a^{\dagger}ba$ where one particle is lost.

We remark that all of the perturbative analysis carried out below is valid not only for the $2$-model, but also for the more general $n$-models.  This is so because the solutions for all $n$-models are $U^{\dagger}\vert N,m \rangle$.  We need only change the interpretation of the perturbative analysis in each case.

\section{Effects of Parameter Perturbations in the $2$-Model}

While satisfying eqs. (\ref{coeffId}) is sufficient to obtain an analytic solution of $H_2$, the method presented above fails to produce solutions if the parameters in $H_2$ deviate even slightly from these constraints.  Thus, in order to compare this model to experimental results we would like to study the solutions to $H_2$ when the constraints are only approximately satisfied. We proceed with these perturbation calculations in the section below.

We begin by perturbing each of the coupling constants in $H_2$ away from the conditions given by eqs. (\ref{coeffId}), i.e. perturbations of the form $\omega \mapsto \omega + \delta_{\omega}$ where $\omega$ satisfies eqs. (\ref{coeffId}) and $\delta_{\omega}$ is small.    We omit the study of perturbations $A_0 \mapsto A_0 + \delta_{A_0}$ since this simply leads to a shift in the energy of the state by $\delta_{A_0}$.  In what follows, we assume that all eigenstates are non-degenerate; the degenerate case is discussed in Section 7.  Throughout the paper, given an operator $\mathcal{O}$ we define $\tilde{\mathcal{O}} := U\mathcal{O}U^{\dagger}$.  To perform many of the calculations that follow we have made use of a coordinate transformation between the $\left\{ N,m \right\}$ basis and the $\left\{ n_a, n_b \right\}$ basis, $\vert N,m \rangle \mapsto \left| n_a= \frac{N+m}{2}, n_b=\frac{N-m}{2} \right>$ where $n_a$ is a eigenvalue of $\hat{n}_a = a^{\dagger}a$ and similarly for $n_b$.

\subsection{$\omega$ Perturbation}
A change $\omega \mapsto \omega + \delta_{\omega}$ results in the perturbation $H_{\omega}= \delta_{ \omega} \left(a^{\dagger}a - b^{\dagger}b \right)$.  A calculation then shows that the non-vanishing matrix elements of $\tilde{H}_{\omega}$ are
\begin{equation}
\langle N,m \vert \tilde{H}_{\omega} \vert N,m \rangle = \delta_{\omega} m\cos\theta
\end{equation}
and
\begin{equation}
\langle N,m \pm 2\vert \tilde{H}_{\omega} \vert N,m \rangle = -\frac{\delta_{\omega}}{2}e^{\pm i\phi}\sin\theta \sqrt{N(N+2) -m(m\pm 2) }.
\end{equation}

\subsection{$\lambda$ Perturbation}
A change $\lambda \mapsto \lambda + \delta_{\lambda}$ results in the perturbation $H_{\lambda}= \delta_{\lambda} \left( e^{i\phi}a^{\dagger}b +e^{-i\phi}ab^{\dagger} \right)$ and we find that
\begin{equation}
\langle N,m \vert \tilde{H}_{\lambda} \vert N,m \rangle = \delta_{\lambda} m\sin\theta
\end{equation}
and
\begin{equation}
\langle N,m \pm 2\vert \tilde{H}_{\lambda} \vert N,m \rangle = \frac{\delta_{\lambda}}{2}\left( e^{\pm i \phi}\cos^2 \frac{1}{2}\theta - e^{\pm 2i\phi} \sin^2\frac{1}{2}\theta \right) \sqrt{N(N+2) -m(m\pm 2) }.
\end{equation}

\subsection{$\mathcal{U}$ Perturbation}
A change $\mathcal{U} \mapsto \mathcal{U} + \delta_{\mathcal{U}}$ results in the perturbation $H_{\mathcal{U}}= \delta_{\mathcal{U}} a^{\dagger}ab^{\dagger}b$ and we find that
\begin{equation}
\langle N,m \vert \tilde{H}_{\mathcal{U}} \vert N,m \rangle = \frac{\delta_{\mathcal{U}}}{4}\sin^2\theta \left(\frac{N^2 +m^2}{2} +N \right) + \delta_{\mathcal{U}}\cos^2\theta \left(\frac{N^2-m^2}{4} \right),
\end{equation}

\begin{equation}
\langle N,m \pm 2\vert \tilde{H}_{\mathcal{U}} \vert N,m \rangle = \frac{\delta_{\mathcal{U}}}{8}e^{\pm i \phi}\sin^2\theta \sqrt{\left(N \pm m\right) \left(N \pm m + 2 \right) } \left( m \pm 1\right)
\end{equation}
and
\begin{equation}
\langle N,m \pm 4\vert \tilde{H}_{\mathcal{U}} \vert N,m \rangle = -\frac{\delta_{\mathcal{U}}}{16}e^{\pm 2i \phi}\sin^2\theta \sqrt{\left[N(N+2) -m(m\pm 2) \right] \left[ N(N+2) -(m \pm 2)(m\pm 4) \right] }.
\end{equation}

\subsection{$\Lambda$ Perturbation}
A change $\Lambda \mapsto \Lambda + \delta_{\Lambda}$ results in the perturbation $H_{\Lambda}= \delta_{\Lambda} \left( e^{2i\phi}a^{\dagger}a^{\dagger}bb + h.c. \right)$ and we find that
\begin{equation}
\langle m \vert \tilde{H}_{\Lambda} \vert m \rangle = \frac{\delta_{\Lambda}}{2} \sin^2 \theta \left(\frac{3m^2-N^2}{2}-N\right),
\end{equation}

\begin{equation}
\langle m \pm 2 \vert \tilde{H}_{\Lambda} \vert m \rangle = \frac{\delta_{\Lambda}}{4} e^{\mp i\phi} \sin 2\theta \sqrt{ \left( N\mp m \right) \left(N\pm m+2\right)}\left(m\pm 1\right),
\end{equation}
and
\begin{equation}
\langle m \pm 4 \vert \tilde{H}_{\Lambda} \vert m \rangle = \frac{\delta_{\Lambda}}{8} e^{\mp 2i\phi} \left( 1 + \cos^2 \theta \right) \sqrt{ \left(N \pm m+2\right)\left(N\pm m+4\right)\left(N\mp m\right)\left(N \mp m-2\right)}.
\end{equation}

\subsection{$\mu$ Perturbation}
A change $\mu \mapsto \mu + \delta_{\mu}$ results in the perturbation $H_{\mu}= \delta_{\mu} \left\{ e^{i\phi}\left( a^{\dagger}a^{\dagger}ab - a^{\dagger}b^{\dagger}bb \right) +h.c. \right\}$ and we find that
\begin{equation}
\langle m \vert \tilde{H}_{\mu} \vert m \rangle = \frac{\delta_{\mu}}{2}\sin 2\theta \left(\frac{3m^2-N^2}{2}-N\right),
\end{equation}

\begin{equation}
\langle m \pm 2 \vert \tilde{H}_{\mu} \vert m \rangle = \frac{\delta_{\mu}}{2} e^{\mp i\phi} \cos 2\theta  \sqrt{ \left(N\mp m\right)\left(N\pm m+2\right)}\left(m\pm 1\right)
\end{equation}
and
\begin{equation}
\langle m \pm 4 \vert \tilde{H}_{\mu} \vert m \rangle = -\frac{\delta_{\mu}}{8} e^{\mp 2i\phi}  \sin 2\theta  \sqrt{ \left(N \pm m+2\right)\left(N\pm m+4\right)\left(N\mp m\right)\left(N \mp m-2\right)}.
\end{equation}

In each of the cases above the perturbation matrix elements simplify significantly in the limiting cases $\theta=0, \pi$, in which the Josephson coupling vanishes ($\lambda =0$), and $\theta=\frac{\pi}{2}$ in which the potential well is symmetric ($\omega =0$).  We elaborate below on the effect of $\theta$ on the particle distributions.

\section{Perturbative Effects on Particle Distribution}

In the unperturbed case we can use the analytic solution $\vert \psi_{m_0} \rangle = U^{\dagger} \vert N, m_0 \rangle$ for parameters satisfying eqs. (\ref{coeffId}) to find an explicit expression for the particle distribution:
\begin{equation}
P=  \vert \langle N,m \vert U^{\dagger} \vert N,m_0 \rangle \vert^2.  \label{probDens}
\end{equation}
Using the homomorphism described in the appendix  to relate the Schwinger $\mathfrak{su}(2)$ representation to the angular momentum representation we  write $P = \vert d_{m,m_0}^N \vert^2$ where\footnote{$d_{m,m_0}^N$ are the Wigner rotation matrix elements under the effect of the Lie algebra homomorphism discussed in the appendix. To see why these appear, observe that $U^{\dagger}=e^{-i\theta\mathbf{n} \cdot \mathbf{J}}$ where $\mathbf{n}=\left(\sin\phi, \cos\phi ,0 \right)$ and $\mathbf{J}=\left( \hat{J}_x, \hat{J}_y, \hat{J}_z \right)$ is the total angular momentum vector operator.  That is, $U^{\dagger}$, and hence $U$, are rotations of the algebra generated by $a$ and $b$.} \cite{gottf}
\begin{equation}
d_{m,m_0}^N=(-1)^{\frac{1}{2}\left(m-m_0\right)} \sqrt{\frac{\left(\frac{N+m}{2}\right)!\left(\frac{N-m}{2}\right)!}{\left(\frac{N+m_0}{2}\right)!\left(\frac{N-m_0}{2}\right)!}} \left(\cos\frac{\theta}{2}\right)^{N} \Sigma
\end{equation}
and
\begin{equation}
\Sigma= \sum_{k=k_-}^{k_+} (-1)^k\left(
\begin{array}{c} \frac{1}{2}\left(N+m_0\right)\\ k \\
\end{array}
\right) \left(
\begin{array}{c} \frac{1}{2}\left( N-m_0 \right)\\ \frac{1}{2}\left(N-m\right)-k \\
\end{array}
\right)\left(\tan \frac{\theta}{2} \right)^{\frac{1}{2}\left(m-m_0\right)+2k}.
\end{equation}
The integers $k_{\pm}$ are chosen so that the arguments of the combinatorial symbols are non-negative; explicitly, $k_-=\hbox{max}\left\{ 0,\frac{m_0-m}{2} \right\}$ and $k_+=\hbox{min} \left\{ \frac{N-m}{2}, \frac{N+m_0}{2} \right\}$.   We plot in Figure \ref{partDistPlot} the unperturbed particle distribution for $N=1000$, $\theta=1$ and $m_0=1000,\;998$ and $996$.  Observe that the particle distribution is independent of phase $\phi$.  We will see that this does not hold in the perturbed case.  As $\theta$ varies the distributions shift along the $m$-axes; for $\theta$ small the maxima shift toward $m=N$, and move toward $m=-N$ as $\theta$ grows.

The canonical $2$-mode BEC predicts (under certain circumstances) that the ground state solution is a superposition of two peaked distributions.  From Figure \ref{partDistPlot} we see that for $m_0 < N-2$ the ground state is a superposition of more than two distributions, an effect that is due to the mode-exchange collisions not included in the canonical model \cite{fuent2}.

\begin{figure}[t]
\psfrag{P}{$P$}
\psfrag{La}{$m$}
\psfrag{1000}{\textit{a)}$m_0=1000$}
\psfrag{998}{\textit{b)}$m_0=998$}
\psfrag{996}{\textit{c)}$m_0=996$}
\centering
\includegraphics[width=1.6in, height=1.6in]{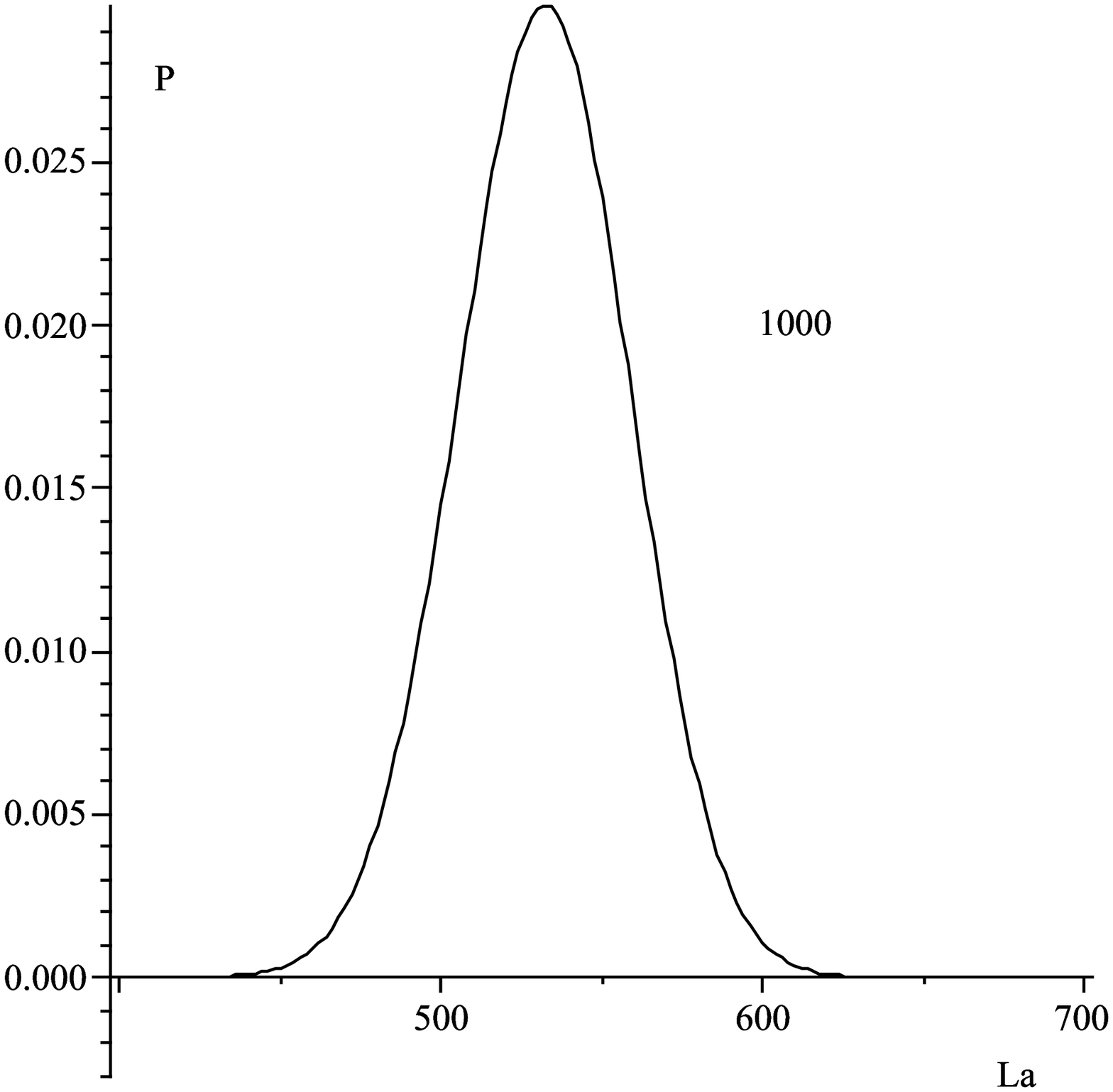}
\hspace{0.65in}
\includegraphics[width=1.6in, height=1.6in]{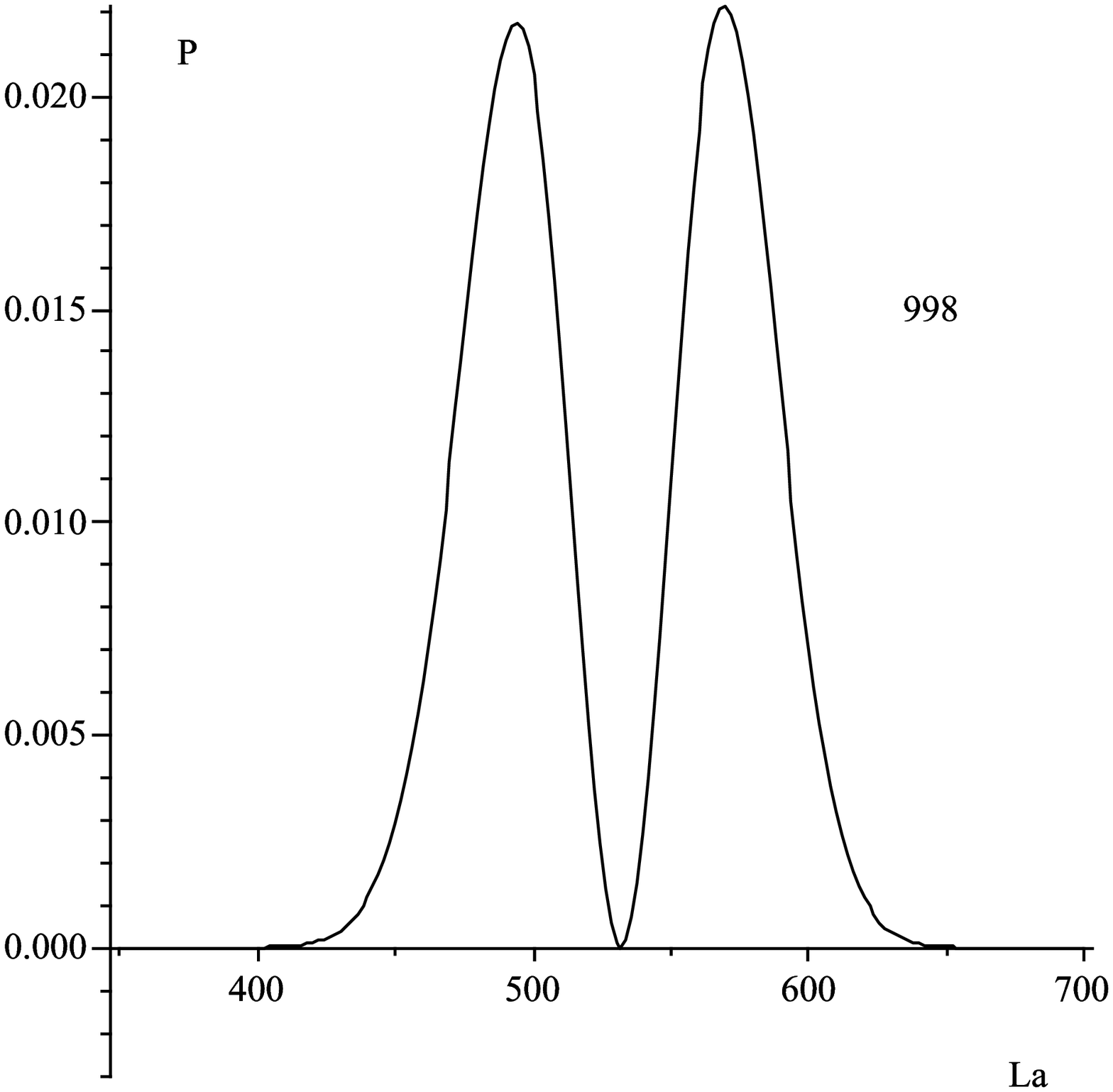}
\hspace{0.65in}
\includegraphics[width=1.6in, height=1.6in]{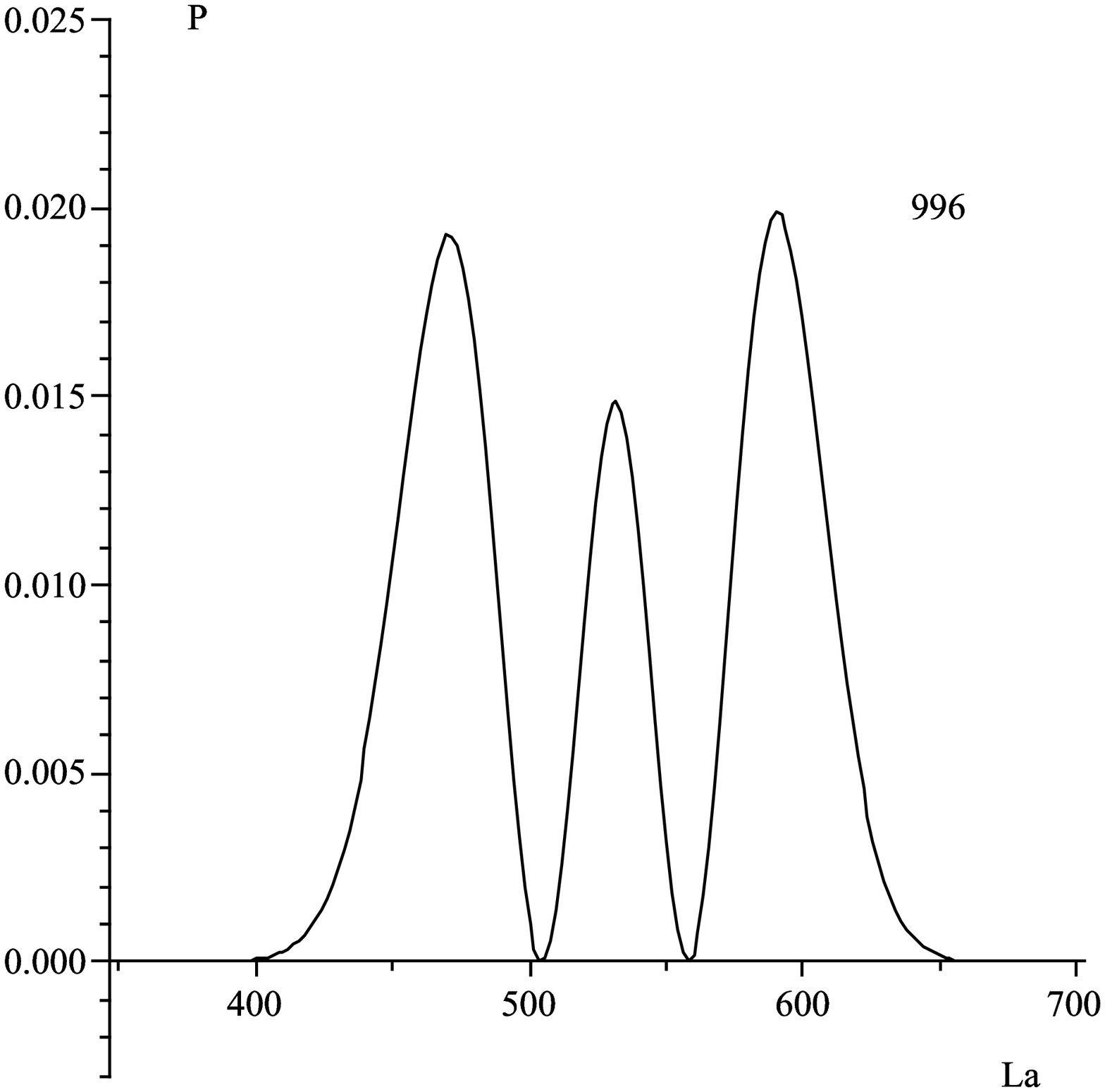}
\caption{The figures show the unperturbed particle distributions $P$ given by eq. (\ref{probDens}) for $N=1000$, $\theta=1$ and \textit{a)} $m_0=1000$, \textit{b)} $m_0=998$ and \textit{c)} $m_0=996$.}
\label{partDistPlot}
\end{figure}

In order to include perturbative effects in the particle distribution we replace the zeroth order ground state wave function $U^{\dagger}\vert N , m_0 \rangle$ in eq. (\ref{probDens}) with that including first order corrections,
\[
U^{\dagger}\vert N,m_0 \rangle^{(0+1)} = U^{\dagger} \left(\vert N, m_0 \rangle + \vert N,m_0 \rangle^{(1)} \right).
\]
Recall that in the non-degenerate case a perturbation $H^{\prime}$ induces a first order wave function correction given by
\begin{equation}
\vert N, m_0 \rangle^{(1)}= \sum_{m\neq m_0} a_{m_0,m} \vert N, m \rangle :=  \sum_{m \neq m_0} \frac{\langle N, m \vert \tilde{H}^{\prime} \vert N,m_0 \rangle}{E_{m_0} - E_m} \vert N,m \rangle    \label{nonDegPert}
\end{equation}
where the superscripts indicate to what order in the perturbation the term corresponds; when no superscript is present the result is to zeroth order. With this, we find that to first order the particle distribution is given by\footnote{Throughout the rest of the paper, where appropriate, we write $\displaystyle\sum_{n=-N}^{N}$ with the understanding that we actually mean $\displaystyle\sum_{n\in \left\{ -N, -N+2, \dots, N-2, N \right\}}$. Summations of this form arise because of the restriction on the eigenvalues of the relative number operator.}
\begin{equation}
P^{(0+1)} = \vert d_{m,m_0}^N \vert^2 +2d_{m,m_0}^N \sum_{n=-N}^N  \hbox{Re}\left( a_{m_0,n} \right)d_{m,n}^N .  \label{pertPD}
\end{equation}
We show in Figures \ref{partDistOmega} - \ref{partDistMu} the particle distributions under each of the five parameter perturbations considered above.  Note that in certain cases $P^{(0+1)}$ may be negative because the first order term in eq. (\ref{pertPD}) depends linearly on the perturbation coefficients $a_{m,n}$; this of course is just a reflection that only first order corrections have been taken into account.

We see from eq. (\ref{pertPD}) and the fact that $a_{m,n}$ is proportional to the sign of the perturbation that each of the perturbations can be used to either enhance or diminish the particle distributions in the vertical direction for a given $m$. That is, given a fixed relative number $m$, by choosing the sign of the perturbation appropriately we can increase or decrease the probability of the system being in the state $\vert N, m \rangle$.   We also see that  as $\vert m_0 \vert$ decreases the number of local maxima of the particle distribution increases, and the asymmetry of the perturbations becomes more significant, vertically diminishing probability amplitudes on one side of the central maxima while vertically stretching the amplitudes on the opposite side. The figures also show that the system is much less sensitive to perturbations in $\omega$ and $\lambda$ than it is to perturbations in the collision coupling constants; in the figures the perturbations of $\omega$ and $\lambda$ are between $7$ and $100$ times greater than those in the collision perturbations ($\Lambda, \mu$ and $\mathcal{U}$), while the corrections in all of the cases are of the same scale. Moreover, out of the collision parameters the particle distribution is most sensitive to perturbations in coupling constants of the \textit{mode-exchange} collisions ($\mu$ and $\Lambda$).  It is thus beneficial to maximize the matching of the mode-exchange collision constraints  (\ref{coeffId}) through redefinition of parameters, allowing the perturbation to have a larger effect on the $\omega$, $\lambda$ and $\mathcal{U}$ terms.

The effect on the particle distributions of varying $\theta$ is to shift the position of the central maxima on the $m$-axis: for $\theta=\frac{\pi}{2}$, corresponding to the case in which the Josephson-type interaction is maximal ($\lambda=A_1$) and the energy of each mode is equal ($\omega=0$), the distribution is centered around $m=0$, while for $\theta = 0$ ($\theta=\pi$), in which the Josephson-type interaction vanishes ($\lambda=0$) and the energy difference of the modes is maximal ($\vert \omega \vert= A_1$) it is centered $m=N$ ($m=-N$).  In each of the cases, despite being centered around different values of $m$ the perturbations display the same qualitative behaviour for any choice of $\theta$.  Throughout the paper we have chosen $\theta=1$ in the figures as this lies between the extremes of $\theta=0$ and $\theta = \frac{\pi}{2}$. As $\theta$ varies the effects of the perturbations vary in size, but show the same characteristics.  For example, as $\theta$ decreases from $1$ to $0$ (but with $\delta$ fixed), the perturbations have a larger effect on the particle distributions, which indicates that for small values of $\theta$, we must be more careful about the size of the chosen $\delta$.  This can be explained as follows.  In the perturbed case the area under the particle distributions sums to unity only to first order in the perturbation.  This sum is also a function of $\theta$ (since the perturbed terms are), unlike in the unperturbed case.  We have found that as $\theta$ decreases the sum of probabilities decreases too, which explains why there is a difference in the particle distributions when we keep the size of the perturbations $\delta$ fixed.  If we decrease both $\theta$ and $\delta$, then we find that the particle distributions essentially only shift along the $m$-axis, but not in their general shape.

\begin{figure}[t]
\psfrag{P}{$P$}
\psfrag{om}{$m$}
\psfrag{1000}{\textit{a)}$m_0=1000$}
\psfrag{998}{\textit{b)}$m_0=998$}
\psfrag{996}{\textit{c)}$m_0=996$}
\centering
\includegraphics[width=1.6in, height=1.6in]{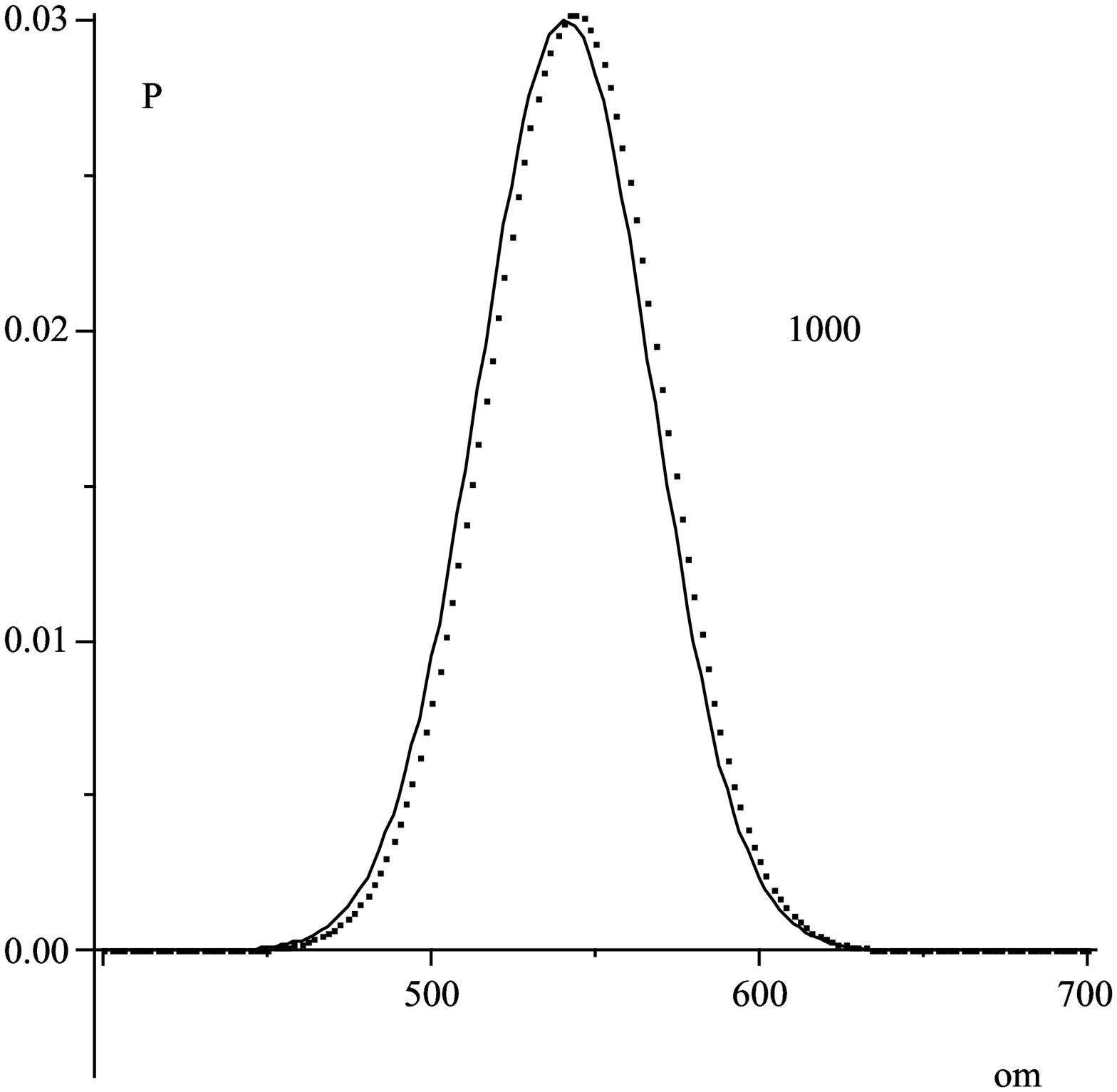}
\hspace{0.65in}
\includegraphics[width=1.6in, height=1.6in]{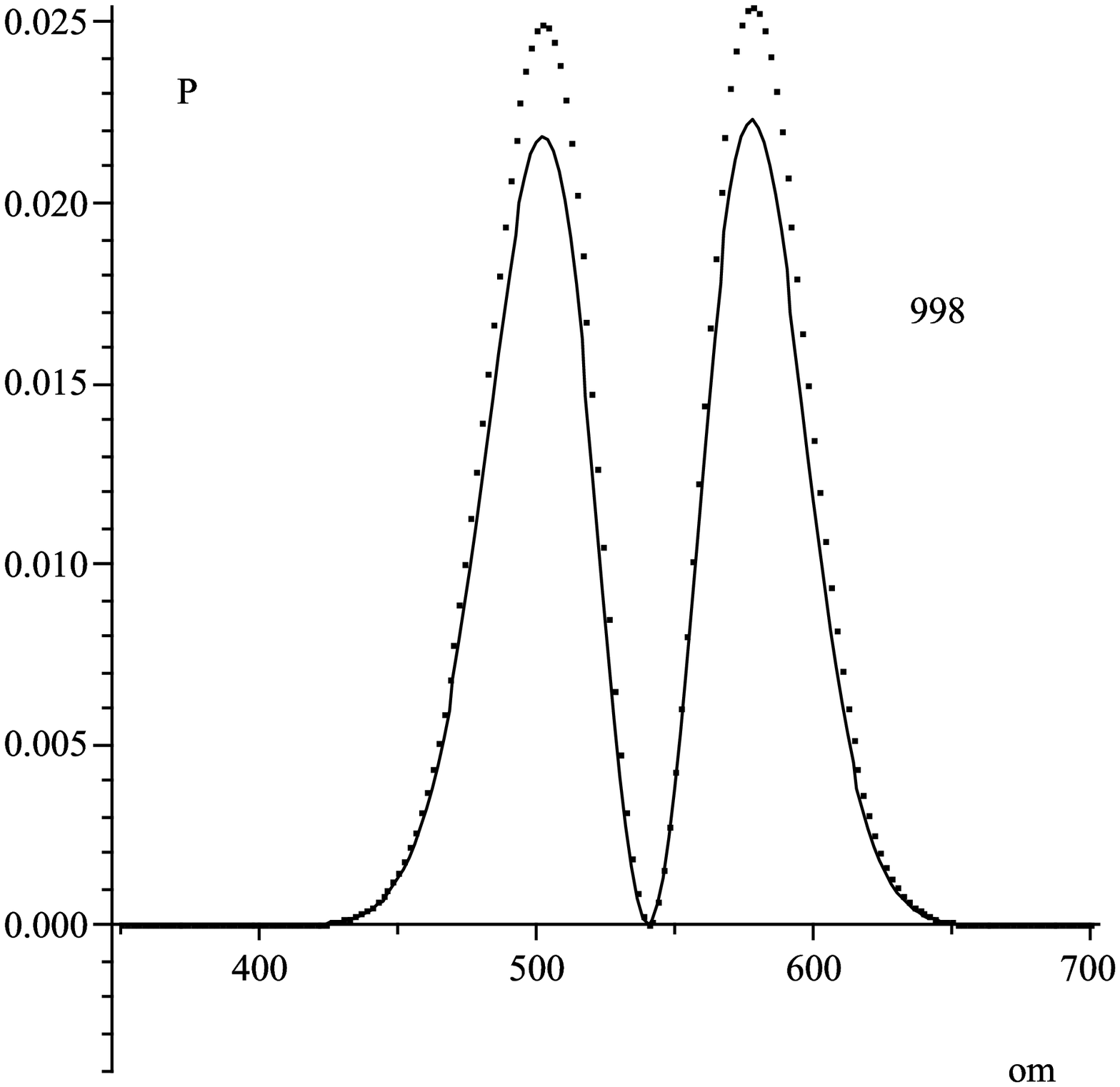}
\hspace{0.65in}
\includegraphics[width=1.6in, height=1.6in]{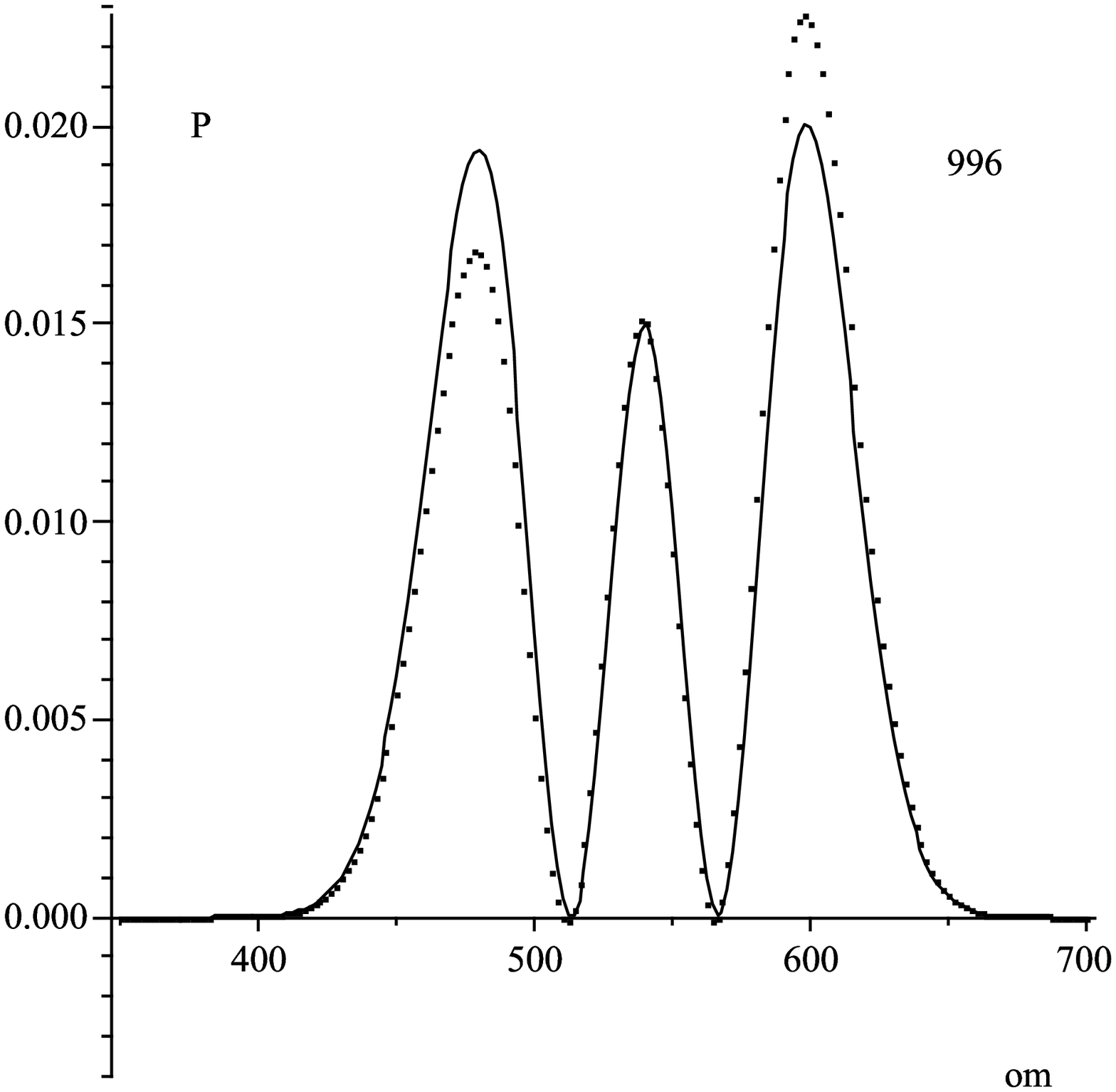}
\caption{The figures show the unperturbed (solid) and perturbed (dot) particle distributions $P$ given by eqs. (\ref{probDens}) and (\ref{pertPD}) for $N=1000$, $\theta=1$, $\delta_{\omega} =15$ and \textit{a)} $m_0 =1000$, \textit{b)} $m_0=998$ and \textit{c)} $m_0=996$.}
\label{partDistOmega}
\end{figure}

\begin{figure}[t]
\psfrag{P}{$P$}
\psfrag{la}{$m$}
\psfrag{1000}{\textit{a)}$m_0=1000$}
\psfrag{998}{\textit{b)}$m_0=998$}
\psfrag{996}{\textit{c)}$m_0=996$}
\centering
\includegraphics[width=1.6in, height=1.6in]{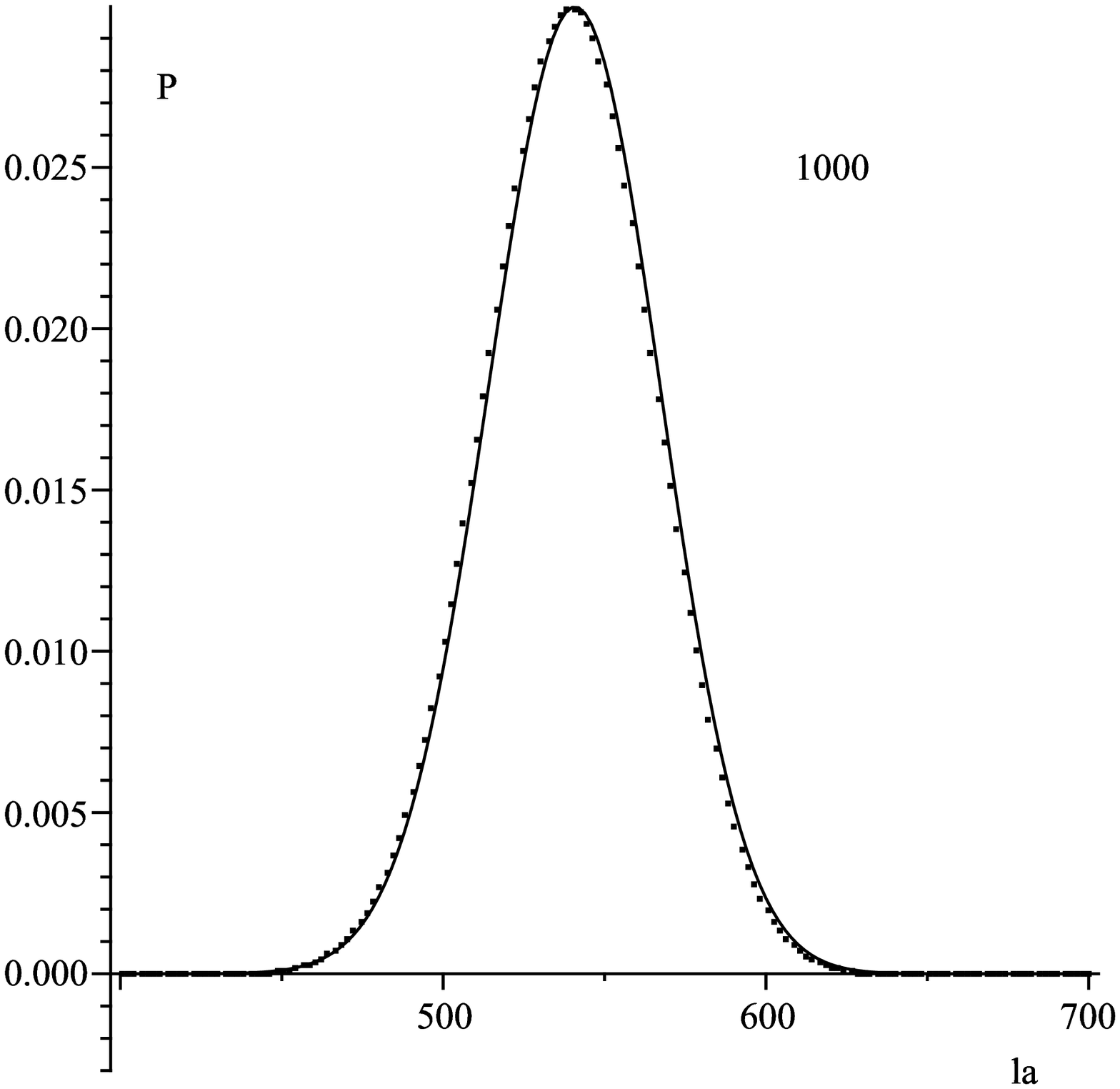}
\hspace{0.65in}
\includegraphics[width=1.6in, height=1.6in]{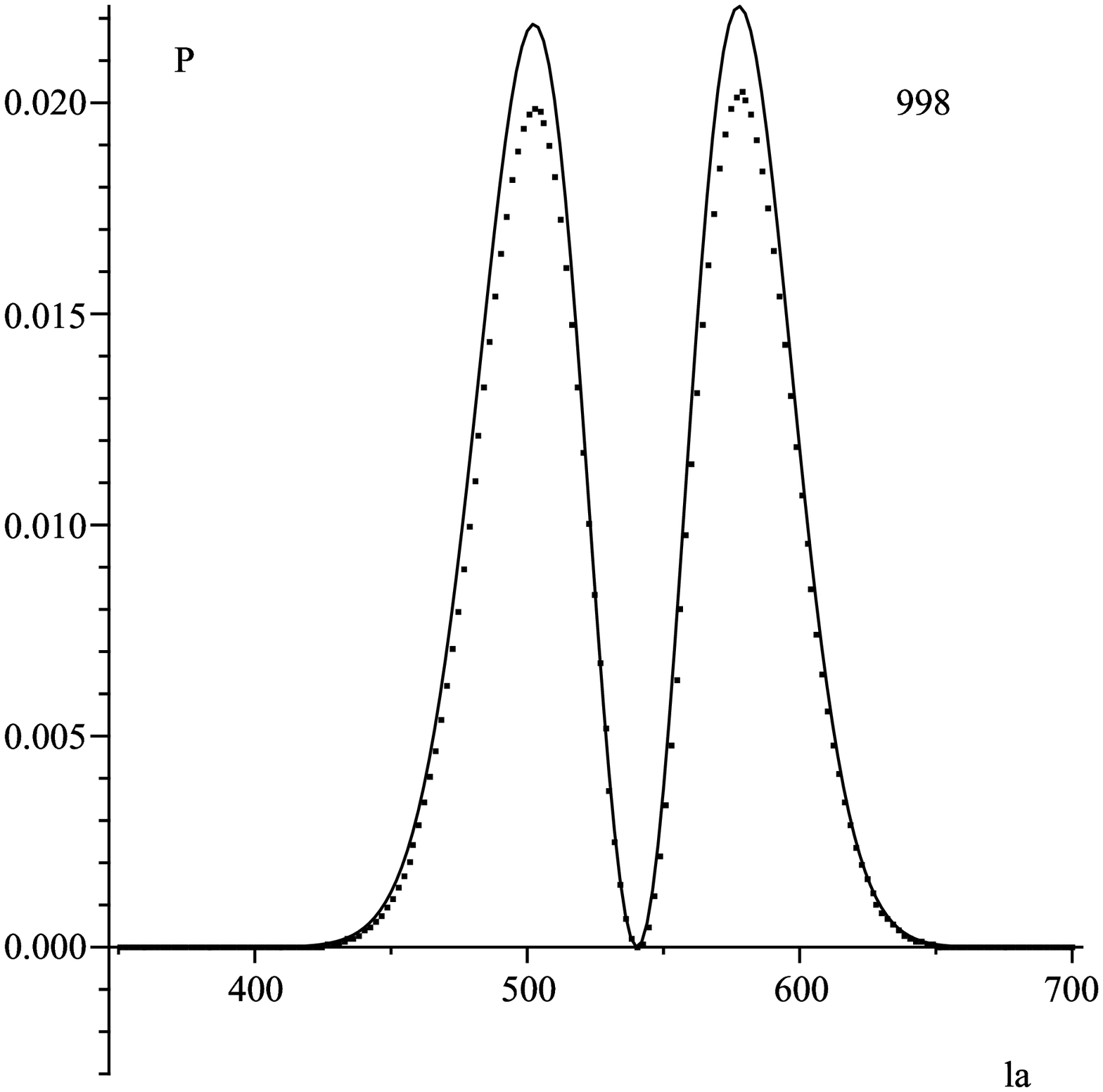}
\hspace{0.65in}
\includegraphics[width=1.6in, height=1.6in]{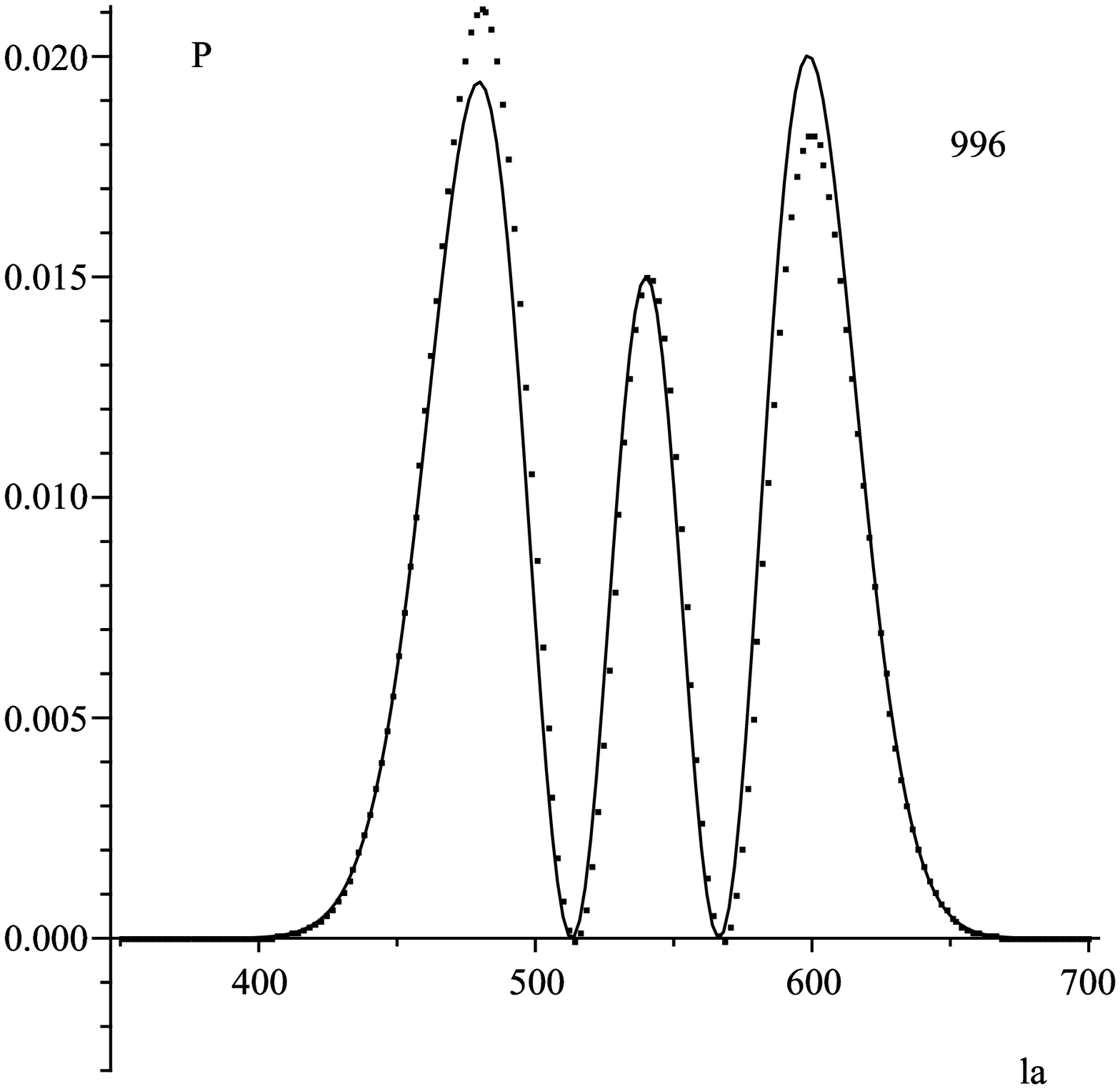}
\caption{The figures show the unperturbed (solid) and perturbed (dot) particle distributions $P$ given by eqs. (\ref{probDens}) and (\ref{pertPD}) for $N=1000$, $\theta=1$, $\delta_{\lambda} =15$ and \textit{a)} $m_0 =1000$, \textit{b)} $m_0=998$ and \textit{c)} $m_0=996$.}
\label{partDistlambda}
\end{figure}

\begin{figure}[t]
\psfrag{P}{$P$}
\psfrag{U}{$m$}
\psfrag{1000}{\textit{a)}$m_0=1000$}
\psfrag{998}{\textit{b)}$m_0=998$}
\psfrag{996}{\textit{c)}$m_0=996$}
\centering
\includegraphics[width=1.6in, height=1.6in]{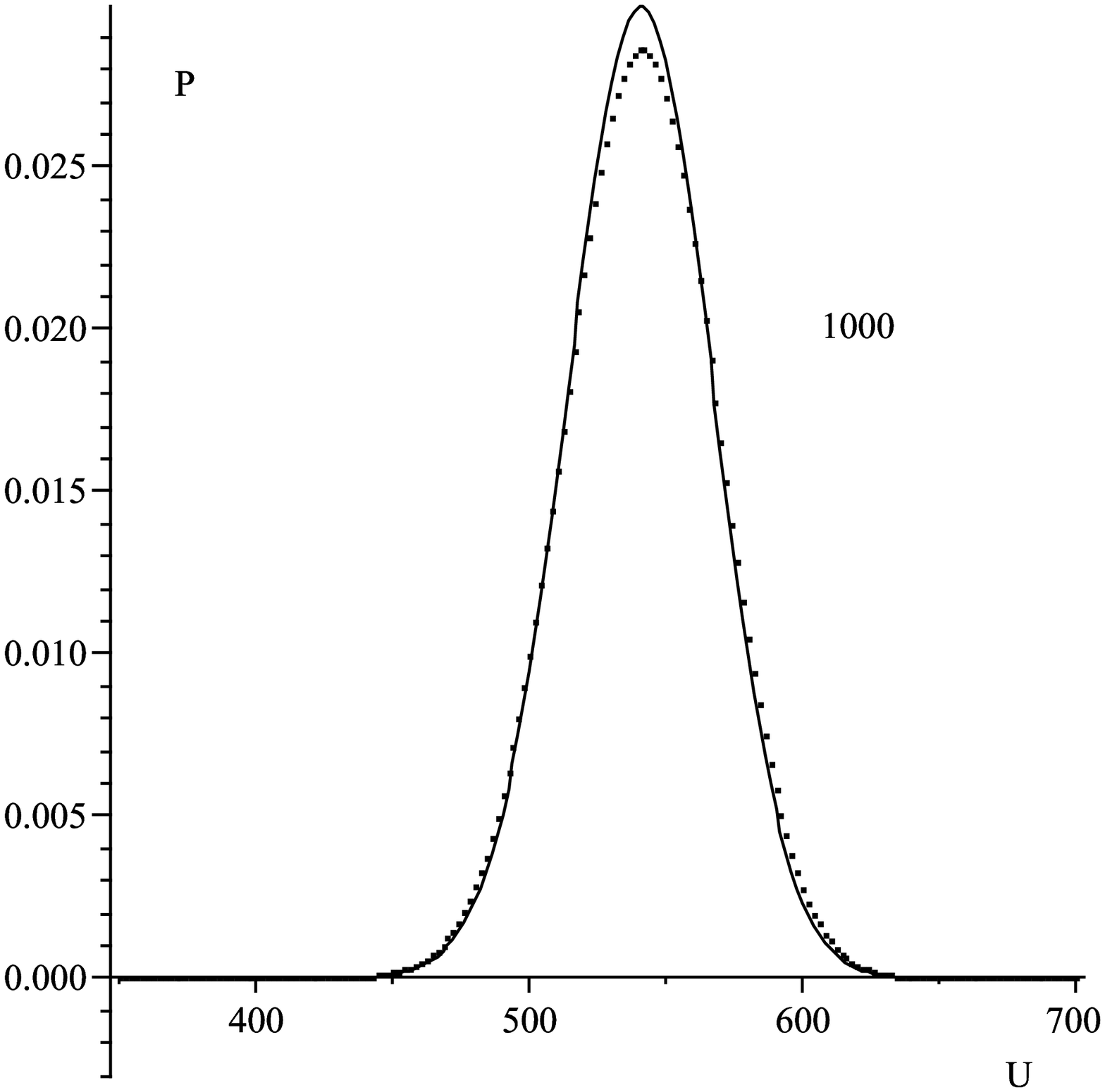}
\hspace{0.65in}
\includegraphics[width=1.6in, height=1.6in]{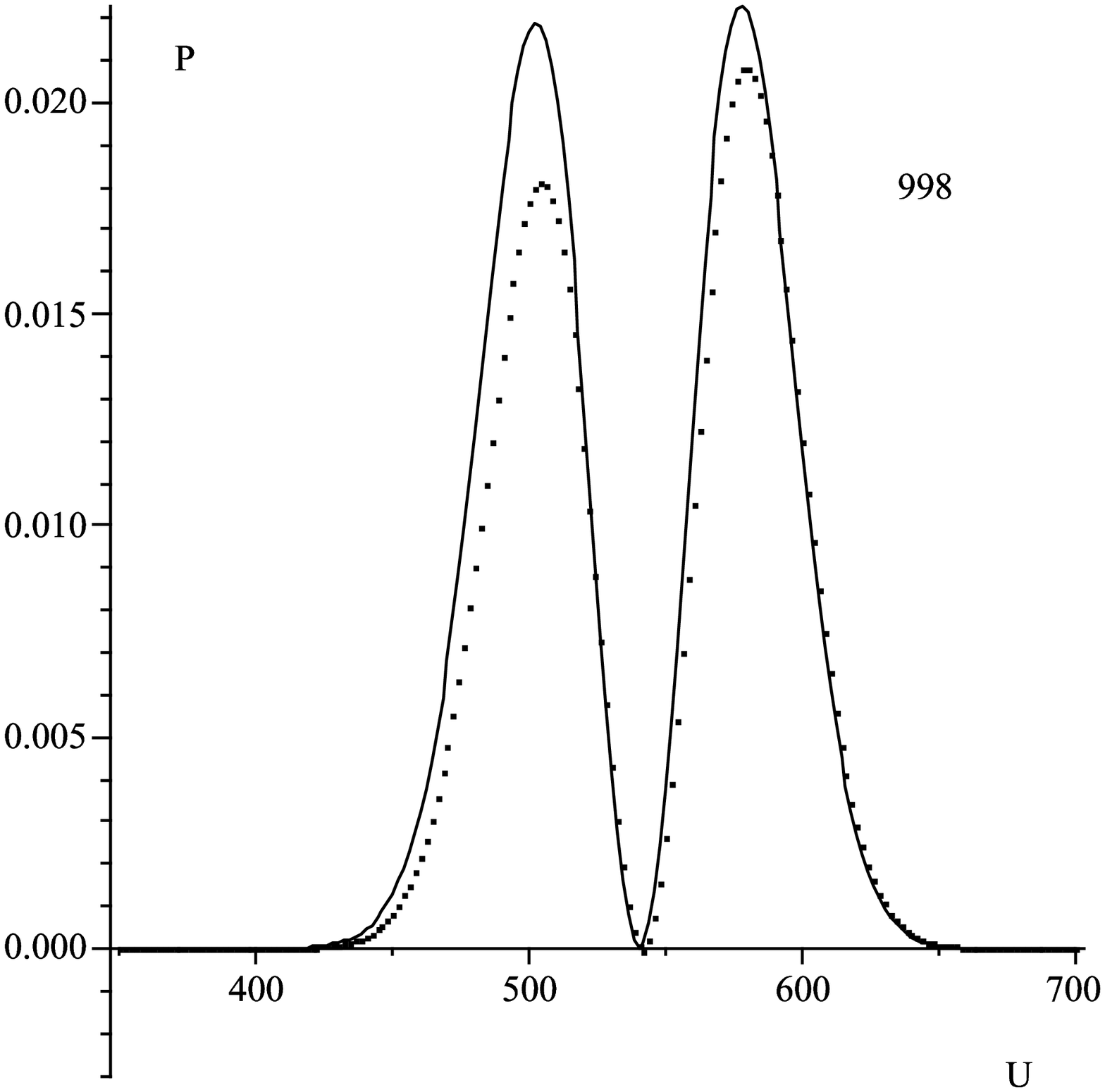}
\hspace{0.65in}
\includegraphics[width=1.6in, height=1.6in]{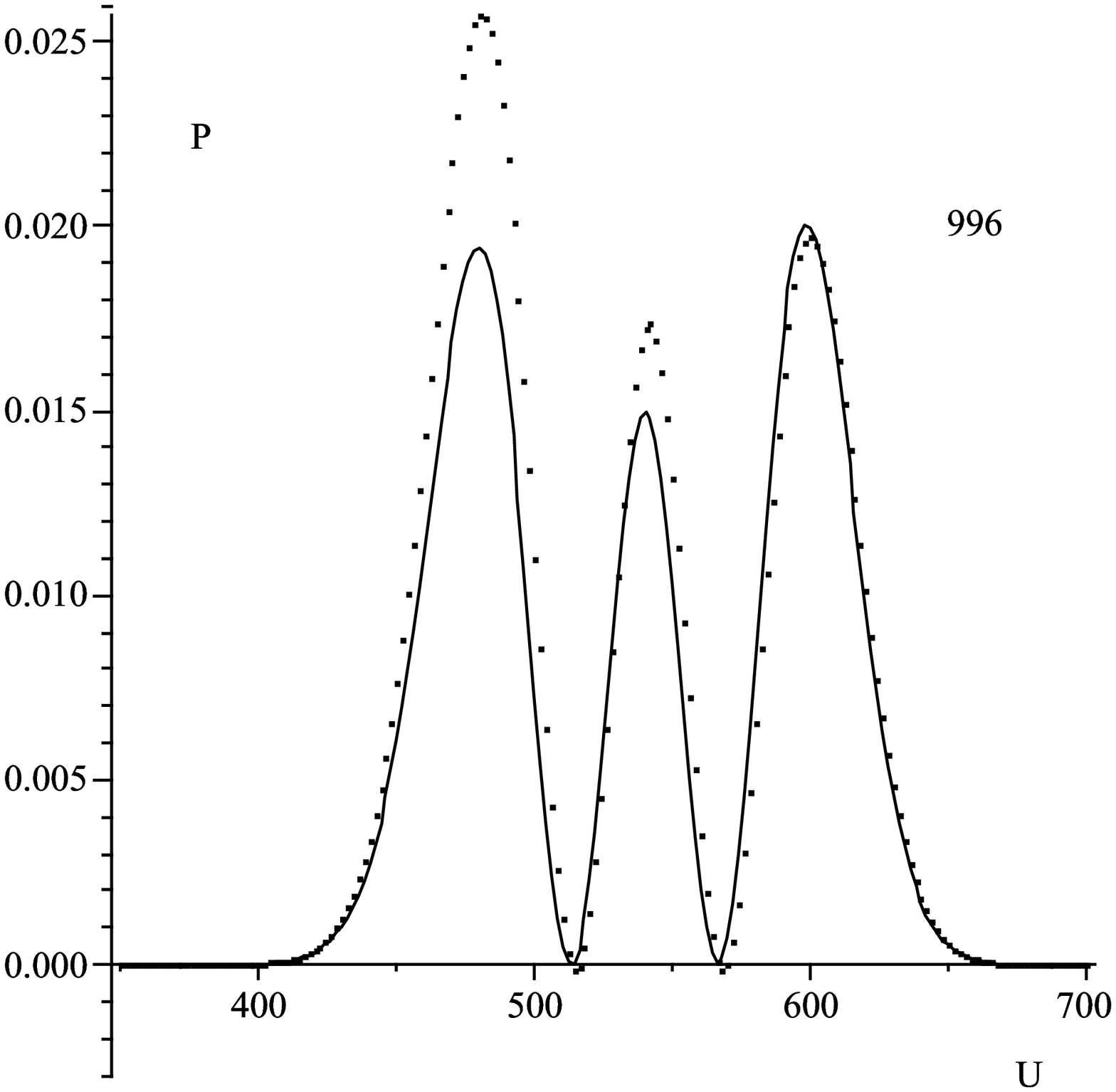}
\caption{The figures show the unperturbed (solid) and perturbed (dot) particle distributions $P$ given by eqs. (\ref{probDens}) and (\ref{pertPD}) for $N=1000$, $\theta=1$, $\delta_{\mathcal{U}} =2$ and \textit{a)} $m_0 =1000$, \textit{b)} $m_0=998$ and \textit{c)} $m_0=996$.}
\label{partDistU}
\end{figure}

\begin{figure}[t]
\psfrag{P}{$P$}
\psfrag{La}{$m$}
\psfrag{1000}{\textit{a)}$m_0=1000$}
\psfrag{998}{\textit{b)}$m_0=998$}
\psfrag{996}{\textit{c)}$m_0=996$}
\centering
\includegraphics[width=1.6in, height=1.6in]{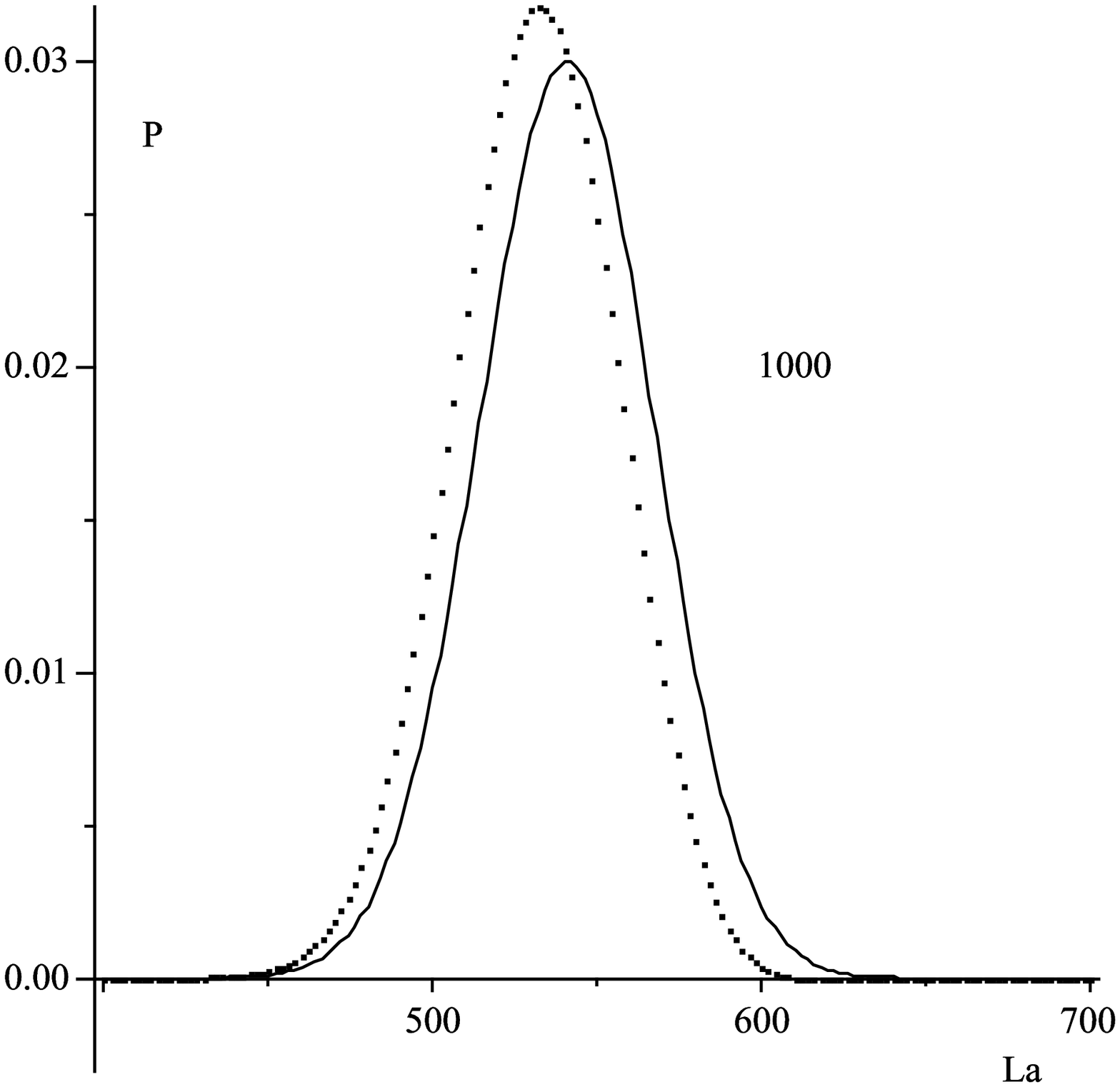}
\hspace{0.65in}
\includegraphics[width=1.6in, height=1.6in]{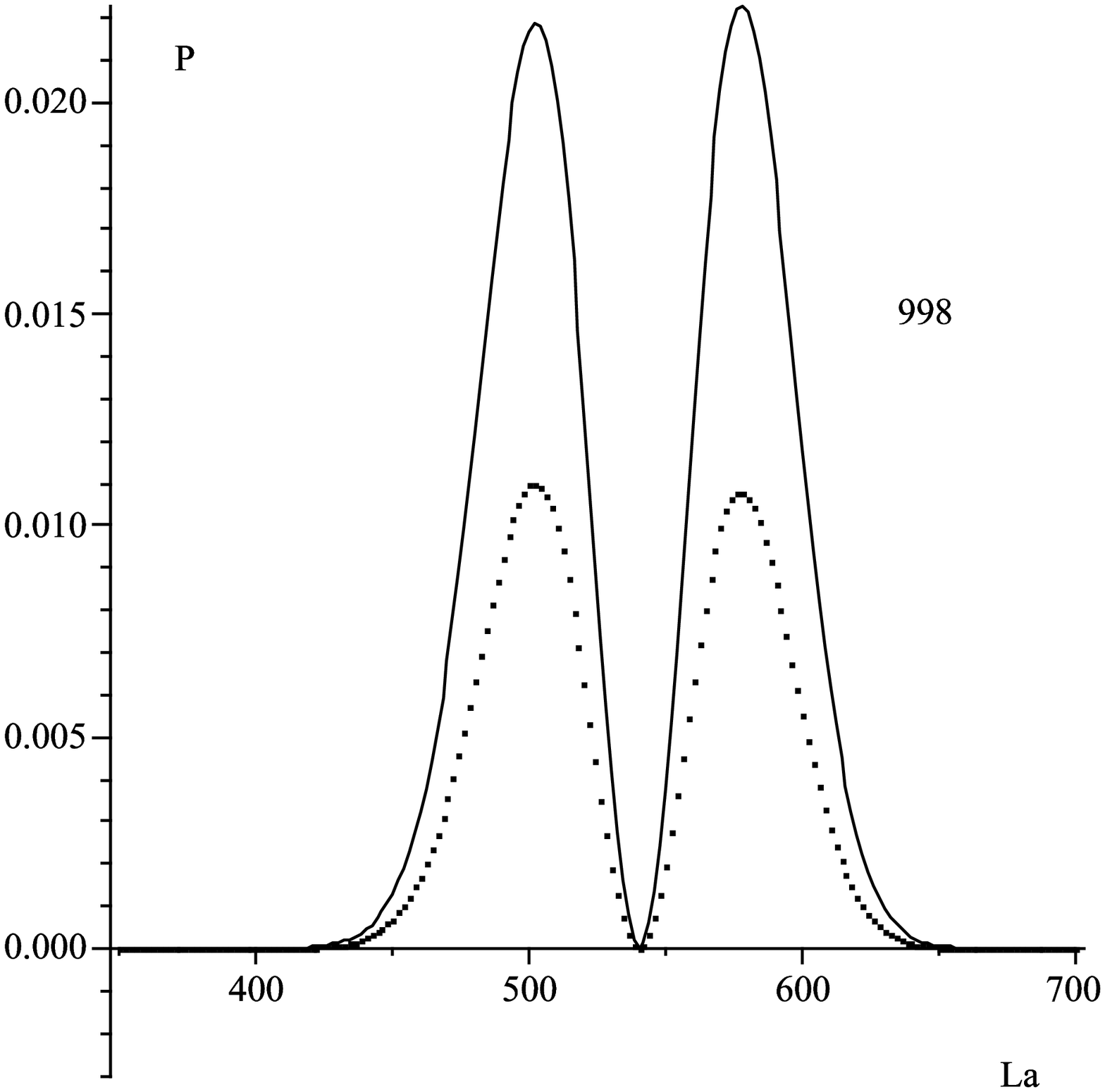}
\hspace{0.65in}
\includegraphics[width=1.6in, height=1.6in]{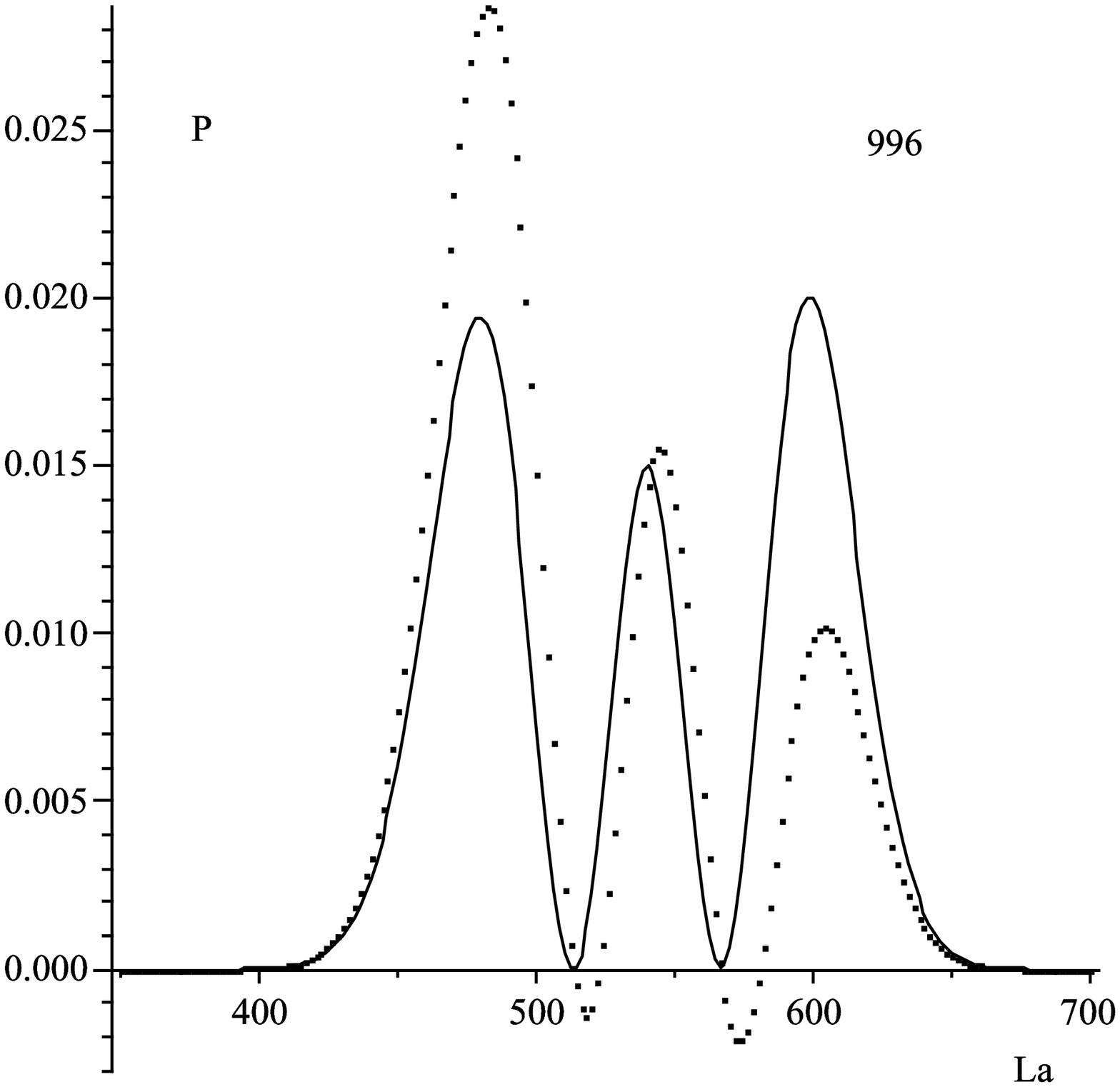}
\caption{The figures show the unperturbed (solid) and perturbed (dot) particle distributions $P$ given by eqs. (\ref{probDens}) and (\ref{pertPD}) for $N=1000$, $\theta=1$, $\delta_{\Lambda} =0.1$ and \textit{a)} $m_0 =1000$, \textit{b)} $m_0=998$ and \textit{c)} $m_0=996$.}
\label{partDistLambda}
\end{figure}

\begin{figure}[t]
\psfrag{P}{$P$}
\psfrag{mu}{$m$}
\psfrag{1000}{\textit{a)}$m_0=1000$}
\psfrag{998}{\textit{b)}$m_0=998$}
\psfrag{996}{\textit{c)}$m_0=996$}
\centering
\includegraphics[width=1.6in, height=1.6in]{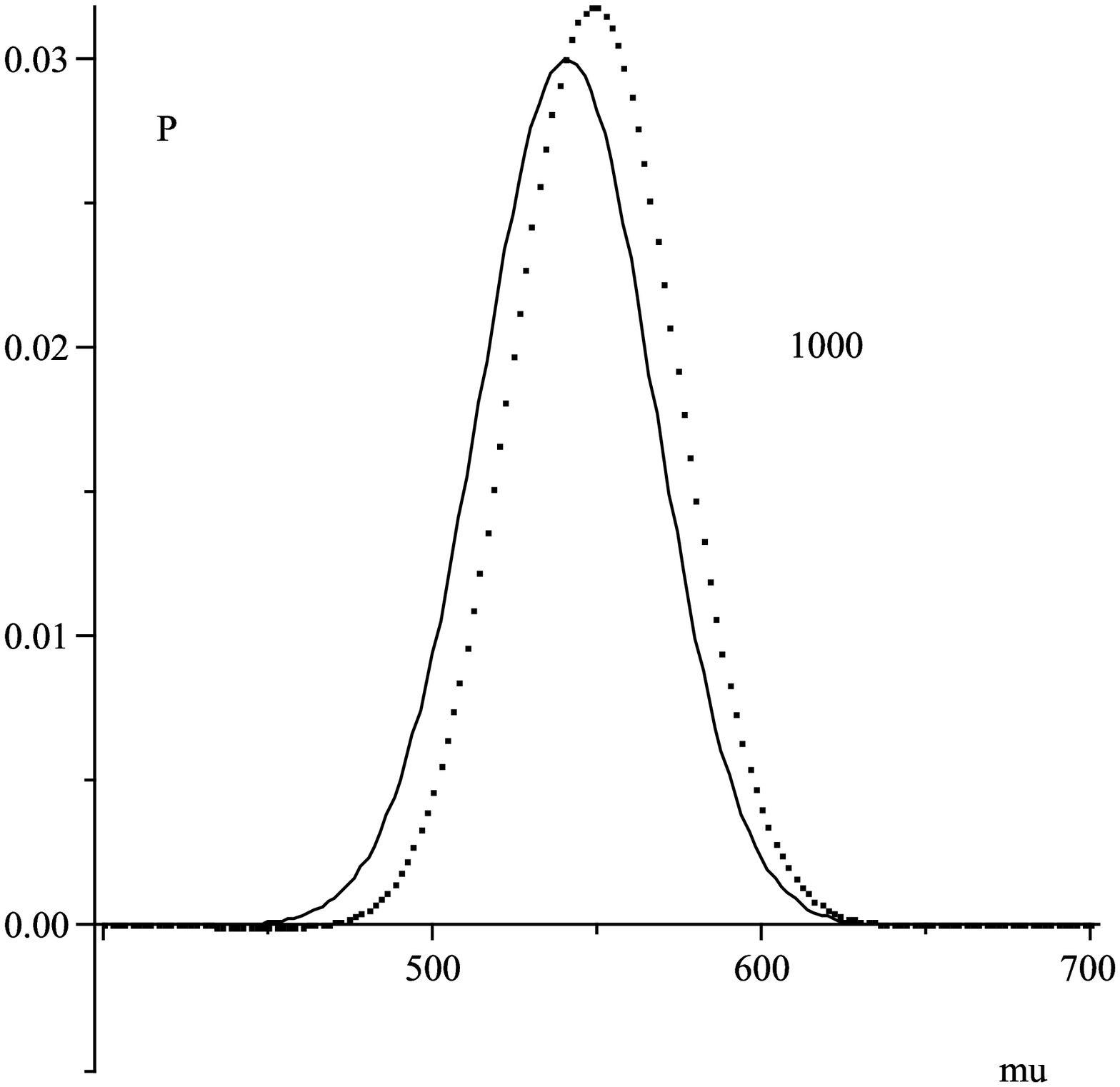}
\hspace{0.65in}
\includegraphics[width=1.6in, height=1.6in]{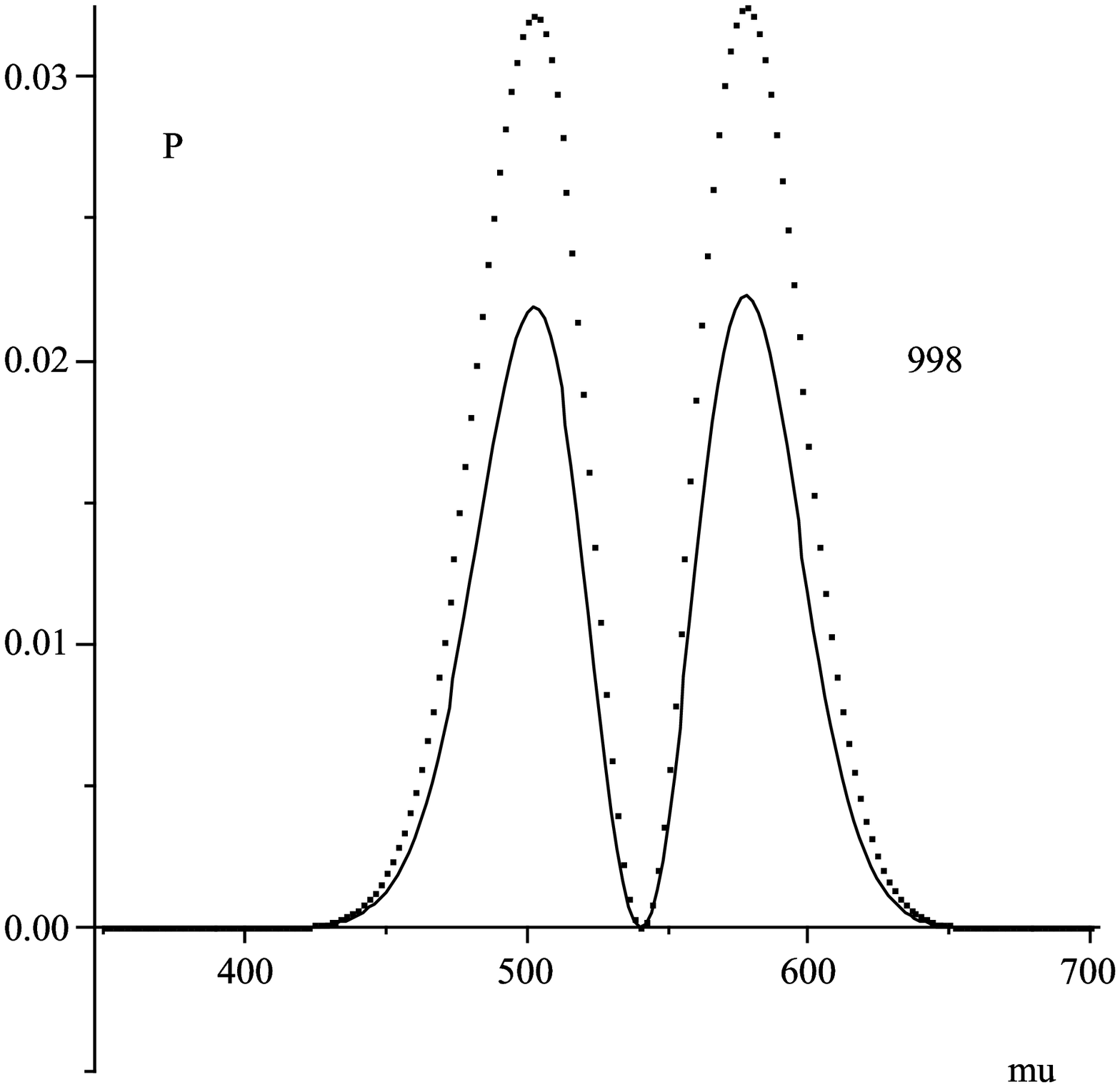}
\hspace{0.65in}
\includegraphics[width=1.6in, height=1.6in]{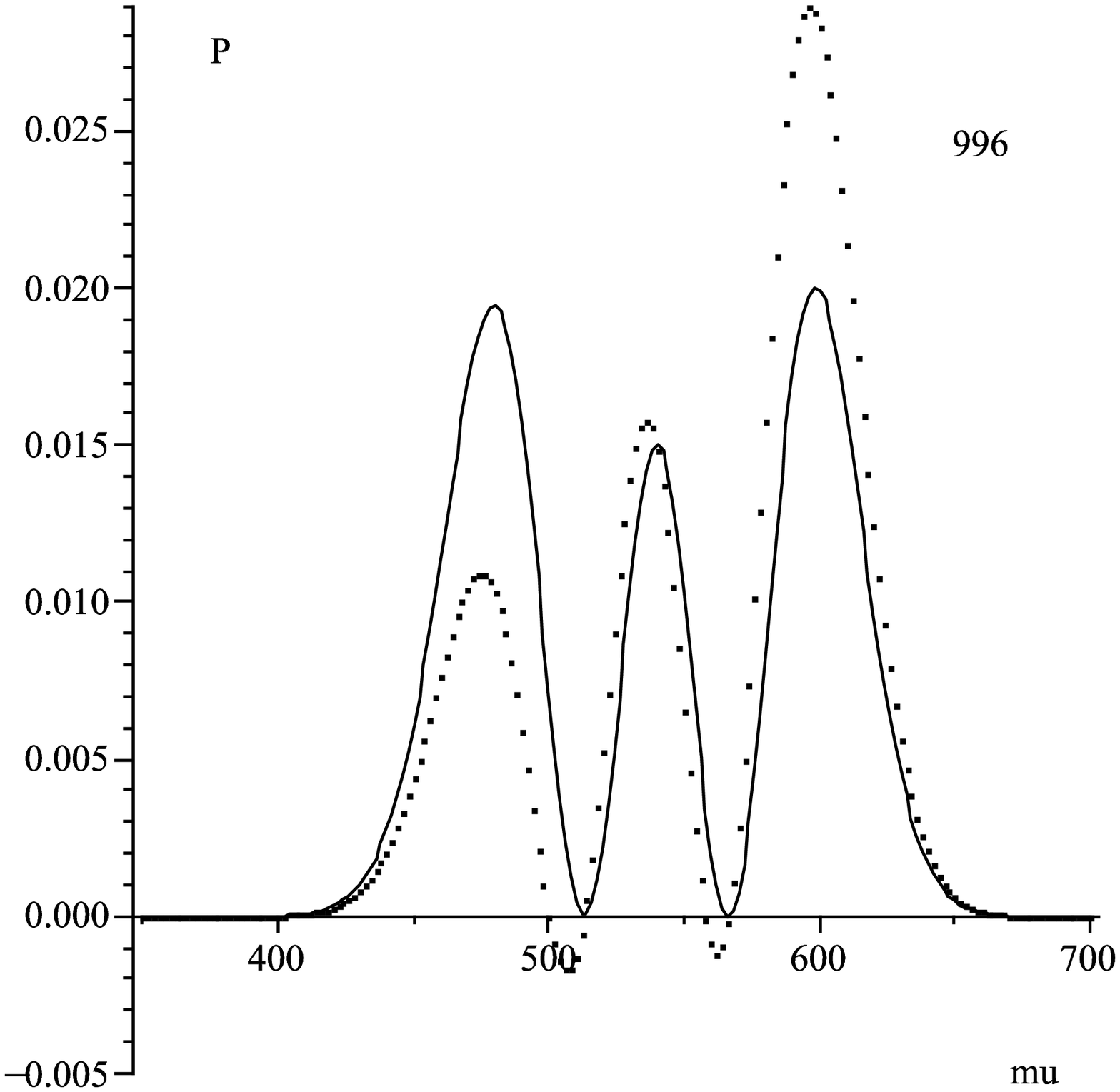}
\caption{The figures show the unperturbed (solid) and perturbed (dot) particle distributions $P$ given by eqs. (\ref{probDens}) and (\ref{pertPD}) for $N=1000$, $\theta=1$, $\delta_{\mu} =0.1$ and \textit{a)} $m_0 =1000$, \textit{b)} $m_0=998$ and \textit{c)} $m_0=996$.}
\label{partDistMu}
\end{figure}

\section{Perturbative Effects on Entanglement}

Entanglement arises in many-body quantum systems because of the superposition principle and the tensor product structure of the Hilbert space\footnote{Note that we have been using the shorthand $\vert N, m\rangle \equiv \vert n_a=\frac{N+m}{2} \rangle \otimes \vert n_b=\frac{N-m}{2}  \rangle$.}, a property of utmost importance in quantum control and quantum information.  In particular, states with high entanglement are desirable because of their utility in carrying out quantum information tasks \cite{niels}.  We would like to determine whether or not by careful choice of  the perturbations we can increase the entanglement from that of the unperturbed case.

For a bipartite quantum system the von-Neumann entropy $S=- Tr \left(\rho \log_2 \rho \right)$ is a standard measure of the entanglement of the system, $\rho$ being the reduced density matrix.  In the case at hand this reduces to
\begin{equation}
S(N,m_0, \theta)= -\sum_{m=-N}^N  P(N,m_0,m, \theta) \log_2 P(N,m_0,m, \theta)   \label{ent_Form}
\end{equation}
where $P$ is either the perturbed or unperturbed particle distribution, depending on the situation.  As noted in the section above, the first order particle distribution given by eq. (\ref{pertPD}) can be negative for certain choices of perturbations. To resolve this problem we could either replace $P^{(0+1)}$ with $\left| P^{(0+1)} \right|$ or we could include the second order term $\left| \sum_{n=-N}^N a_{m_0,n}d_{m,n}^N \right|^2$ in $P^{(0+1)}$; we use the former approach to keep the calculations strictly to first order.

When the particle distribution contains first order corrections, $P^{(0+1)}=P^{(0)}+P^{(1)}$, the entanglement is given by
\begin{equation}
S^{(0+1)}=S^{(0)}-\sum_{m=-N}^N P^{(1)}(N,m_0,m, \theta)\left( \log_2 P^{(0)}(N,m_0,m, \theta) + \frac{1}{\ln 2} \right).   \label{pertEnt}
\end{equation}
For the case at hand we read off that $P^{(1)}=2d_{m,m_0}^N \sum_{n=-N}^N  \hbox{Re}\left( a_{m_0,n} \right) d_{m,n}^N $.  Hence we have a criterion for increasing the entanglement of the system.  Namely, we increase the entanglement (to first order) precisely when
\begin{equation}
\sum_{m=-N}^N P^{(1)}(N,m_0,m, \theta)\left( \log_2 P^{(0)}(N,m_0,m, \theta) + \frac{1}{\ln 2} \right) <0.	 \label{incEnt}
\end{equation}
Note that $P^{(1)}(N,m_0,m, \theta)$ is proportional to the perturbation strengths, so that we can increase or decrease the entanglement by choosing the sign of the perturbation appropriately.  Continuing, we note that for each $m$ we have $\log_2 P^{(0)}(N,m_0,m, \theta) + \frac{1}{\ln 2} >0$, so that to maximize the entanglement  we should maximize each $P^{(1)}$.  Of course, $\vert P^{(1)} \vert $ can be made arbitrarily large simply by choosing $\delta$ arbitrary large.  However, we are limited in such a choice since we must keep $\delta$ small so that perturbation theory can be trusted.

\begin{figure}[t]
\centering
\includegraphics[width=2in, height=2in]{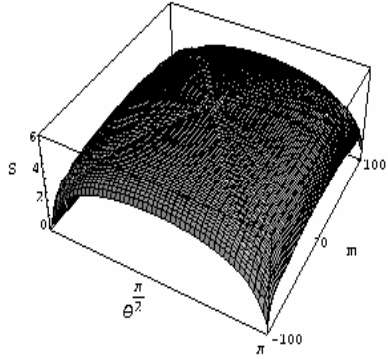}
\caption{The figure shows the unperturbed entanglement using eq. (\ref{ent_Form}) for $N=100$.}
\label{ent_Unpert}
\end{figure}

We plot in Figure \ref{ent_Unpert} the unperturbed entanglement as a function of $\theta$ and $m_0$.  See \cite{fuent2} for an extensive analysis. In Figures \ref{ent_omega}-\ref{ent_mu} we plot the perturbed entanglements $S^{(0+1)}$ for $\delta=0.01$ as well as the differences $\Delta S := S^{(0+1)}-S^{(0)}$ between the perturbed and unperturbed entanglements for $\delta=0.1$; we choose a different $\delta$ in the latter case so that the differences are more evident.   In each of the plots of $\Delta S$ we see that the largest changes occur approximately along the diagonal lines, connecting the points $(\theta,m_0)=(\pi,-N)$ and $(\pi,N)$ and the points $(-\pi,-N)$ and $(\pi,N)$; these lines in $\theta-m_0$ space correspond to certain strengths of the coupling constants in the Hamiltonian $H_2$  viewed as functions of $m_0$.  Note this type of behaviour also occurs along these lines in the unperturbed entanglement plots, as shown in Figure \ref{ent_Unpert}.  Comparing the magnitude of the perturbative effects we observe that perturbations to the mode-exchange collision terms ($\mathcal{U},\Lambda, \mu$) have a much greater effect, their maximum difference being about an order of magnitude larger than those for $\lambda$ and $\omega$. It is also interesting to observe that the coherent states, which correspond to $U^\dagger|N,N\rangle$ and $U^\dagger|N,-N\rangle$ are the states which are maximally affected by perturbations. In such states $A_{1}$ is much larger than $A_{2}$ so that the rate of collisions is relatively small.  To further study this we have plotted in Figure \ref{Entan_Cut_Plots} a two-dimensional cut of Figures  \ref{ent_omega}\textit{b)}-\ref{ent_mu}\textit{b)}, where we have fixed $\theta=\frac{\pi}{4}$ and allowed $m_0$ to vary.  We observe that in each of the plots $\Delta S$ has both a local maximum and local minimum as $m_0$ approaches $N$.

We also observe that the most significant changes occur in the region $m_0 >0$, where more particles lie in the $a$ mode. Positive $m_0$ implies that, for fixed $A_1 > 0$, the scattering length for same-mode collisions is positive and as $m_0$ approaches $N$, the collision rate becomes smaller. We then conclude that condensates with negative scattering lengths are more resilient to parameter perturbations and high same-mode collision rates help stabilize the condensate.

\begin{figure}[t]
\centering
\includegraphics[width=2in, height=2in]{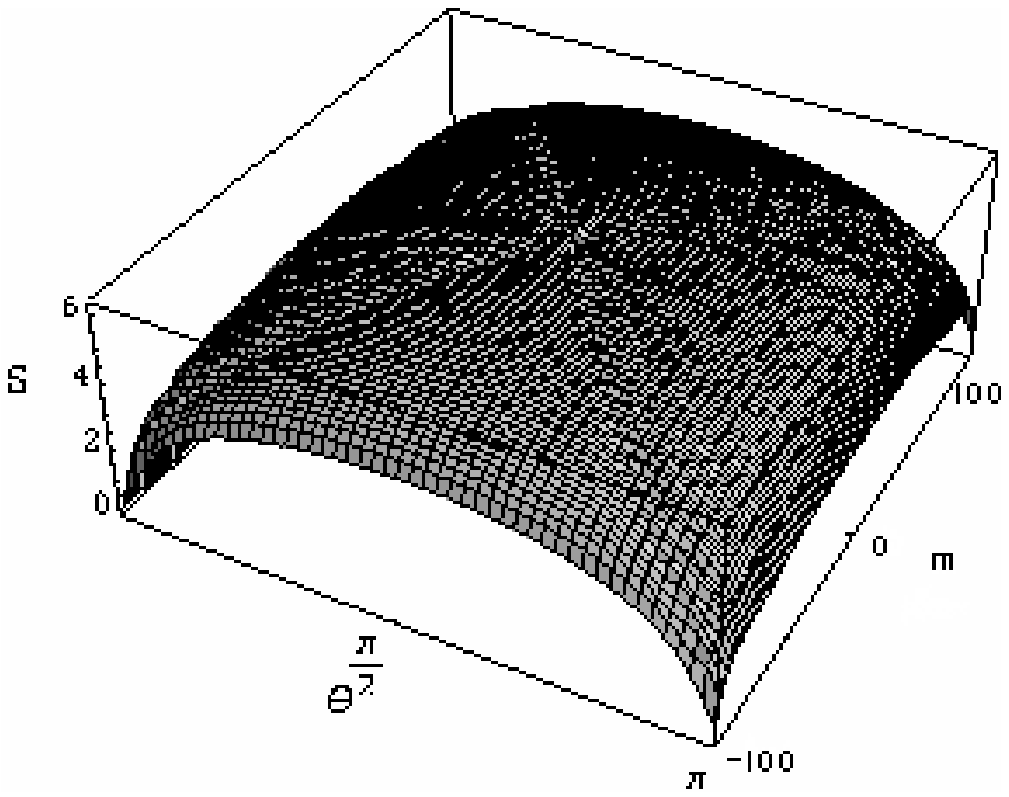}
\hspace{0.4in}
\includegraphics[width=2in, height=2in]{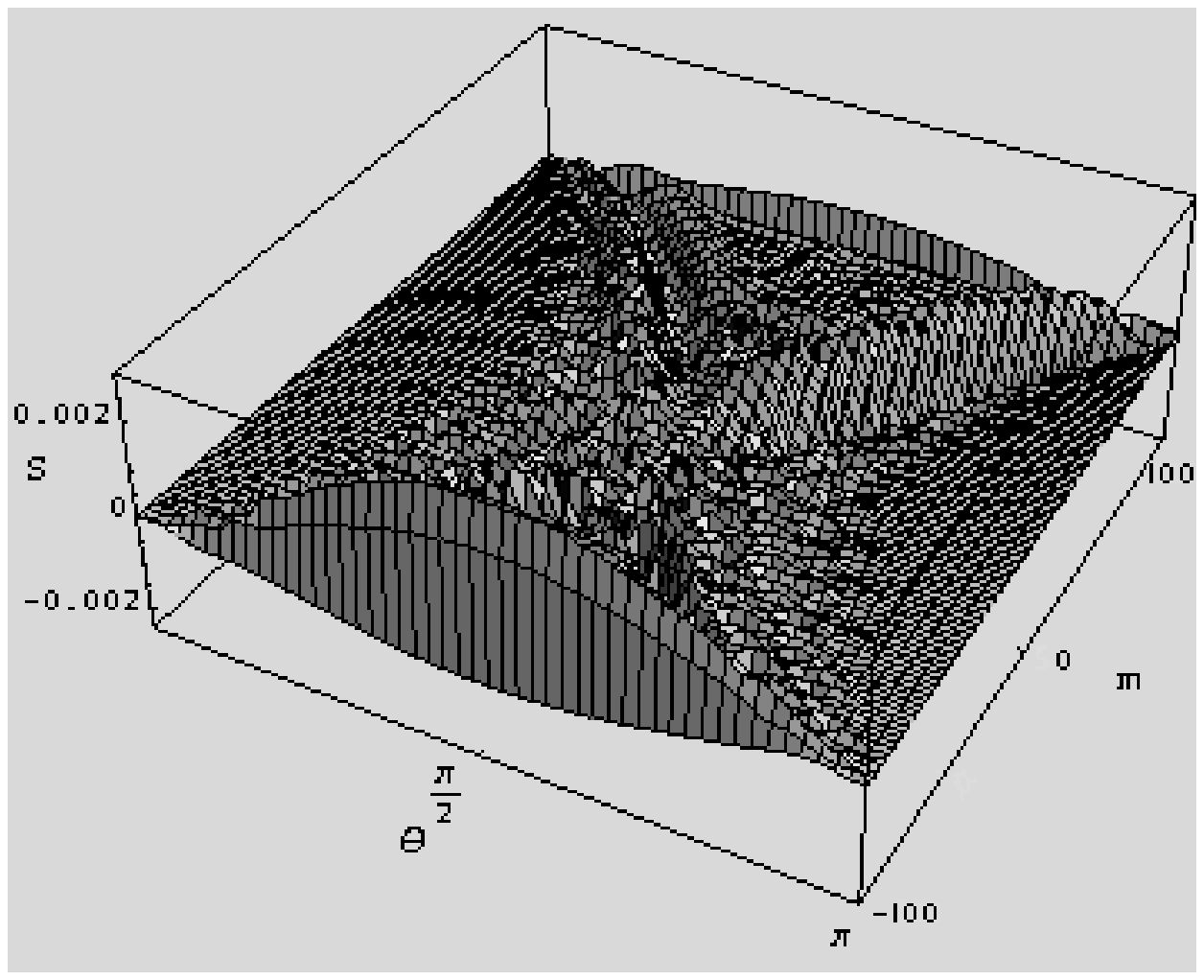}
\caption{The figures show the \textit{a)} perturbed entanglement for $\delta_{\omega}=0.01$ as a function of $m_0$ and $\theta$ for $N=50$ and \textit{b)} the difference $S^{(0+1)}-S^{(0)}$ for $N=100$ and $\delta_{\omega}=0.1$.}
\label{ent_omega}
\end{figure}

\begin{figure}[t]
\centering
\includegraphics[width=2in, height=2in]{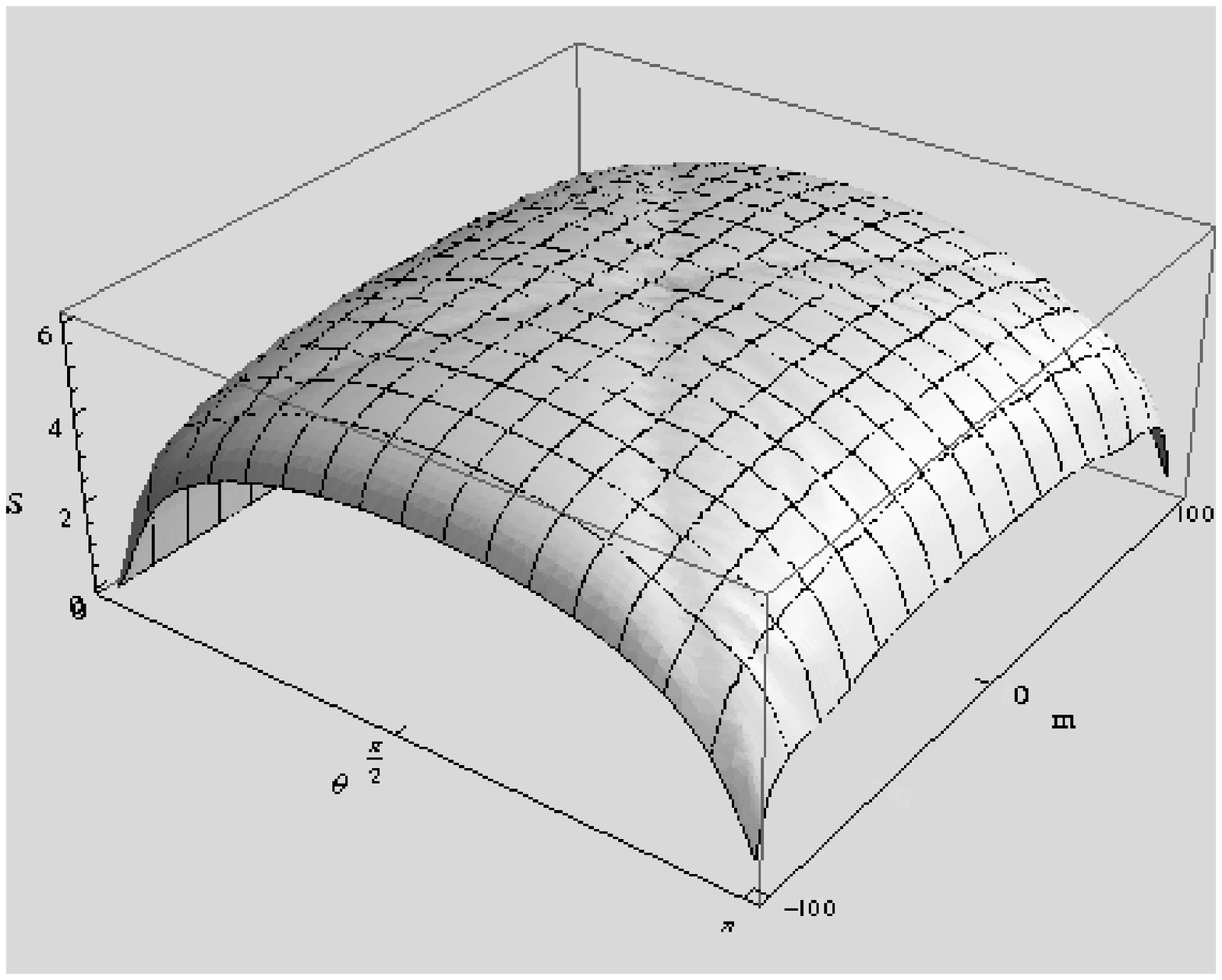}
\hspace{0.4in}
\includegraphics[width=2in, height=2in]{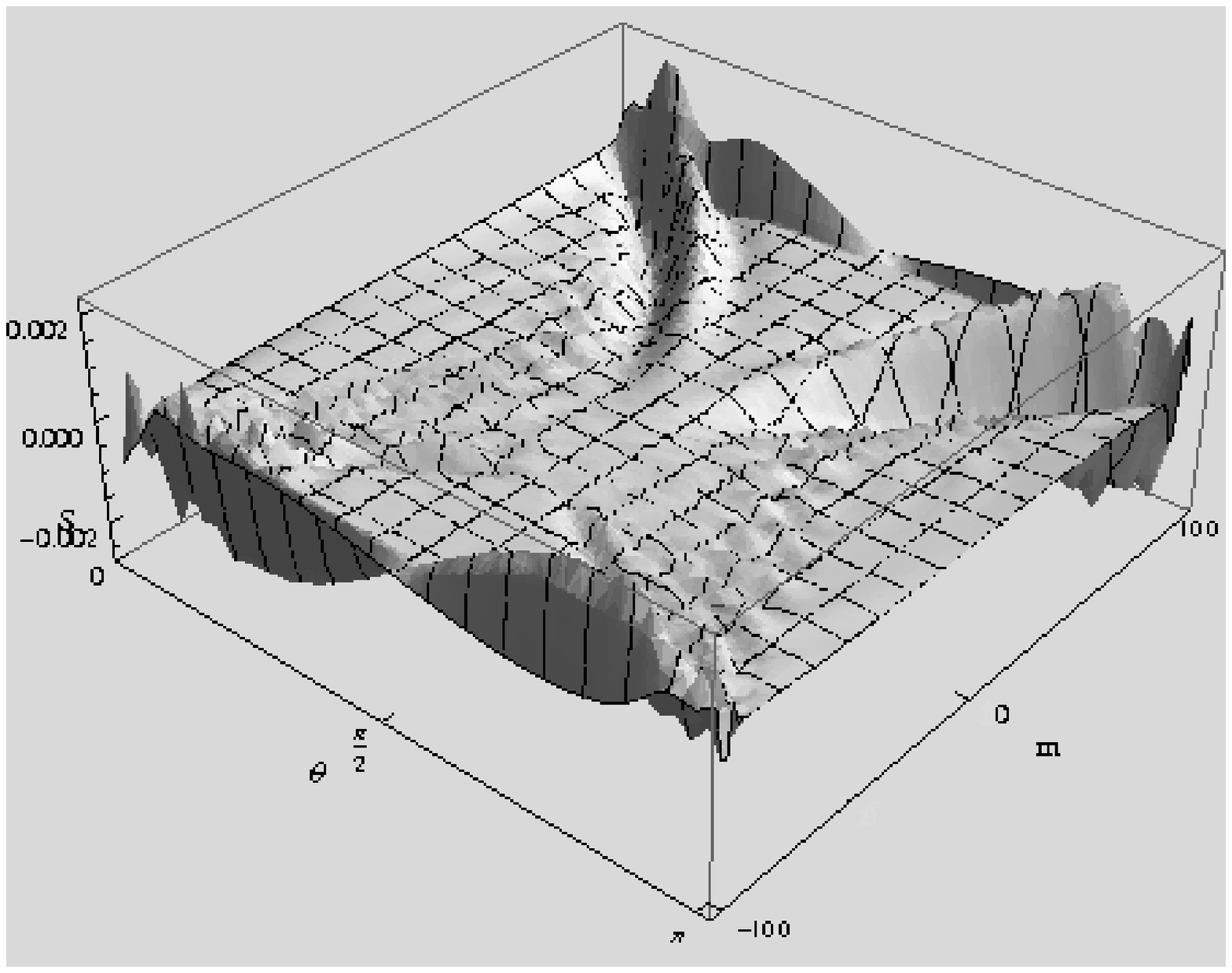}
\caption{The figures show the \textit{a)} perturbed entanglement for $\delta_{\lambda}=0.01$ as a function of $m_0$ and $\theta$ for $N=100$ and \textit{b)} the difference $S^{(0+1)}-S^{(0)}$ for $N=100$ and $\delta_{\lambda}=0.1$.}
\label{ent_small_lambda}
\end{figure}

\begin{figure}[t]
\centering
\includegraphics[width=2in, height=2in]{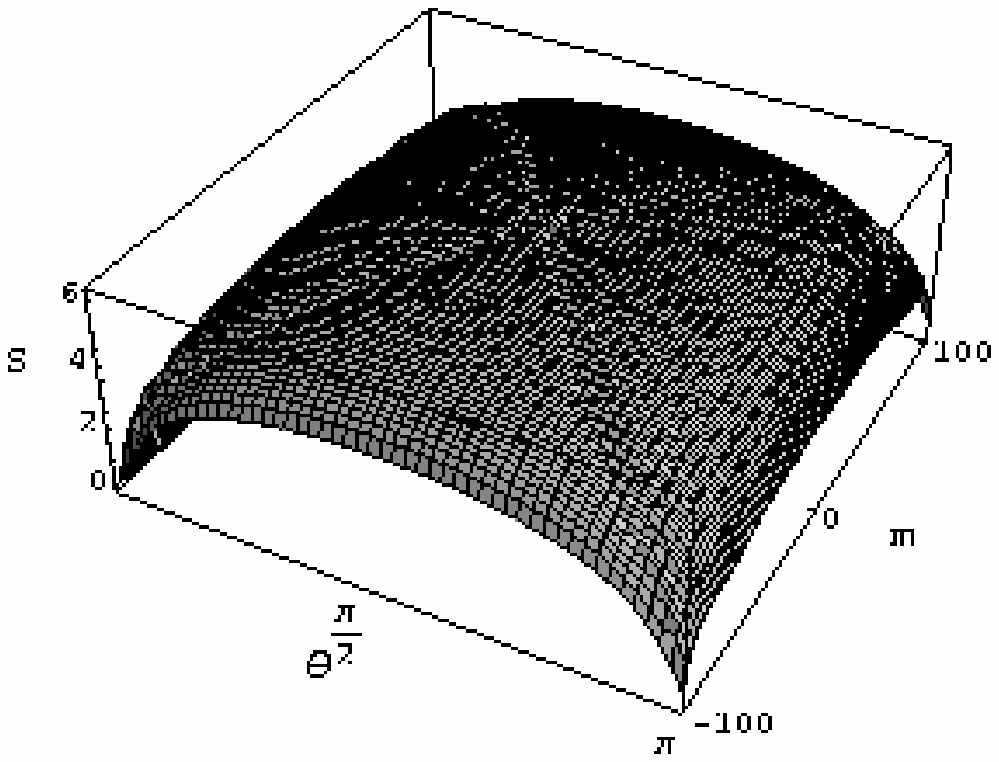}
\hspace{0.4in}
\includegraphics[width=2in, height=2in]{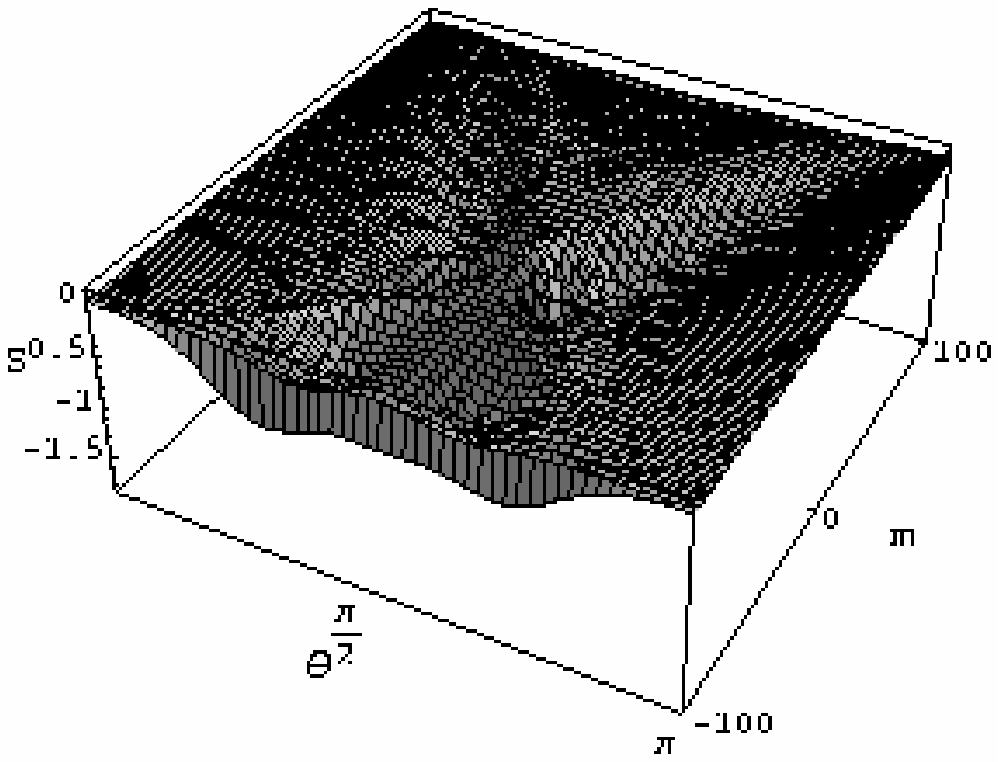}
\caption{The figures show the \textit{a)} perturbed entanglement for $\delta_{\mathcal{U}}=0.01$ as a function of $m_0$ and $\theta$ for $N=100$ and \textit{b)} the difference $S^{(0+1)}-S^{(0)}$ for $N=100$ and $\delta_{\mathcal{U}}=0.1$.}
\label{ent_U}
\end{figure}

\begin{figure}[t]
\centering
\includegraphics[width=2in, height=2in]{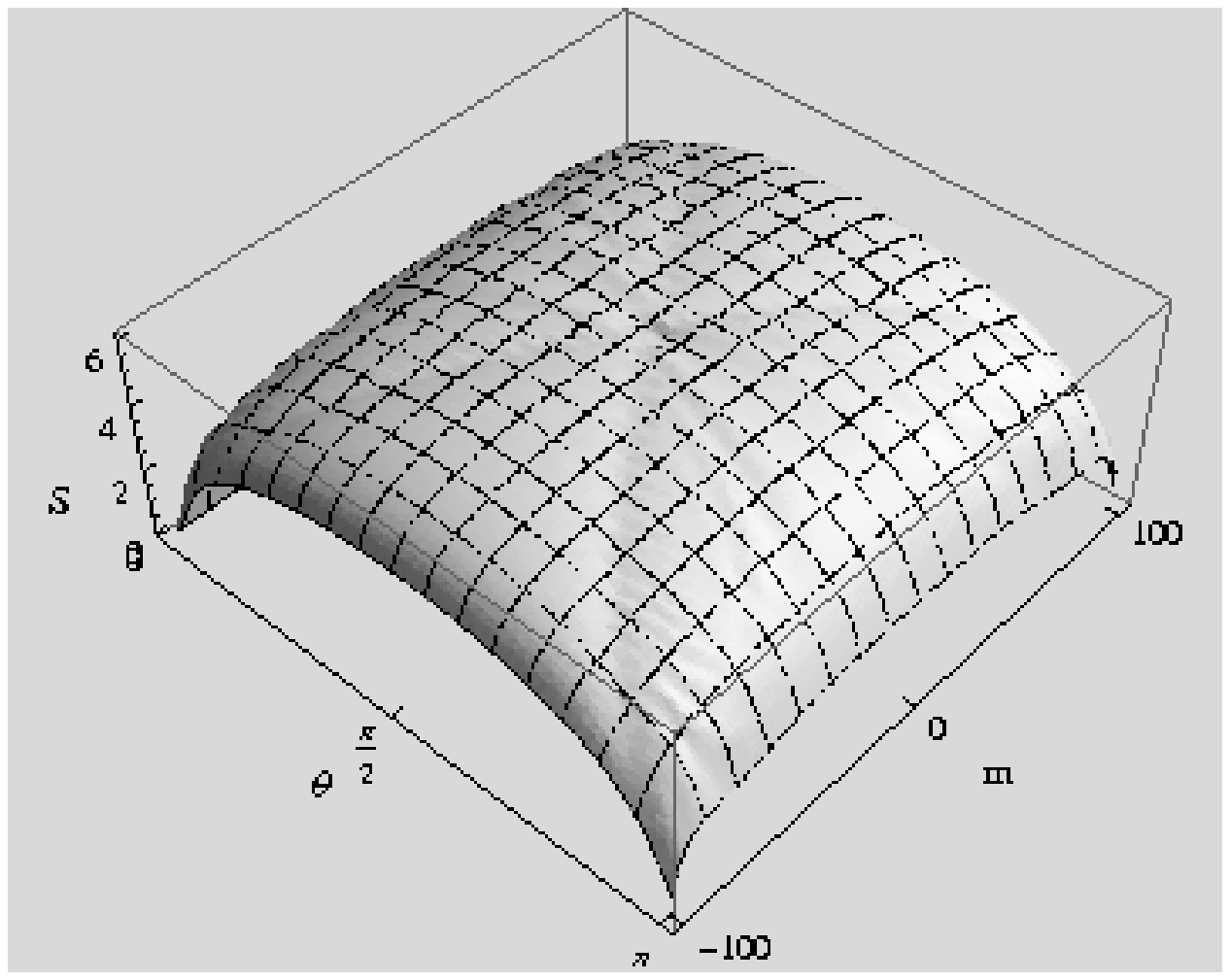}
\hspace{0.4in}
\includegraphics[width=2in, height=2in]{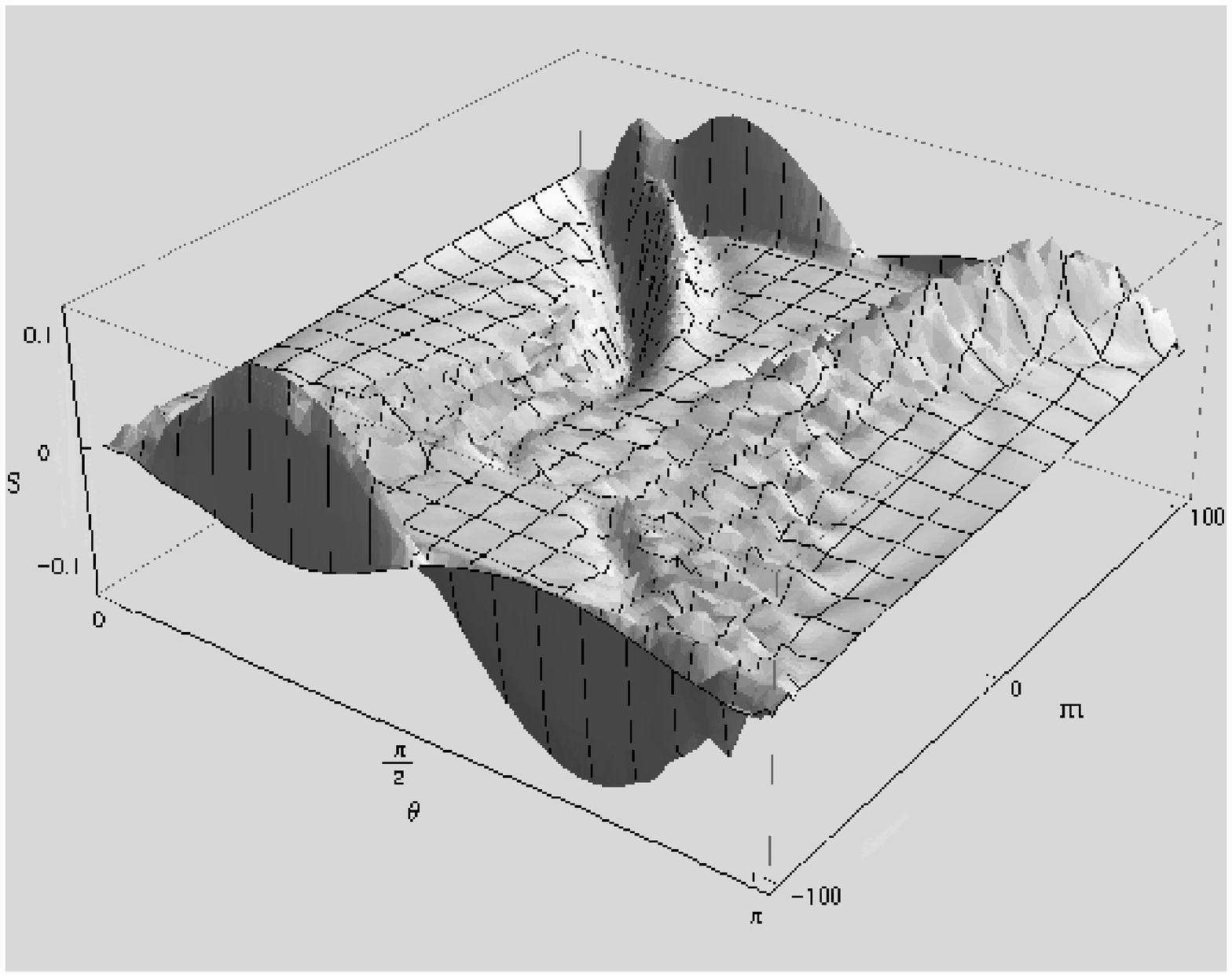}
\hspace{0.4in}
\caption{The figures show the \textit{a)} perturbed entanglement for $\delta_{\Lambda}=0.01$ as a function of $m_0$ and $\theta$ for $N=100$ and \textit{b)} the difference $S^{(0+1)}-S^{(0)}$ for $N=100$ and $\delta_{\Lambda}=0.1$.}
\label{ent_Lambda}
\end{figure}

\begin{figure}[t]
\centering
\includegraphics[width=2in, height=2in]{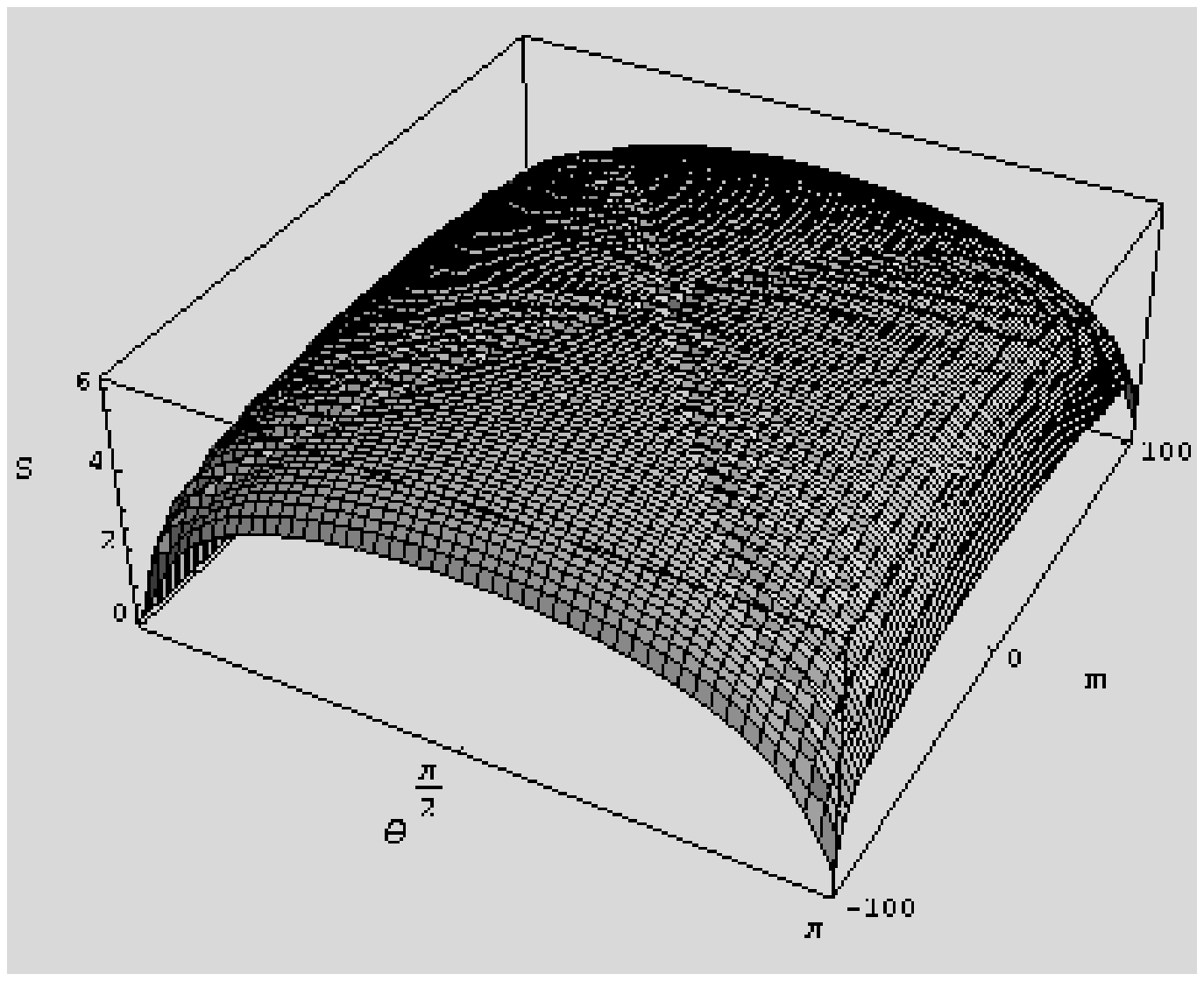}
\hspace{0.4in}
\includegraphics[width=2in, height=2in]{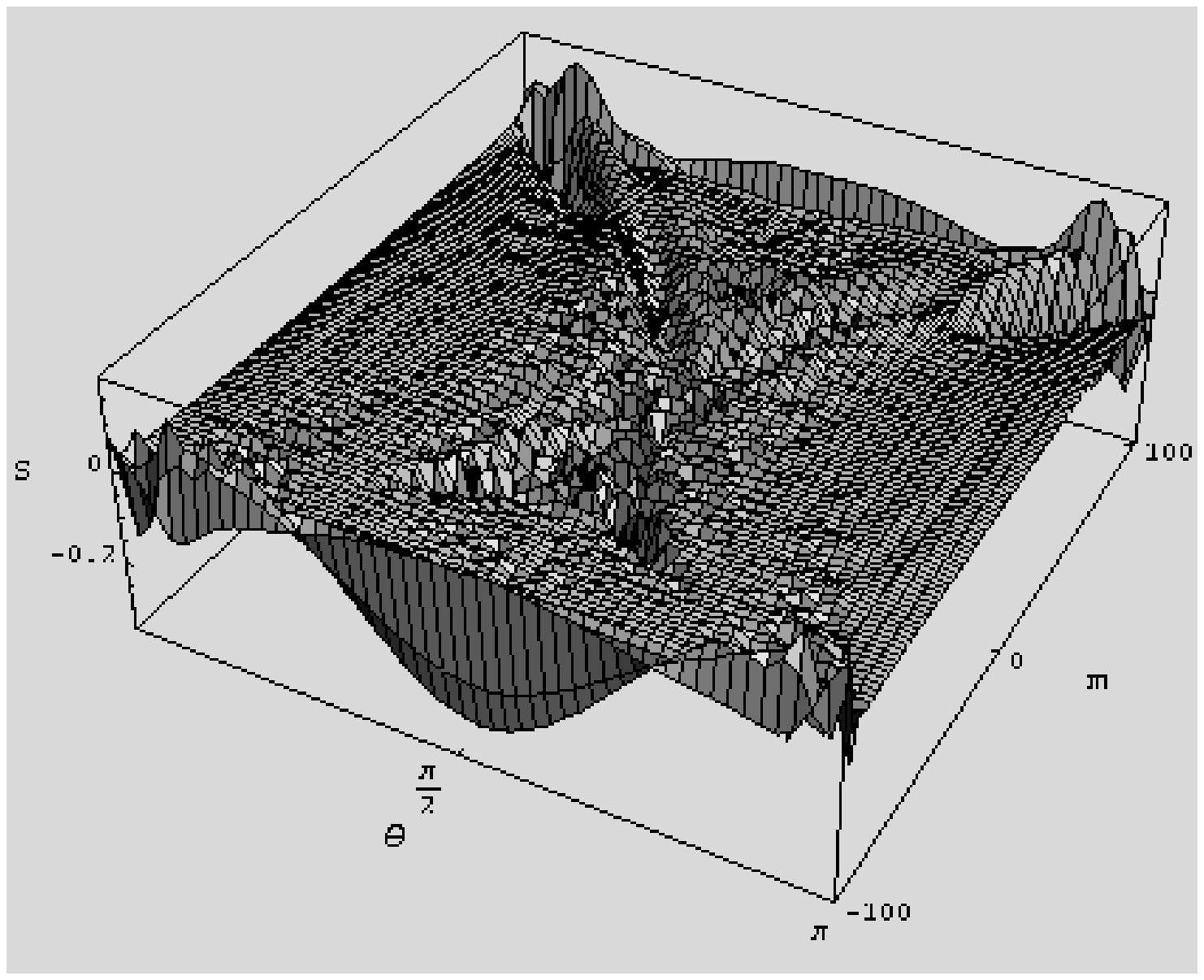}
\hspace{0.4in}
\caption{The figures show the \textit{a)} perturbed entanglement for $\delta_{\mu}=0.01$ as a function of $m_0$ and $\theta$ for $N=100$ and \textit{b)} the difference $S^{(0+1)}-S^{(0)}$ for $N=100$ and $\delta_{\mu}=0.1$.}
\label{ent_mu}
\end{figure}

\begin{figure}[t]
\psfrag{a}{\textit{a})}
\psfrag{b}{\textit{b})}
\psfrag{c}{\textit{c})}
\psfrag{d}{\textit{d})}
\psfrag{e}{\textit{e})}
\centering
\includegraphics[width=1.2in, height=1.1in]{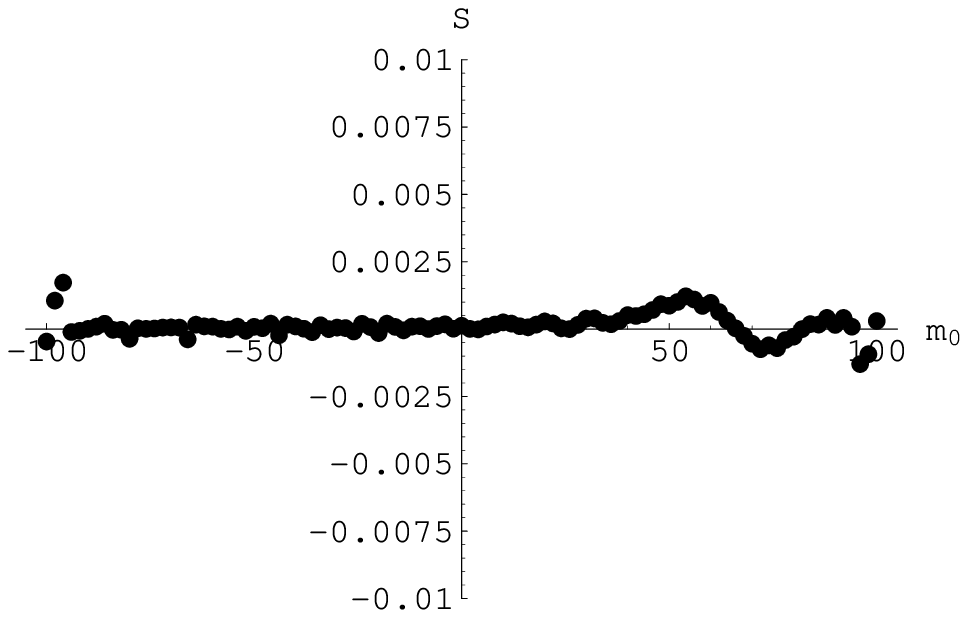}
\hspace{0.1in}
\includegraphics[width=1.2in, height=1.1in]{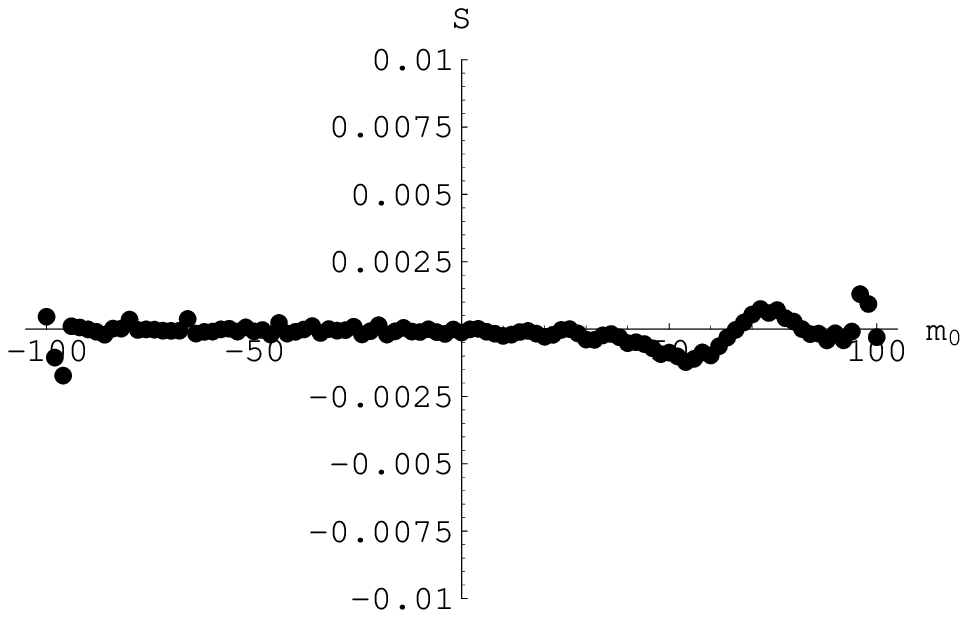}
\hspace{0.1in}
\includegraphics[width=1.2in, height=1.1in]{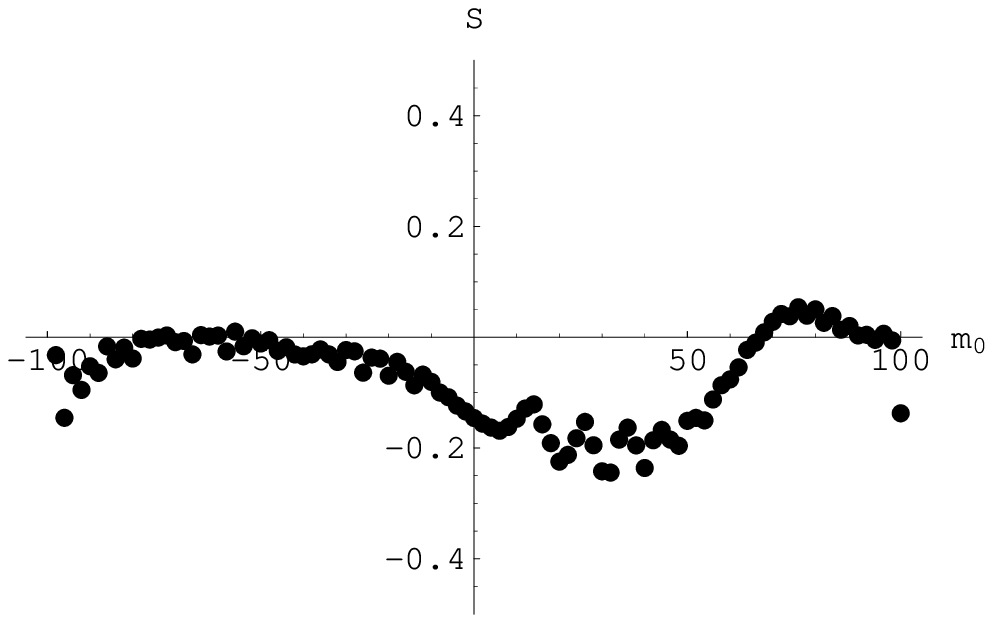}
\hspace{0.1in}
\includegraphics[width=1.2in, height=1.1in]{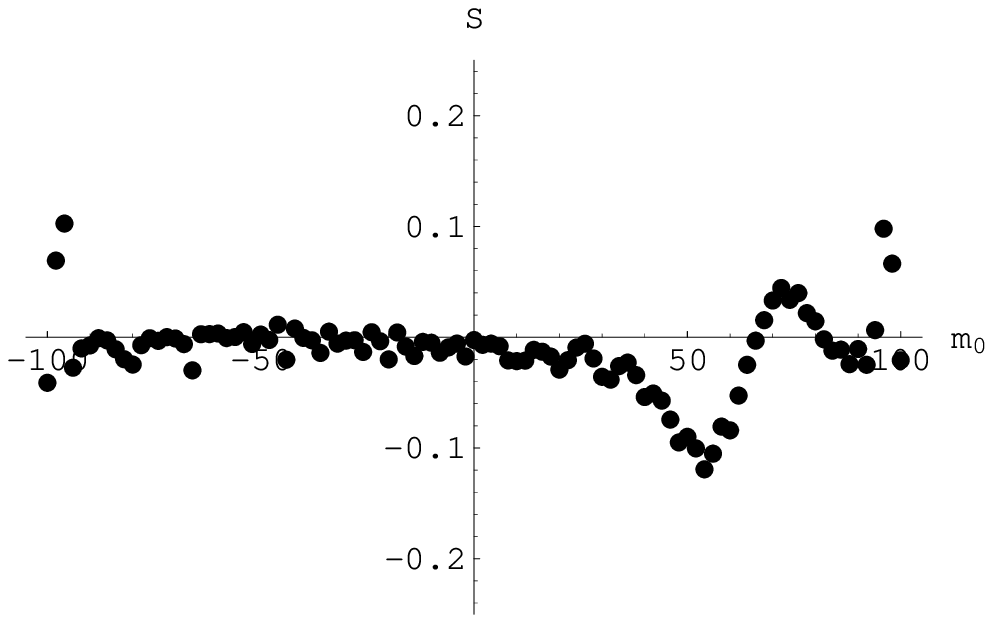}
\hspace{0.1in}
\includegraphics[width=1.2in, height=1.1in]{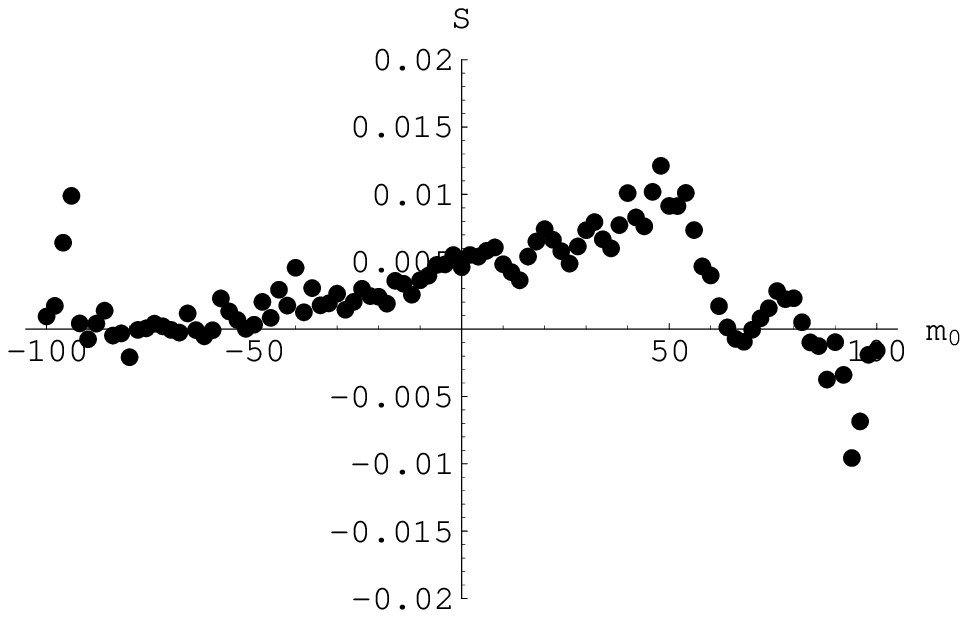}
\caption{Shown in  \textit{a)} - \textit{e)} are plots of $S^{(0+1)}-S^{(0)}$  for  $\omega, \lambda, \mathcal{U}, \Lambda$ and $\mu$, respectively, for $\theta=\frac{\pi}{4}$ as a function of $m_0$.  The perturbation strengths used are all $\delta=0.1$.}
\label{Entan_Cut_Plots}
\end{figure}

\section{Evolution of Relative Population}

With an analytic solution to the system we may study the evolution of the relative population $\langle \hat{m} \rangle (t) := \langle \psi (t) \vert \hat{m} \vert \psi (t) \rangle$ as a function of time, where the initial state is given by $\vert \psi (t=0) \rangle = \sum_{m=-N}^N C_m U^{\dagger}\vert N,m \rangle$.   For simplicity we restrict our study to the case in which $C_m \in \mathbb{R}$.  To first order in the perturbation parameter we have
\[
\vert \psi (t) \rangle^{(0+1)} = \sum_{m=-N}^N C_m e^{-iE_m^{(0+1)} t}U^{\dagger}\left\{ \vert N,m \rangle +\sum_{k=-N}^N a_{m,k}\vert N,k \rangle \right\}
\]
We compute  $\langle \hat{m} \rangle^{(0+1)} (t) = \langle \hat{m} \rangle (t) + \langle \hat{m} \rangle^{(1)} (t)$ where
\begin{equation}
\langle \hat{m} \rangle (t) = \cos \theta \sum_{m=-N}^N m C_m^2 + \sin \theta \sum_{m=-N}^{N-2} C_mC_{m+2} \sqrt{ N(N+2) - m(m+2) } L_m (t)  \label{unpertEv}
\end{equation}
with $L_m :=\cos\left( \phi +\left( E_{m+2}-E_m \right)t \right)$ and
\begin{equation}
\begin{array}{lcl}
	\displaystyle \langle \hat{m} \rangle^{(1)} (t)& = &\displaystyle 2 \sum_{m=-N}^N  \Biggl\{\sum_{l=-N}^N a_{m,l} C_mC_l  \cos\theta \cos(E_l - E_m)t + \\ & &  \displaystyle \sin\theta \left( \sum_{l=-N}^{N-2} a_{m,l}    C_m C_{l+2} \sqrt{N(N+2) - l (l+2)} \cos \left[ \phi + (E_{l+2} - E_m )t \right] + \right.  \\ && \displaystyle   \left. \sum_{l=-N+2}^{N} a_{m,l} C_m C_{l-2} \sqrt{N(N+2) - l (l-2)} \cos \left[ \phi - (E_{l-2} - E_m )t \right] \right) \Biggr\}. \label{pertEv}
\end{array}
\end{equation}

\begin{figure}[t]
\psfrag{m}{$\langle \hat{m} \rangle (t)$}
\psfrag{t}{$t$}
\psfrag{a}{\textit{a)}}
\psfrag{b}{\textit{b)}}
\psfrag{c}{\textit{c)}}
\centering
\includegraphics[width=2in, height=2in]{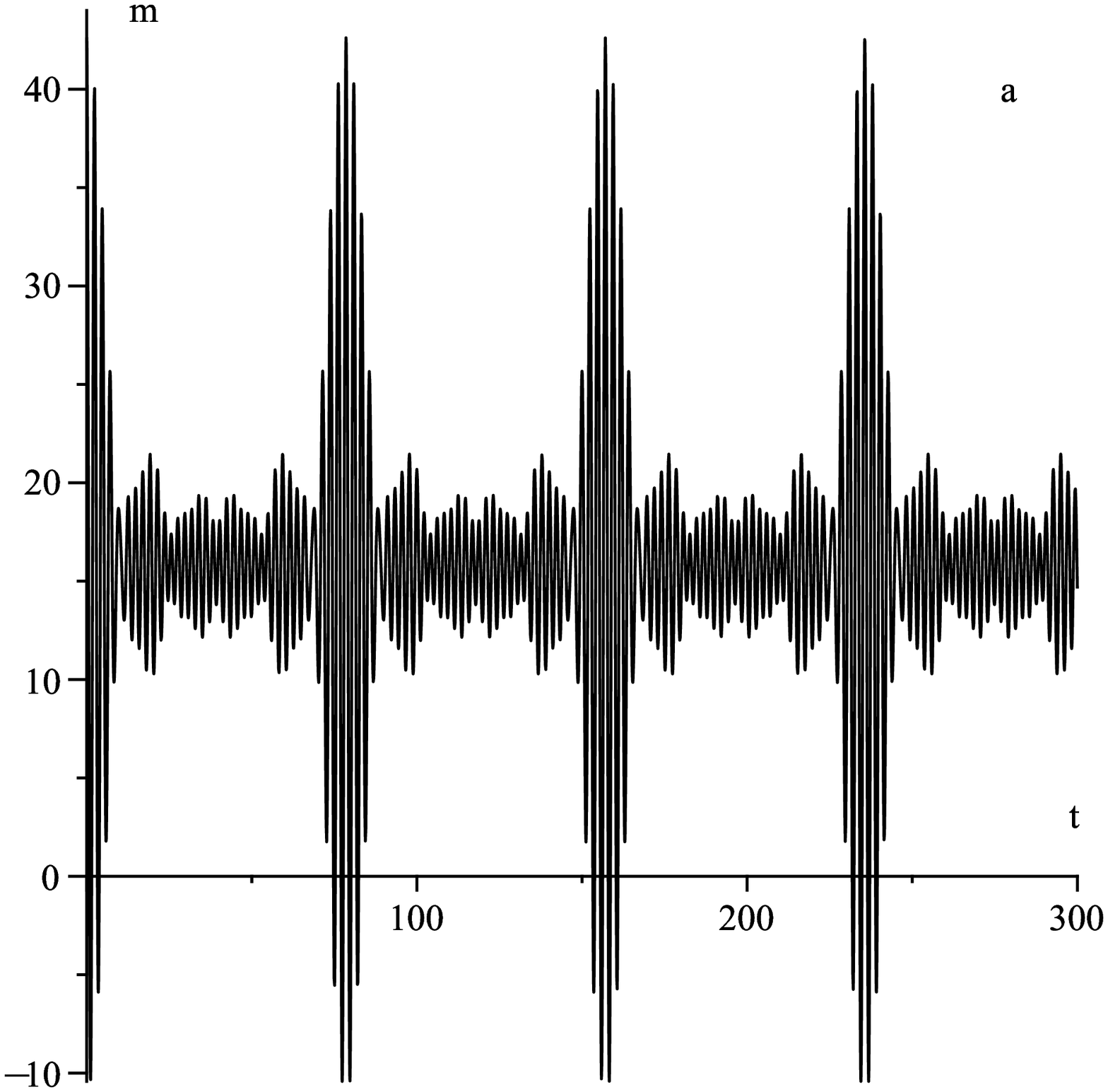}
\hspace{0.4in}
\includegraphics[width=2in, height=2in]{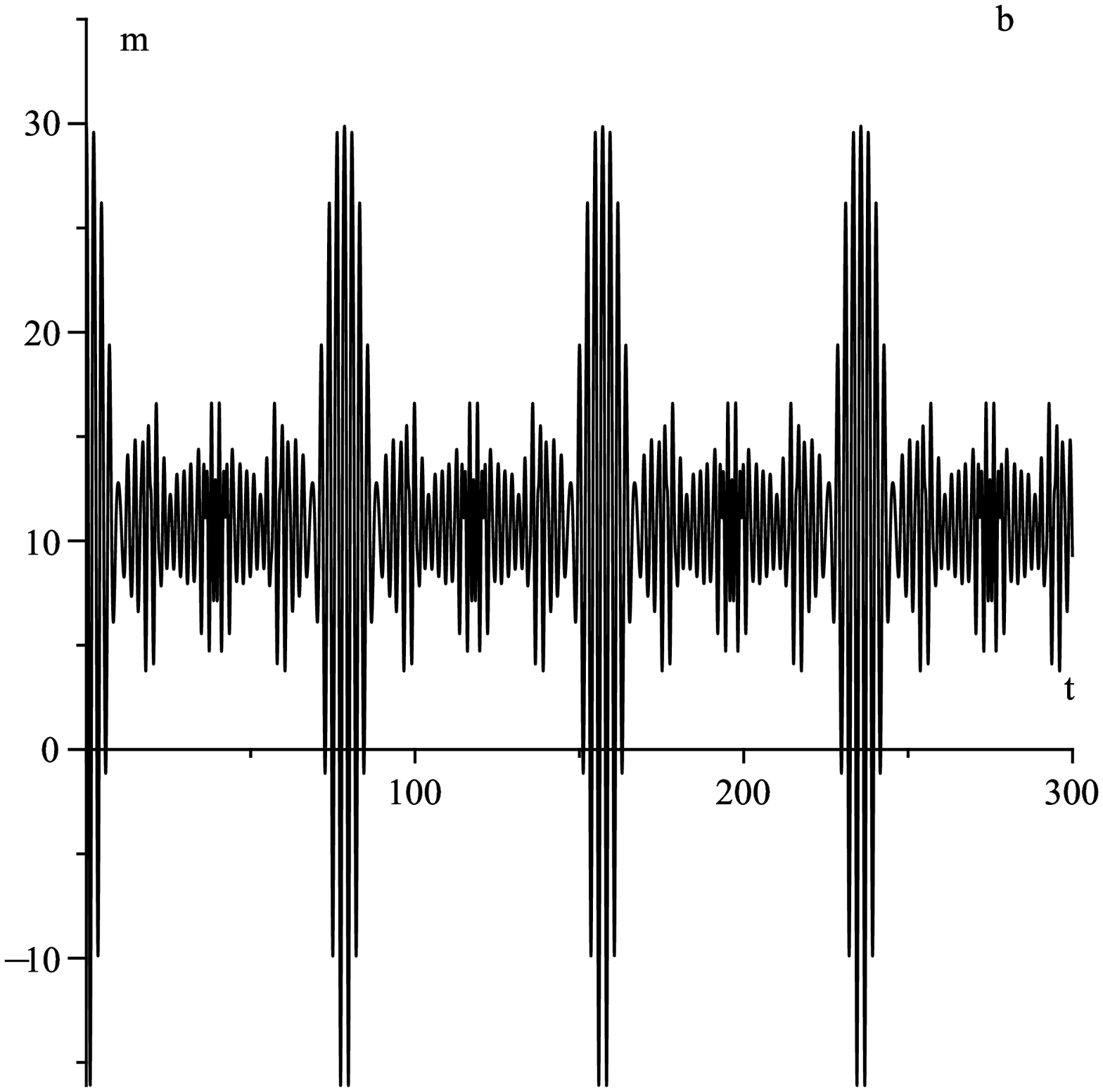}
\hspace{0.4in}
\includegraphics[width=2in, height=2in]{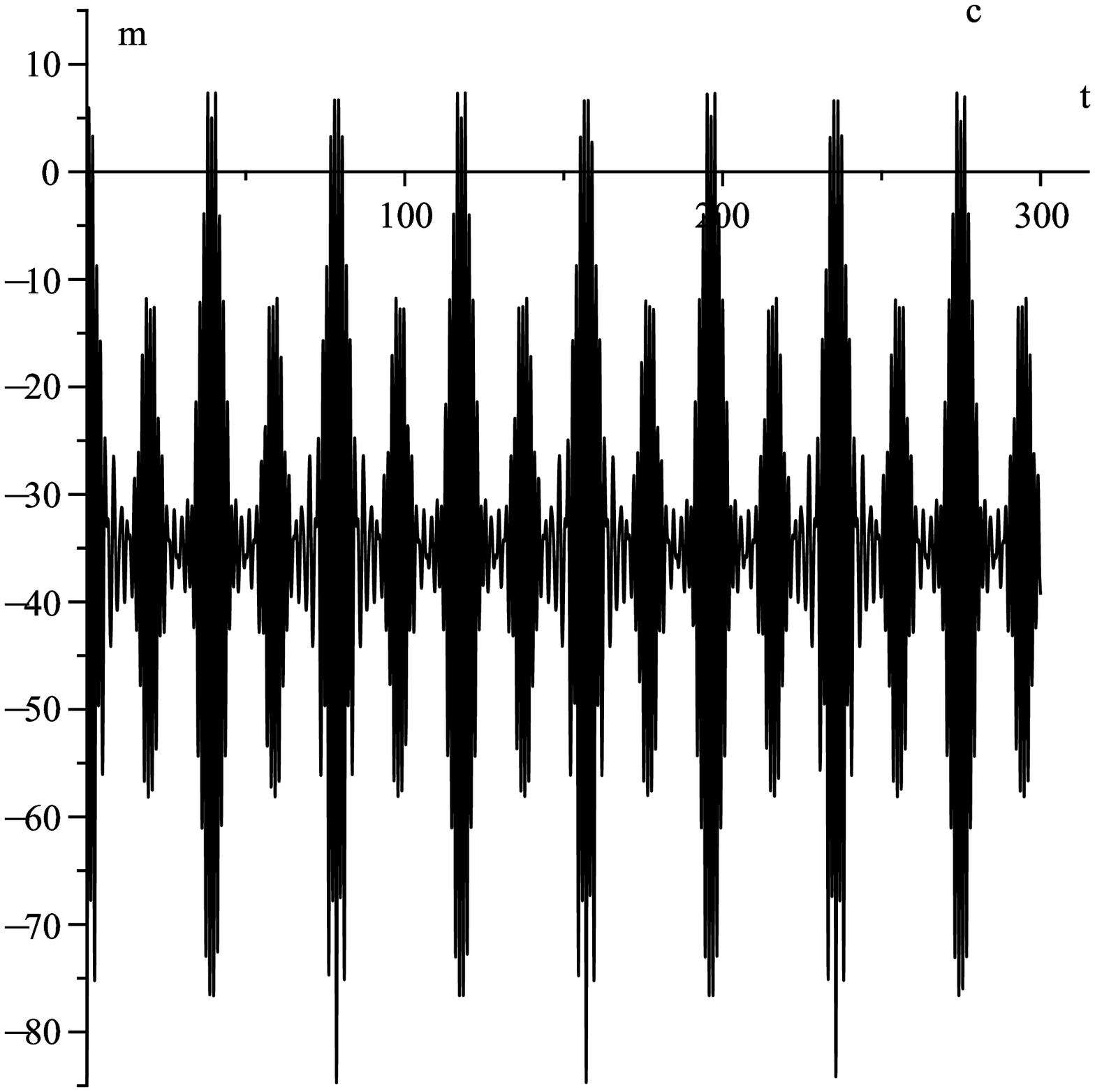}
\caption{Figure \textit{a)} plots the unperturbed evolution of relative probability $\langle \hat{m} \rangle (t)$ for $N=50$. Figures \textit{b)} and \textit{c)} plot the perturbed evolution $\langle \hat{m} \rangle^{(0+1)} (t)$ given by eqs. (\ref{unpertEv}) and (\ref{pertEv}) for $N=50$ with $\delta_{\omega}=\frac{1}{200}$ and $\delta_{\omega}=\frac{1}{20}$, respectively. }
\label{evoPlot}
\end{figure}

We plot in Figure \ref{evoPlot} \textit{a)} the unperturbed evolution of relative population given by eq. (\ref{unpertEv}).  We see the Rabi-like oscillations with relative population collapse and revival.  In Figures \ref{evoPlot} \textit{b)} and \textit{c)} we plot the evolution of relative population under the parameter perturbations $\delta_{\omega}=0.005$ and $\delta_{\omega}=0.05$, respectively.  We observe that as $\delta_{\omega}$ increases the time-averaged value of $\langle \hat{m} \rangle^{(0+1)}$ decreases;  from eq. (\ref{pertEv}) changing the sign of $\delta_{\omega}$ would have increased this average value.  We see in Figure \ref{evoPlot} \textit{c)} that the perturbation has broken down; the maximum value of $\vert \langle \hat{m} \rangle^{(0+1)} \vert$ is greater than the total particle number.  This breakdown  reminds us that we must not let the perturbation grow too large for our analysis to be reliable. We also see (in Figure \ref{evoPlot} \textit{c}) that as $\delta_{\omega}$ grows the time of population collapse significantly decreases.  Because of the complexity of the correction term eq. (\ref{pertEv}) an analytic study of the effects of the perturbations on population collapse and revival times is not possible.

>From eq. (\ref{unpertEv}) we see that in the unperturbed case an initially pure state remains pure.  Eq. (\ref{pertEv}) shows that the same holds true even in the perturbed cases.  Hence, $ \langle \hat{m} \rangle^{(0+1)} (t)$ has non-trivial time-dependence if and only if the initial state is entangled.

\section{Degenerate Perturbations}
Throughout the above analysis we have assumed that the unperturbed states in question are non-degenerate.  We proceed now to study the degenerate case, which results for specific values of the constants $A_1$ and $A_2$.  For fixed total particle number $N$ two distinct states, labeled by relative population numbers $m_1$ and $m_2$, have the same energy precisely when $m_1+m_2= -\frac{A_1}{A_2}$.  From the analysis of the perturbations completed above we know that the perturbation matrix elements are non-vanishing only if $m_1-m_2 \in \left\{ 2,4\right\}$, where we have assumed without loss of generality that $m_1 > m_2$.  Combining these two observations we see that there are at most two pairs of degenerate states.

Let us consider as an example the perturbation $\omega \mapsto \omega + \delta_{\omega}$ with $A_1=-\left[ N+(N-2) \right] A_2$, so that the only pair of degenerate states is $\vert N, N \rangle$ and $\vert N, N-2 \rangle$; each have energy $A_2N\left( 2-N \right)$. The matrix of interest is
\begin{equation}
\Delta_{\omega}= \left(
\begin{array}{cc} \langle N,N \vert \tilde{H}_{\omega} \vert N,N \rangle & \langle N,N \vert \tilde{H}_{\omega} \vert  N, N-2 \rangle \\
\langle N, N-2 \vert \tilde{H}_{\omega} \vert N, N \rangle & \langle N, N-2 \vert \tilde{H}_{\omega} \vert  N, N-2 \rangle \\
\end{array}
\right);
\end{equation}
its eigenvalues and eigenvectors yield the first-order energy and wave function corrections.  Using the calculations performed above it remains to find the solutions $\epsilon_{\pm}$ of the quadratic $\hbox{det}\left(\Delta_{\omega}-\epsilon I_{2\times 2} \right)=0$. The eigenvalues (energy corrections) and corresponding eigenvectors (wave function corrections) of $\Delta_{\omega}$ are
\begin{subequations}
\begin{equation}
\frac{ \epsilon_{\pm}}{\delta_{\omega} } = (N-1)\cos\theta \pm \sqrt{ N - (N-1)\cos^2 \theta }  \label{degen_En_omega}
\end{equation}
\begin{equation}
\vert \pm \rangle =a_{\pm} \left( \sqrt{N}\sin\theta \vert N, N-2 \rangle + \left( \cos\theta \pm \sqrt{N - (N-1)\cos^2 \theta} \right) \vert N, N \rangle \right)
\end{equation}
\end{subequations}
with $a_{\pm}$  a suitable normalization constant.

Figure \ref{epsilonPlot_omega} \textit{a)} plots $\epsilon_{\pm}$ as a function of $\theta$ for $N=1000$ and $\delta_{\omega}=0.01$. We see that the energy is lowered for all values of $\theta$; had $\delta_{\omega}$ been negative the opposite would have been true.  We see from Figures \ref{epsilonPlot_omega} \textit{b)} and \textit{c)}, which plot the perturbed particle distributions for $\delta_{\omega}=\pm 0.01$, that even for very small perturbations in $\omega$ there is a noticeable change in the particle distribution.  This is in contrast to the non-degenerate perturbation of $\omega$, where even for $\delta_{\omega} = 15$ there was not a large change in the particle distribution.  The induced change in the particle distribution is also qualitatively different from that in the non-degenerate case.  This arises because the correction in the degenerate case is sinusoidal with frequency much greater than that of the unperturbed particle distribution. Regardless of the sign of $\delta_{\omega}$ we see that the central maximum of the particle distributions is shifted to smaller values of $m$.

We plot in Figure \ref{ent_Degen_omega} \textit{a)} the perturbed entanglement as a function of $\theta$ and $m_0$ for $\delta_{\omega}=0.005$ and $N=100$, while \textit{b)} of the same figure plots the difference $\Delta S$ for the same configuration.  Even for a small perturbation strength $\delta_{\omega}=0.005$ the perturbation to the entanglement is  still significant, again showing that a degenerate system is more sensitive to perturbations than is the non-degenerate case.  We also observe that the largest perturbations are present for $m_0$ close to $-N$ and $\theta \approx \frac{\pi}{4}, \frac{3\pi}{4}$. For $m_0$ close to $N$, the perturbations become negligible, regardless of $\theta$. Therefore, we find that high collision rates also help stabilize the condensate against perturbations in the degenerate case. However, in this case, condensates with positive scattering lengths are more stable.

Similarly, we find that for a perturbation $\lambda \mapsto \lambda + \delta_{\lambda}$, assuming again that $\vert N, N \rangle$ and $\vert N, N-2\rangle$ are the degenerate states, the energy corrections are
\begin{equation}
\frac{ \epsilon_{\pm}}{\delta_{\lambda} } = (N-1)\sin\theta \pm \sqrt{N-(N-1)\sin^2\theta},   \label{degen_En_lambda}
\end{equation}
which is just the energy correction equation for perturbations in $\omega$ with the substitution  $\cos\theta \mapsto \sin\theta$; we obtain the wavefunction corrections $\vert \pm \rangle$ for $\delta_{\lambda}$ in the same manner.  Figure \ref{epsilonPlot_lambda} plots the perturbed energy and perturbed particle distributions for $\delta_{\lambda}=1$ and $\delta_{\lambda}=\pm 0.01$.  The same comments made above about Figure \ref{epsilonPlot_omega} for perturbations $\delta_{\omega}$ hold in this case as well. In Figure \ref{ent_Degen_lambda} we plot the perturbed entanglement and entanglement difference $\Delta S$ for $\delta_{\lambda}=0.005$.  The figure shows similar behaviour to Figure \ref{ent_Degen_omega}.  However, we note that in the case of $\lambda$ perturbations $\Delta S$ is strictly positive (for positive $\delta_{\lambda}$).  There are again extrema (this time both maxima) for $m_0=-N$ and $\theta \approx  \frac{\pi}{4}, \frac{3\pi}{4}$, with $\Delta S$ vanishing as $m_0$ approaches $N$.  Also as in the case above, the system is very sensitive to perturbations, with $\Delta S \approx 10$ at some points.

\begin{figure}[t]
\psfrag{a}{\textit{a})}
\psfrag{b}{\textit{b})}
\psfrag{c}{\textit{c})}
\psfrag{th}{$\theta$}
\psfrag{en}{$\epsilon$}
\psfrag{theta}{" "}
\psfrag{m}{$m$}
\psfrag{P}{$P$}
\centering
\includegraphics[width=2in, height=2in]{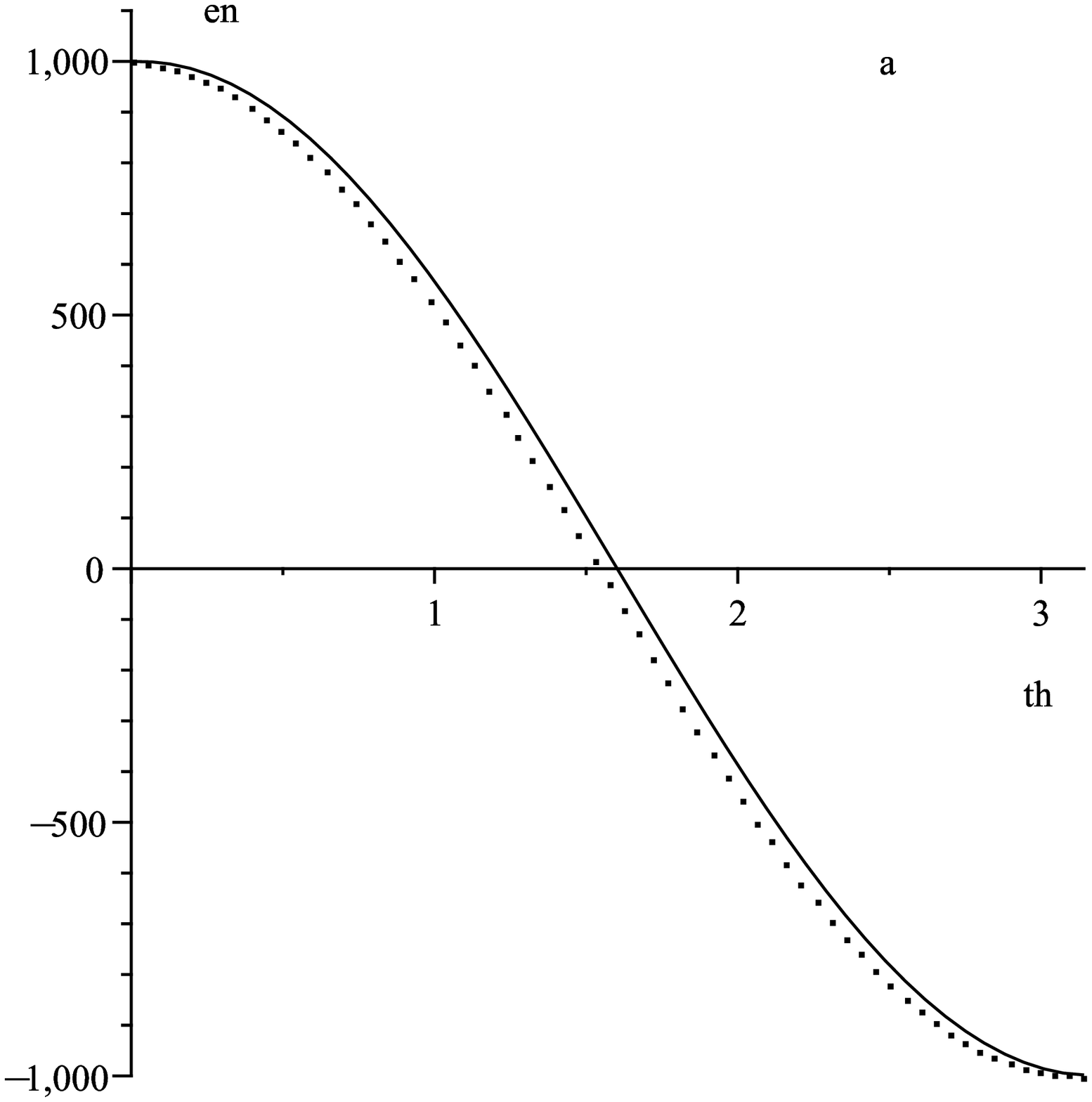}
\hspace{0.4in}
\includegraphics[width=2in, height=2in]{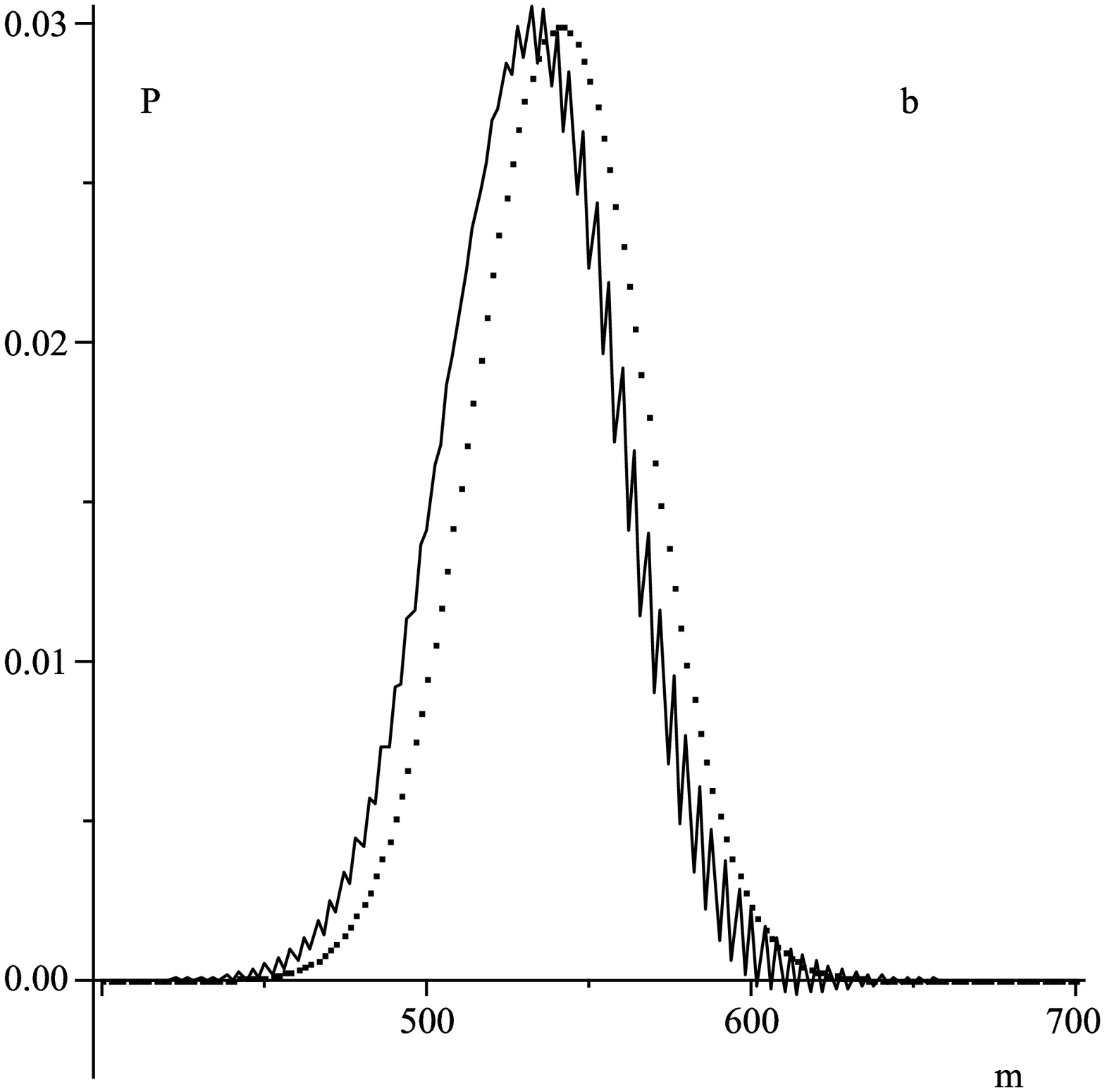}
\hspace{0.4in}
\includegraphics[width=2in, height=2in]{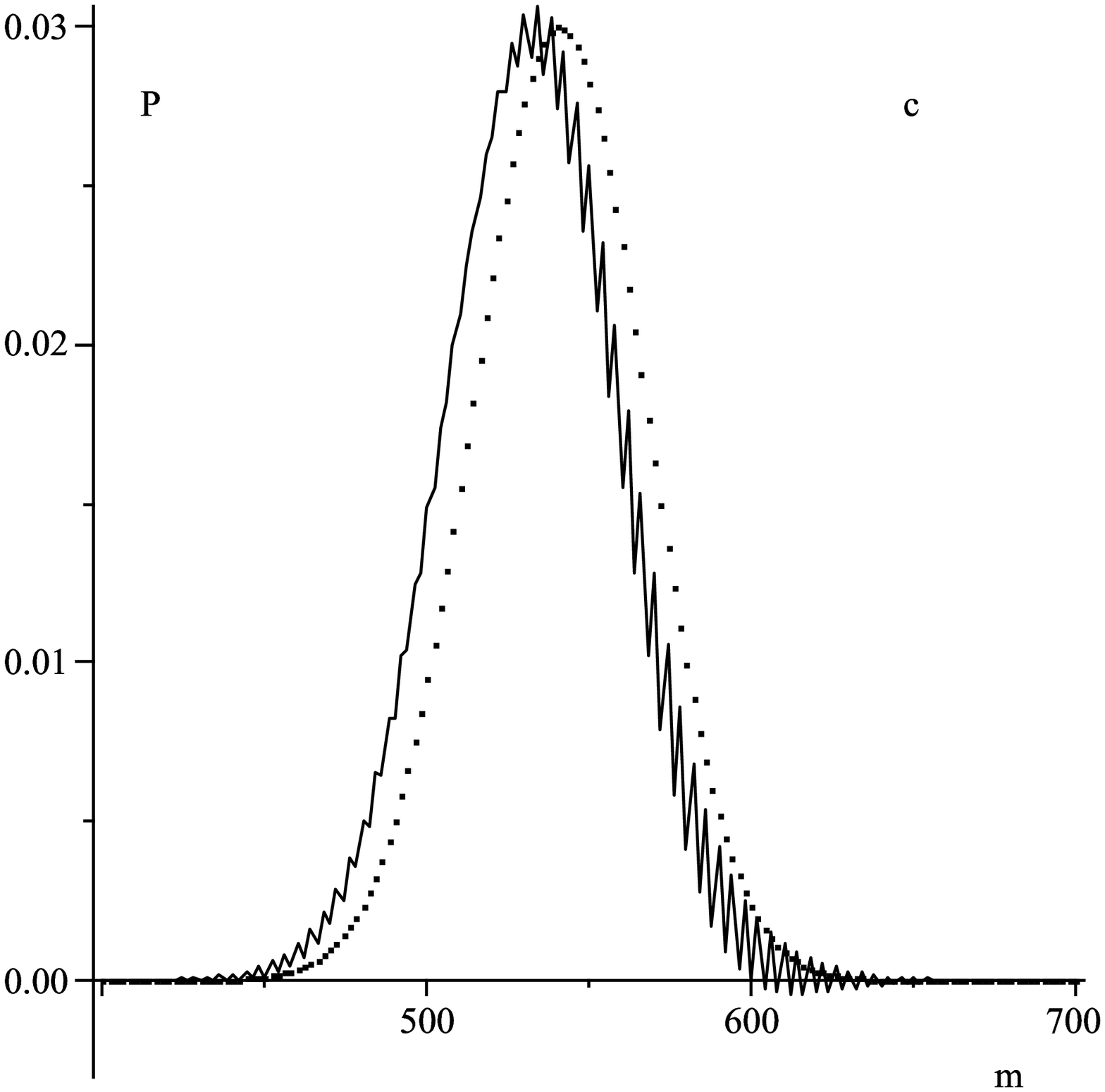}
\caption{Figure \textit{a)} plots $\epsilon_{+}$ (solid) and $\epsilon_{-}$ (dot) from eq. (\ref{degen_En_omega}) for degenerate perturbations with $N=1000$ and $\delta_{\omega}=1$ as a function of $\theta$.   Figures \textit{b)} and \textit{c)} plot the corresponding unperturbed (dot) and perturbed (solid) particle distributions $P$ with $N=1000$, $\theta=1$ for $\delta_{\omega}=0.01$ and $\delta_{\omega}=-0.01$, respectively.}
\label{epsilonPlot_omega}
\end{figure}

\begin{figure}[t]
\centering
\includegraphics[width=2in, height=2in]{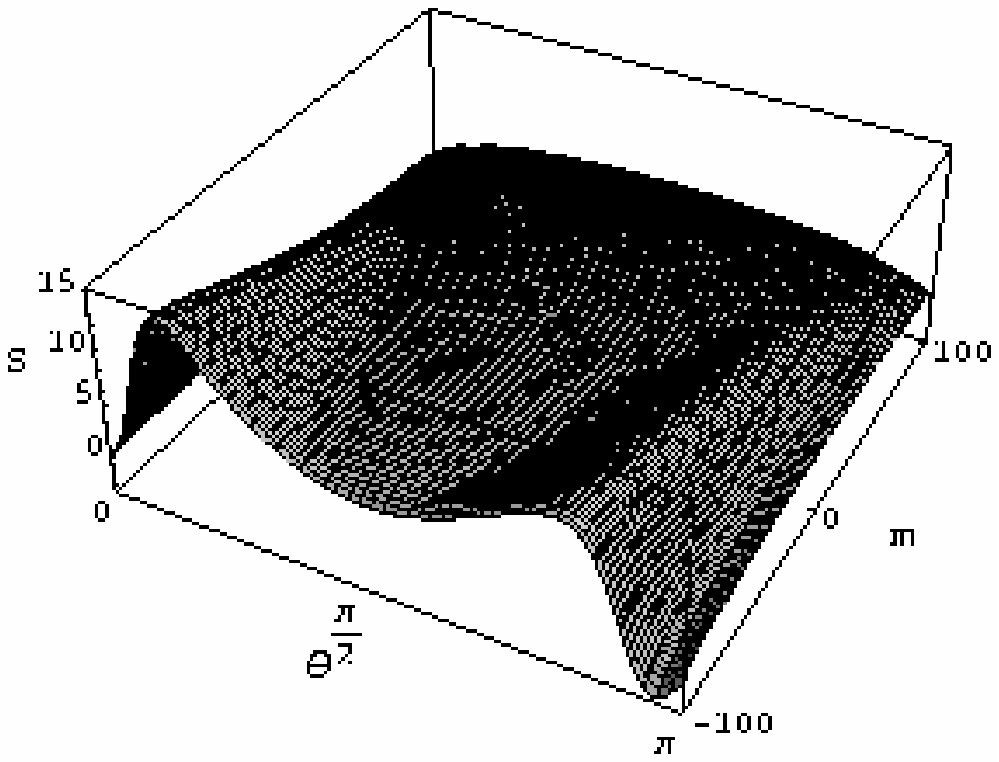}
\hspace{0.4in}
\includegraphics[width=2in, height=2in]{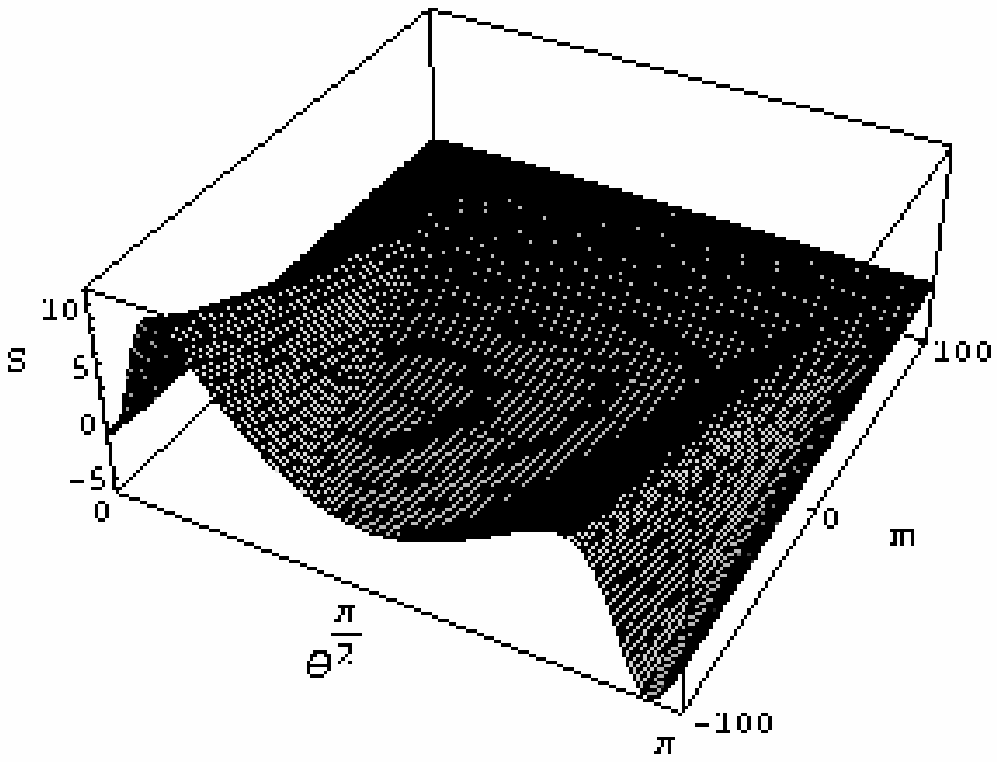}
\hspace{0.4in}
\caption{The figures show the \textit{a)} perturbed entanglement for $\delta_{\omega}=0.005$ as a function of $m_0$ and $\theta$ for $N=100$ and \textit{b)} the difference $S^{(0+1)}-S^{(0)}$ with $\delta_{\omega}=0.005$ for $N=100$.}
\label{ent_Degen_omega}
\end{figure}

\begin{figure}[t]
\psfrag{a}{\textit{a})}
\psfrag{b}{\textit{b})}
\psfrag{c}{\textit{c})}
\psfrag{th}{$\theta$}
\psfrag{en}{$\epsilon$}
\psfrag{theta}{" "}
\psfrag{m}{$m$}
\psfrag{P}{$P$}\centering
\includegraphics[width=2in, height=2in]{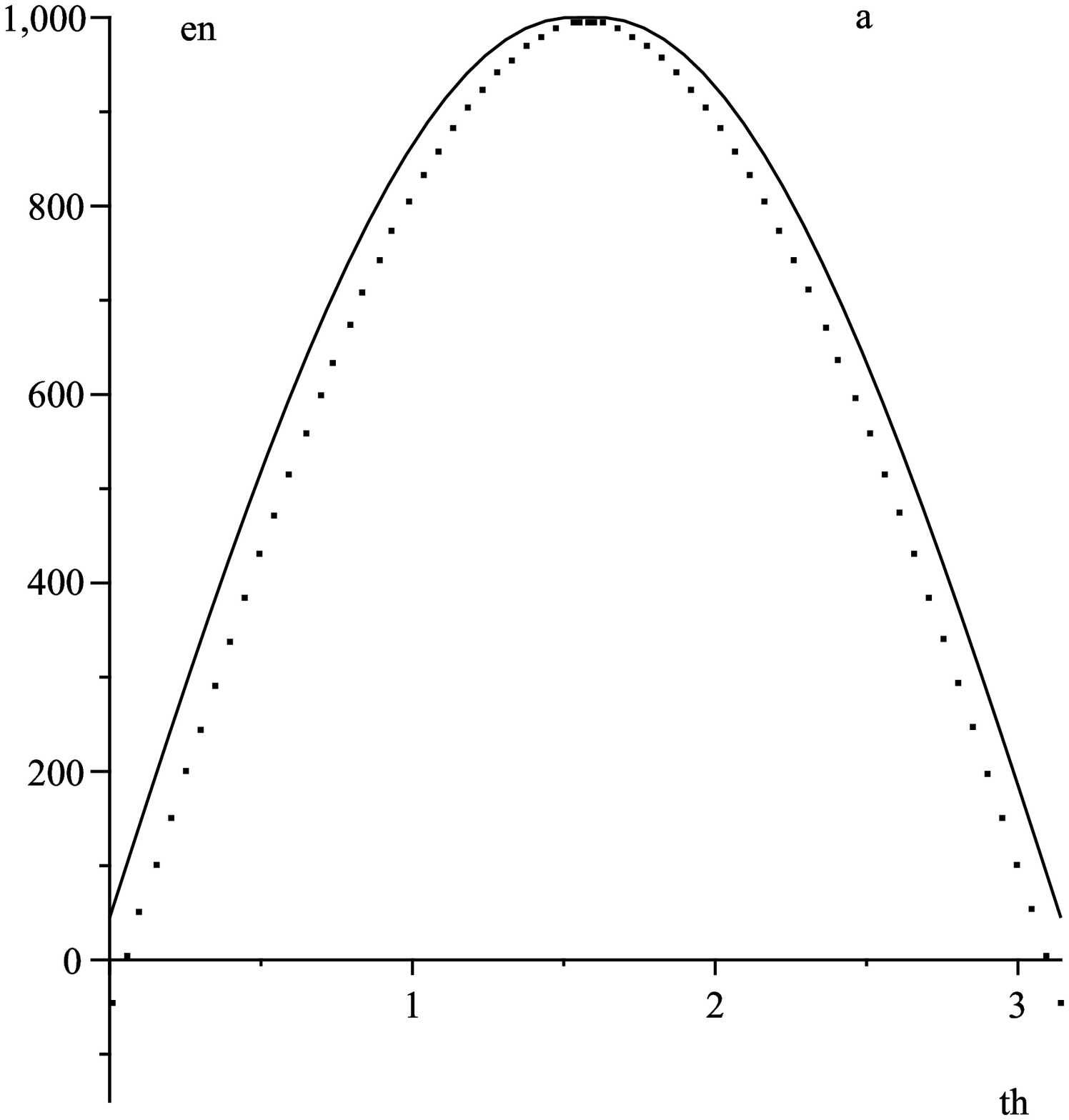}
\hspace{0.4in}
\includegraphics[width=2in, height=2in]{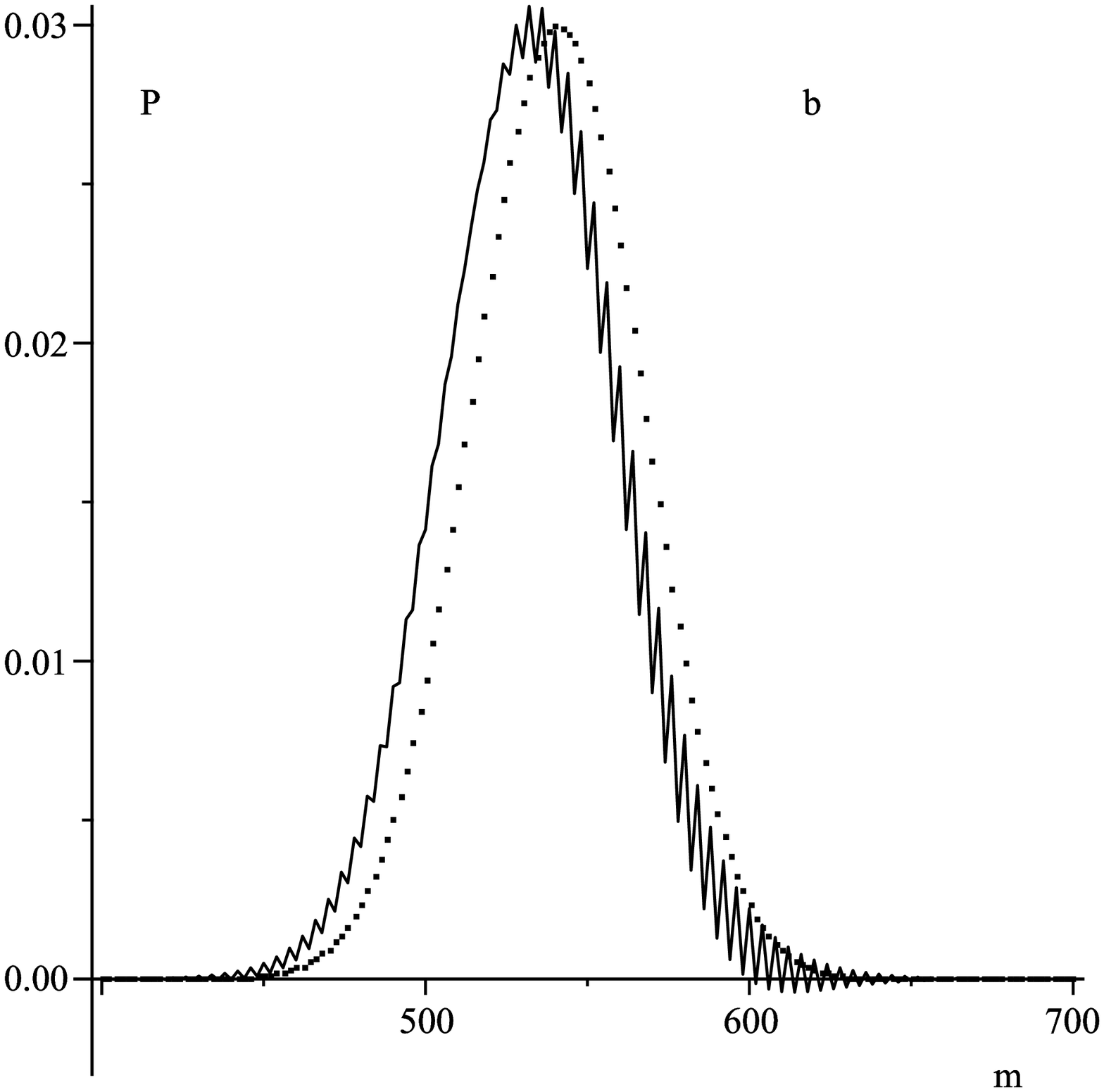}
\hspace{0.4in}
\includegraphics[width=2in, height=2in]{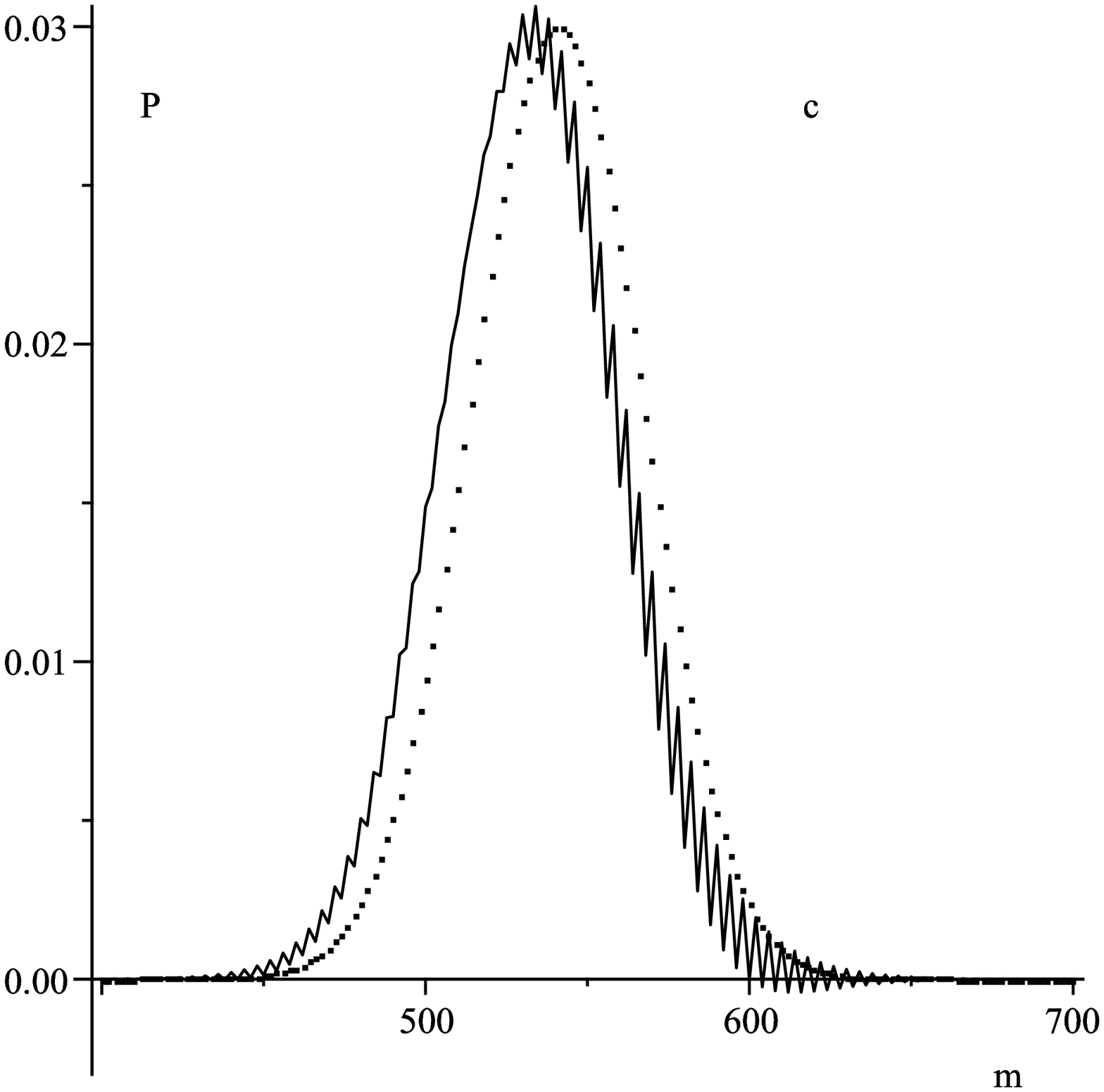}
\caption{Figure \textit{a)} plots $\epsilon_{+}$ (solid) and $\epsilon_{-}$ (dot) from eq. (\ref{degen_En_lambda}) for degenerate perturbations with $N=1000$ and $\delta_{\lambda}=1$ as a function of $\theta$.  Figures \textit{b)} and \textit{c)} plot the corresponding unperturbed (dot) and perturbed (solid) particle distributions $P$ with $N=1000$, $\theta=1$ for $\delta_{\lambda}=0.01$ and $\delta_{\lambda}=-0.01$, respectively.}
\label{epsilonPlot_lambda}
\end{figure}

\begin{figure}[t]
\centering
\includegraphics[width=2in, height=2in]{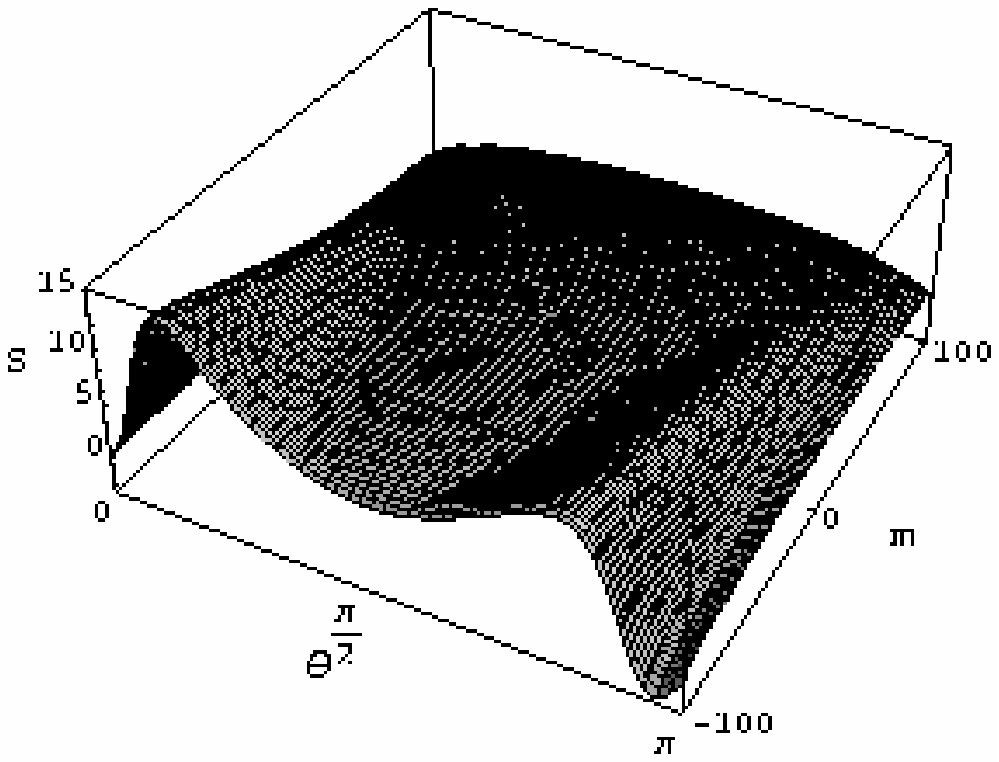}
\hspace{0.4in}
\includegraphics[width=2in, height=2in]{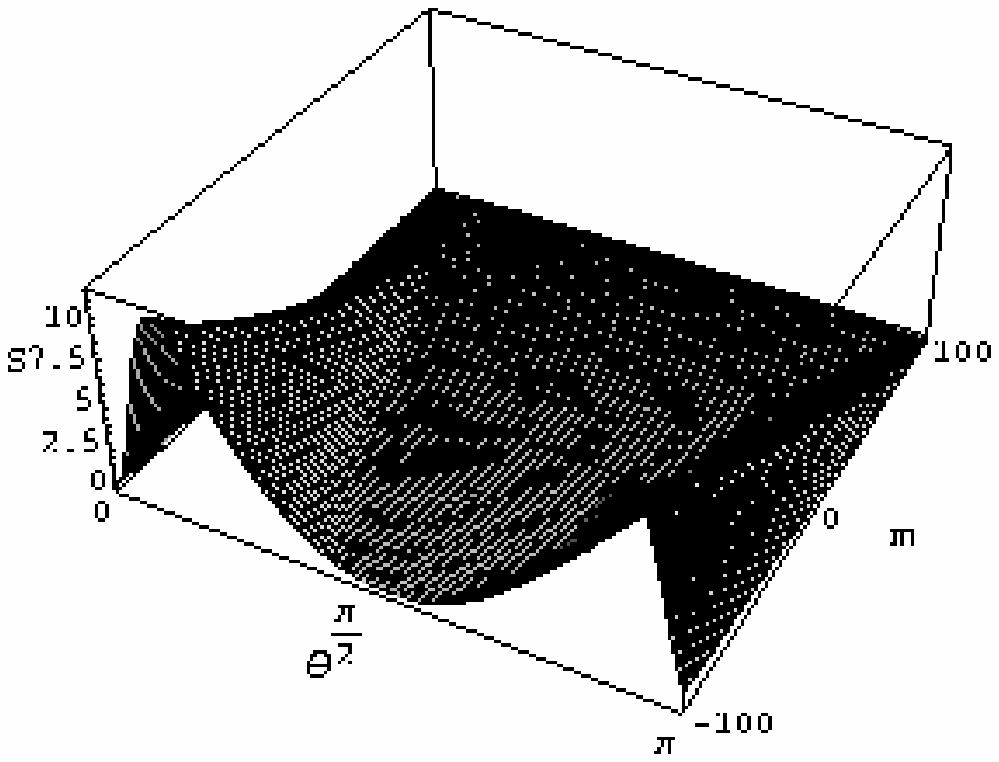}
\hspace{0.4in}
\caption{The figures show the \textit{a)} perturbed entanglement for $\delta_{\lambda}=0.005$ as a function of $m_0$ and $\theta$ for $N=100$ and \textit{b)} the difference $S^{(0+1)}-S^{(0)}$ with $\delta_{\lambda}=0.005$ for $N=100$.}
\label{ent_Degen_lambda}
\end{figure}

The study of degenerate perturbations in the remaining parameters is completed in the same way as is done above, so we omit these. Note that high collision rates help stabilize the condensate against perturbations in $\lambda$ for both positive and negative scattering lengths.

\section{External Perturbations}
As mentioned in the introduction it is of interest to study perturbations that model additional interaction terms not included in the original Hamiltonian eq. (\ref{becHam}).  We will largely be interested in perturbations that do not preserve the total number of particles in the system.  We begin with a general discussion of loss terms and proceed to use this formalism to discuss the effects of background collisions and three-body recombination.

\subsection{General Loss Terms}

Interactions in BECs that do not preserve the total particle number are often minimized in experimental settings as there is currently no known analytical model that involves such terms.  While the Hamiltonian $H_2$ we are studying also does not include such terms, in the case that the effects of particle loss terms are expected to be minimal, we may treat loss terms as perturbations to the system.  A primary source of particle loss is inelastic collisions \cite{corne,holzm}.  In magnetic traps particle-type exchange terms dominate loss mechanisms, while in optical traps, these terms may be neglected \cite{stamp}. It is thus of interest to obtain predictions of the effects of loss terms by treating them as perturbations to the exactly solvable $2$-model.

The most general loss term can be written as
\begin{equation}
H_{loss} = \sum_{k_a=0}^{N_a} \sum_{k_b=0}^{N_b} f_{k_a, k_b} \left( \hat{n}_a, \hat{n}_b \right) a^{k_a}b^{k_b}   \label{lossHam}
\end{equation}
where $f_{k_a, k_b} \left( \hat{n}_a, \hat{n}_b \right)$ are some functions and $N= N_a + N_b$ is the total particle number of the condensate. By probabilistic arguments we expect $f_{k_a, k_b} \rightarrow 0$ as $k_a + k_b$ grows, so that the cases of primary interest are those with $k_a$ and $k_b$ small. However, it should be noted that higher order collisions (i.e. not just $k_a=1,2, k_b=1,2$) are of physical significance, particularly when the condensate is in its coldest phase and of high particle density \cite{corne, holzm}.  By suitably choosing $f_{k_a, k_b}$ we can model specific loss terms.

In order to study loss terms we must first make some adjustments to the analysis performed above.  The full Hilbert space $\mathcal{H}$ of the Hamiltonian eq. (\ref{simpHam}) can be orthogonally decomposed as the Fock space $\mathcal{H} = \bigoplus_{n=0}^{\infty} \mathcal{H}_n$ where, in the $\{ N, m \}$ basis, $\mathcal{H}_n := \left\{ \; \vert n, m \rangle \;\; \vert \;\; m= -n, -n+2, \dots, n-2,n \right\} $. Since the Hamiltonian and perturbations we have considered so far have all commuted with the total number operator we have been able to first choose a total particle number $N$ for the system, or equivalently the subspace $\mathcal{H}_N \subset \mathcal{H}$ of the total Hilbert space, and then proceed with calculations.  In order to study loss terms we must enlarge the state space to be $\mathcal{H}_{\mathcal{A}} = \bigoplus_{n\in \mathcal{A} } \mathcal{H}_n$ where $\mathcal{A} \subset \left\{ 0 \right\} \bigcup \mathbb{N} $ is the set of all accessible total particle numbers.  For example, if $H_{loss} \propto a^2$ then $\mathcal{A} = \left\{ N, N-2 \right\}$.

The energy of the state $\vert N,m \rangle$ is $E_m = A_1m +A_2 m^2$.  Although $E_m$ is only functionally dependent on $m$, it has an implicit dependence on $N$ since $N$ restricts the values of $m$.  Hence, while $\vert N, m \rangle $ is non-degenerate as an element of $\mathcal{H}_N$, it may be degenerate as an element of $\mathcal{H}_{\mathcal{A}}$ depending on $m$ and $\mathcal{A}$.  Again, considering the example in which $H_{loss} \propto a^2$ we see that $\vert  N, m \rangle \in \mathcal{H}_N \oplus \mathcal{H}_{N-2}$ is non-degenerate if $m= \pm N$ and degenerate otherwise.  In general, let $\mathcal{S}_{\mathcal{A}} (m) = \left\{ \vert n, m \rangle \; \vert \; n\in \mathcal{A} \right\}$ and say that $\mathcal{S}_{\mathcal{A}} (m)$ is degenerate if it contains more than one element and say it is non-degenerate otherwise.  We examine the effects of the degeneracy $\mathcal{S}_{\mathcal{A}} (m)$ below. Note that the degeneracy studied in this section is caused by the interactions (which determine the Hilbert space), whereas the degeneracy studied in the previous section was caused by a specific choice of the coupling constants $A_1$ and $A_2$.

In order to study the particle distribution of the perturbed state $\vert N, m_0\rangle^{(0+1)}$ we must modify eq. (\ref{probDens}); if we were to use this formula there would be no perturbative effects on the particle distribution since $\mathcal{H}_i$ and $\mathcal{H}_{j}$ are orthogonal if $i\neq j$.  A suitable generalization is given by
\begin{equation}
P_{gen} =  \left| \sum_{n\in\mathcal{A} } \langle n , m \vert U^{\dagger} \vert N,m_0 \rangle^{(0+1)} \right|^2.  \label{Pgen}
\end{equation}
Observe that in the case of no loss terms $\mathcal{A} = \left\{ N \right\}$ so that $P_{gen}$ reduces to eq. (\ref{probDens}).

\subsection{ Effects of the Degeneracy of $\mathcal{S}_{\mathcal{A}} (m)$}
If $\mathcal{S}_{\mathcal{A}} (m)$ is degenerate it is easy to see that the matrix of interest to degenerate perturbation theory $\left( \langle n_i , m \vert \tilde{H}_{loss} \vert n_j, m \rangle \right)_{n_i,n_j \in \mathcal{A} }$ is triangular with zeros along the diagonal.  Indeed, $H_{loss}$ contains only annihilation terms, and conjugation by $U$, denoted here by $\sim$, maps annihilation operators to annihilation operators.  So $\tilde{H}_{loss} \vert n_j, m \rangle$ is a sum of states with total particle number $n$ less than $n_j$, showing the matrix at hand is triangular, and thus has a trivial spectrum consisting of only zeros.  Hence we can learn nothing from first order perturbation theory.

Alternatively, if $\mathcal{S}_{\mathcal{A}} (m)$ is non-degenerate we may apply the tools of non-degenerate perturbation theory.  Although it is straightforward to compute the matrix elements of $H_{loss}$ given by eq. (\ref{lossHam}) in  general, the requirement that $\mathcal{S}_{\mathcal{A}} (m)$ be non-degenerate severely limits the usefulness of such a calculation.  We instead focus on some specific choices of $f_{k_a, k_b}$ that model interactions of physical interest.

\subsubsection{Background Collisions}

Background collisions most often occur in BECs when particles from the condensate collide with a residual background gas in the condensate chamber, or alternatively, with metastable atoms within the condensate \cite{sacke}. Background collisions become more important as the density of the condensate increases.  As a simple illustration of how we may treat background collisions as perturbations, consider the case in which the initial state is $\vert N, N\rangle$, so that only $a$ mode particles can be ejected. The general loss term\footnote{Explicitly, we take $f_{k_a,k_b}=0$ if $k_b\neq 0$ and $f_{k_a,0}=\alpha_{k_a} \in \mathbb{R}$ otherwise.} eq. (\ref{lossHam}), after conjugation by $U$, is written as $\tilde{H}_{loss} = \sum_{k=1}^N \alpha_k \cos^k \frac{1}{2}\theta a^k + O(b)$ where we use $O(b)$ to denote terms with more $b$ powers than $b^{\dagger}$ powers.  Note that any term of $O(b)$ annihilates $\vert N, N \rangle$ so that the perturbative correction can be found by neglecting all such terms.  We compute the desired matrix elements to find that
\begin{equation}
\vert N, N \rangle^{(1)} = \sum_{k=1}^N \frac{\alpha_k\cos^k \frac{\theta}{2} \sqrt{ \prod_{j=0}^{k-1} \left( N - j \right) } }{kA_1 + k(2N - k)A_2} \vert N-k, N-k \rangle.
\end{equation}
which yields the generalized particle distribution
\begin{equation}
P_{gen}= \left| d_{m,N}^N \right|^2 +2\sum_{k=1}^N\frac{ \alpha_k\cos^k \frac{\theta}{2} \sqrt{ \prod_{j=0}^{k-1} \left( N - j \right) } }{kA_1 + k(2N - k)A_2} d_{m,N}^N d_{m,N-k}^{N-k}  \label{lossPD}
\end{equation}
where we set $d_{m,N-k}^{N-k} =0$ if $\vert m \vert > N-k$.  From this expression we see that we could have omitted terms that eject an odd number of $a$-mode particles.  Indeed, if $k$ is odd, then $d_{m,N}^N d_{m,N-k}^{N-k}$ is identically zero as a function of $m$ since if $d_{m,n}^N \neq 0$ then $N,m$ and $n$ all have the same parity.\footnote{We can, however, use second order perturbation theory. In this case terms annihilating an odd number of particles will have an effect on the particle distribution.}  We plot $P_{gen}$ in Figure \ref{BGCPlots}, considering terms that eject $2,4$ and $6$ particles for $N=1000$ and $\theta=1$.  In Figure \ref{BGCPlots}\textit{a)} we set $\alpha_2=-0.1$, $\alpha_4=-0.001$ and $\alpha_6 =-5 \times 10^{-6}$.  In Figures \ref{BGCPlots}\textit{b)} - \ref{BGCPlots}\textit{d)} we increase each of the $\alpha_i$ by a factor of $2$. The figures show that increasing each $\vert \alpha_i \vert$ decreases the height of the particle distribution. Also evident from the figures and the scale of the $\alpha_i$ is that as $i$ increases the perturbations have a larger effect.  That the perturbations do not blow up reflects the requirement that $f_{k_a, k_b} \rightarrow 0$ for large $k_a + k_b$. Indeed, we have found  using numerical simulations that the perturbations that eject $k$ particles diverge with increasing $k$.

\begin{figure}[t]
\psfrag{a}{\textit{a})}
\psfrag{b}{\textit{b})}
\psfrag{c}{\textit{c})}
\psfrag{d}{\textit{d})}
\psfrag{m}{$m$}
\psfrag{P}{$P$}
\centering
\includegraphics[width=1.2in, height=1.2in]{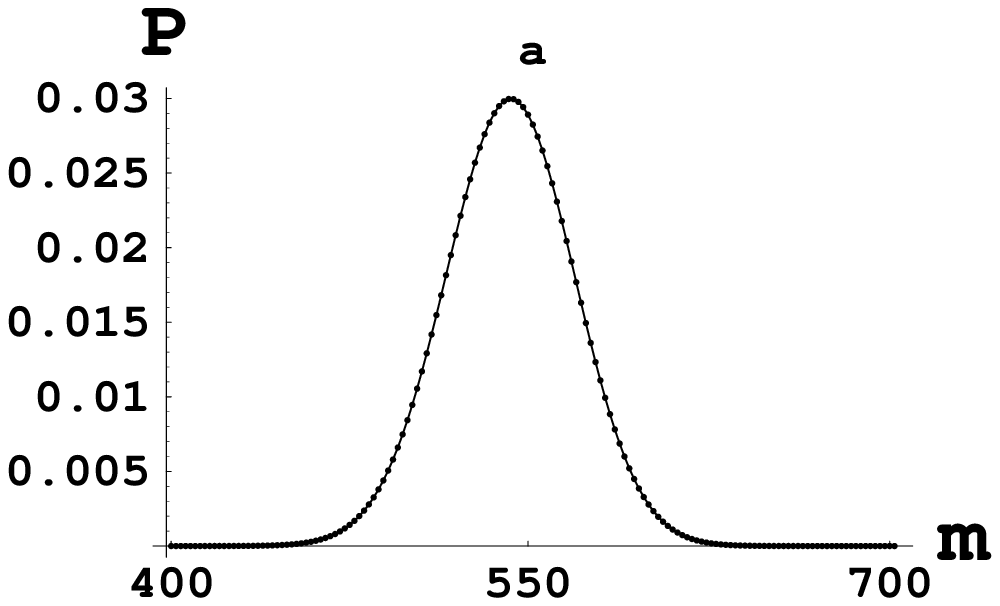}
\hspace{0.3in}
\includegraphics[width=1.2in, height=1.2in]{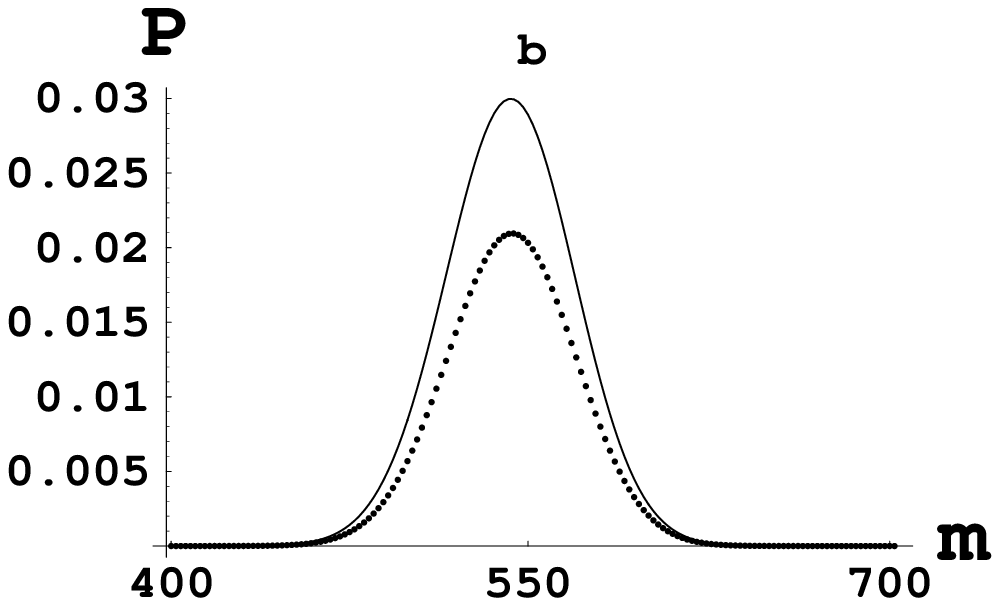}
\hspace{0.3in}
\includegraphics[width=1.2in, height=1.2in]{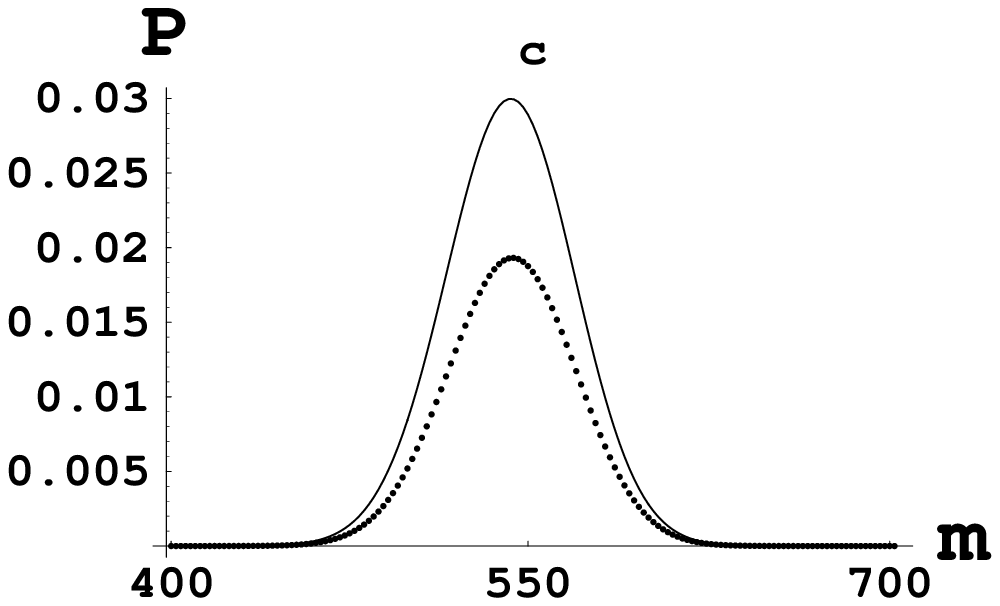}
\hspace{0.3in}
\includegraphics[width=1.2in, height=1.2in]{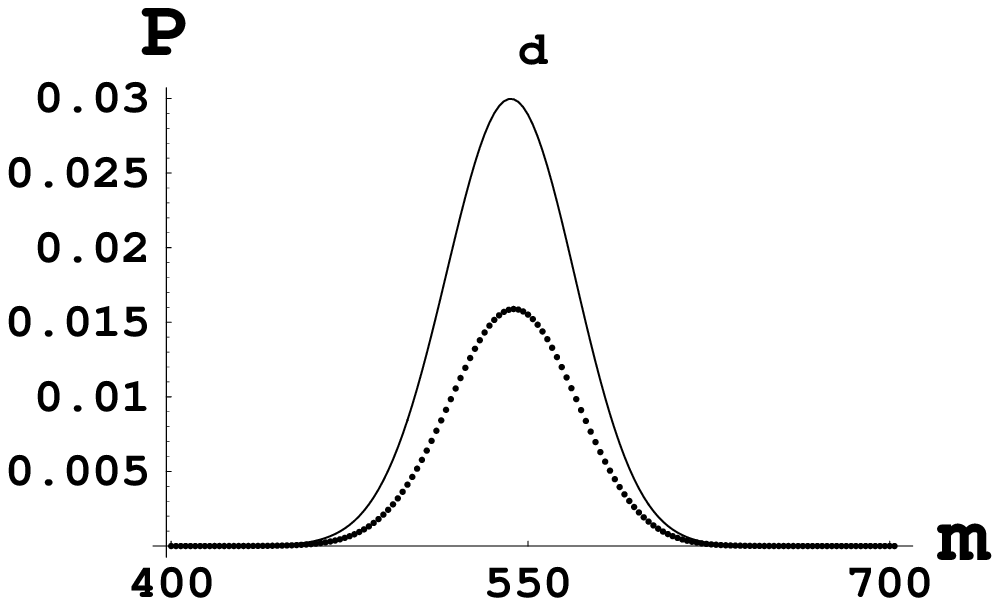}
\caption{The unperturbed (solid) and perturbed (dotted) particle distributions, the latter given by eq. (\ref{lossPD}), for $N=1000$ with \textit{a)} $\alpha_2=-\frac{1}{10}$, $\alpha_4=-\frac{1}{1000}$ and $\alpha_6=-\frac{1}{200000}$, \textit{b)} $\alpha_2=-\frac{2}{10}$, $\alpha_4=-\frac{1}{1000}$, $\alpha_6=-\frac{1}{200000}$,  \textit{c)} $a_2=-\frac{1}{10}$, $\alpha_4=-\frac{2}{1000}$ and $\alpha_6=-\frac{1}{200000}$ and \textit{d)} $\alpha_2=-\frac{1}{10}$, $\alpha_4=-\frac{1}{1000}$ and $\alpha_6=-\frac{2}{200000}$. }
\label{BGCPlots}
\end{figure}

We remark here that as a consequence of $d_{m,N}^N d_{m,N-k}^{N-k}$ vanishing for $k$ odd we cannot learn anything about spin-flip terms from first order perturbation theory.  Indeed, such terms would be modeled by perturbations of the form $b^{\dagger}aa$ and $a^{\dagger}aa$, assuming particles in the $a$ mode have greatest energy; these terms clearly reduce the total number of the system by an odd number.

Setting $P^{(1)}$ equal to the second term on the right hand side of eq. (\ref{lossPD}) we can use eqs. (\ref{pertEnt}-\ref{incEnt}) to compute and increase the entanglement, respectively.  Since $P^{(1)}$ depends linearly on the interaction strengths $\alpha_k$ we see the most obvious manner in which to increase the entanglement is to make $\vert \alpha_k \vert$ large; the sign of $\alpha_k$ depends on $k, \theta$ as well as $A_1$ and $A_2$.

\begin{figure}[t]
\centering
\includegraphics[width=2in, height=2in]{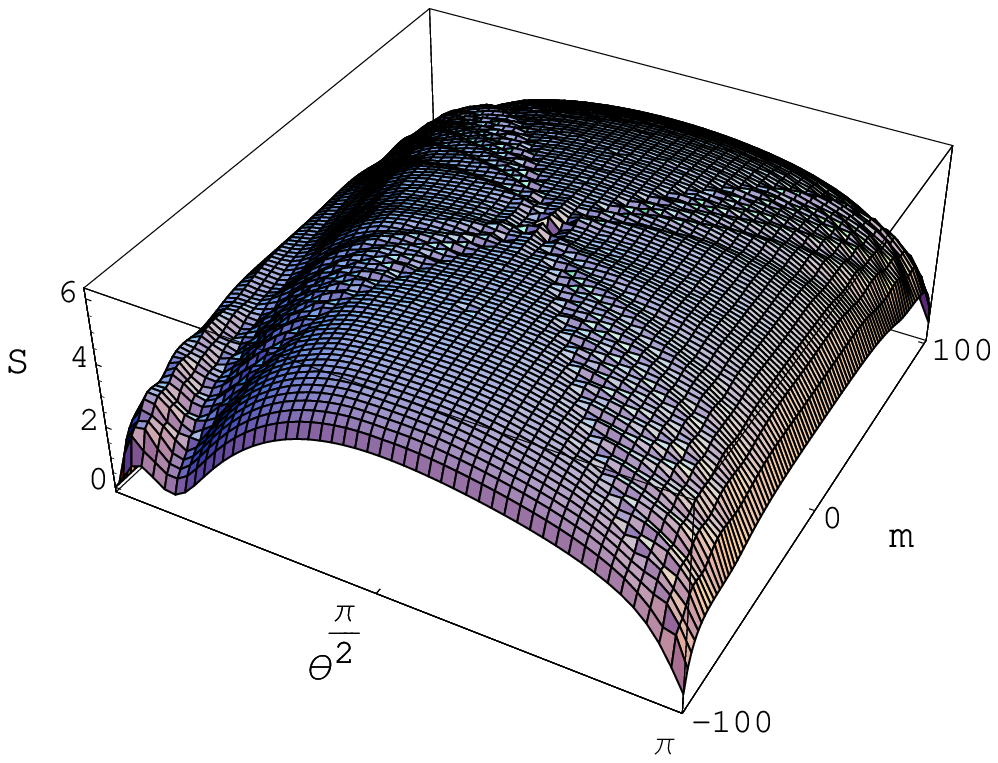}
\hspace{0.4in}
\includegraphics[width=2in, height=2in]{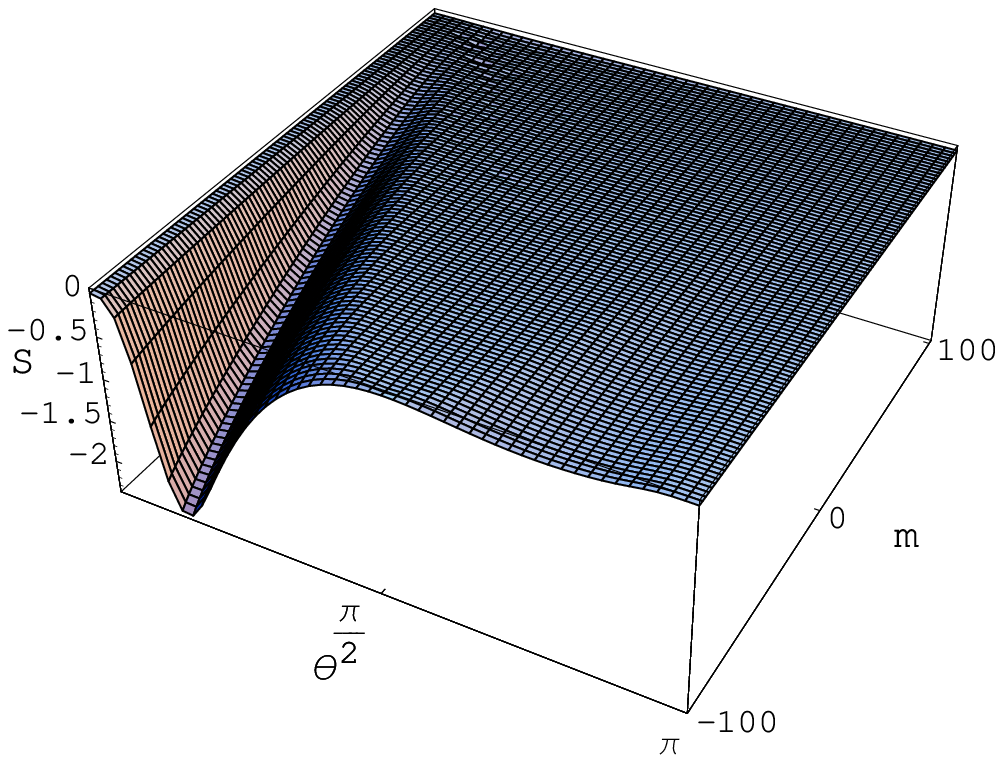}
\hspace{0.4in}
\caption{The figures show the \textit{a)} perturbed entanglement for $\alpha_2=-\frac{1}{10}$, $\alpha_4=-\frac{1}{1000}$ and $\alpha_6=-\frac{1}{200000}$ (as in Figure \ref{BGCPlots} a ) as a function of $m_0$ and $\theta$ for $N=100$ and \textit{b)} the difference $S^{(0+1)}-S^{(0)}$ for the same configuration.}
\label{ent_BGC}
\end{figure}

Figure \ref{ent_BGC} plots the perturbed entanglement caused by the inclusion of background collisions as a function of $m_0$ and $\theta$.  We see that the largest effect of the perturbation occurs in the region $\theta \approx \frac{\pi}{4}$ and small (large negative) values of $m_0$.  There is also a minimum of smaller magnitude around $\theta=\frac{3\pi}{4}$ and $m_0=-N$. The region with large negative $m_0$ corresponds the case where the scattering length between particles is negative, in which most particles lie in the $b$ mode of the condensate. Since the background collisions considered here eject particles from the $a$ mode, they serve to further decrease the value of $m$.  We understand the large effect of the perturbation on the aforementioned region as follows: since most particles lie in the $b$ mode, ejecting any particles from the $a$ mode has a large effect on the system, since it already has only a small number of $a$ mode particles, relative to the number of $b$ mode particles.  Figure \ref{ent_BGC} \textit{b)} shows that for values of $m_0$ above the region in question, where the $b$ mode particles become more scarce, the perturbation has little effect on the system.  Again, we understand this as being because ejecting an $a$ mode particle from a state with large $m$ value is of little significance to the system as a whole.  We also see from this figure that the entanglement is decreased, regardless of $m_0$ and $\theta$, for this specific choice of $\alpha_i$.

\subsubsection{Three-Body Recombination}

Perhaps the most physically important type of inelastic collision leading to particle loss is three-body recombination (TBR) \cite{sodin}.  TBR occurs when three particles in a single mode collide to form a diatomic molecule and a particle of the same mode that carries off any excess energy.  Depending on this energy and the energy of the potential trap, the resultant particle may or may not escape the trap \cite{sodin}.

A perturbation modeling TBR would be most naturally treated in the $\mathfrak{su}(3)$ formalism where the extra bose operator would correspond to the diatomic molecule.  However, in the $\mathfrak{su}(2)$ formalism we have no such third mode.  We thus model TBR by the Hamiltonian
\begin{equation}
H_{TBR} = C \left( \sigma a^{\dagger}aaa + (1-\sigma) aaa \right).
\end{equation}
The parameter $\sigma$ describes the probability of the emitted particle remaining trapped in the condensate.  The term proportional to $aaa$ lowers the total particle number by an odd number and hence, as explained above, will have no effect on the perturbed particle distribution.  We thus neglect this term from our analysis and absorb the constant $C$ into $\sigma$.  With $H_{TBR}$ as a perturbation we have $\mathcal{A} = \{ N, N-2, N-3 \} $.  We take as our unperturbed state $\vert N, N \rangle$, which is non-degenerate.  Proceeding, we find
\begin{equation}
\vert N ,N \rangle^{(1)} = \sigma \frac{\cos^4\frac{\theta}{2} \sqrt{N(N-1)}(N-2)}{2A_1 +4A_2(N-2)} \vert N-2, N-2 \rangle
\end{equation}
which gives the generalized particle distribution
\begin{equation}
P_{gen}(m)= \left| d_{m,N}^N \right|^2 +2\sigma \frac{\cos^4\frac{\theta}{2} \sqrt{N(N-1)}(N-2)}{2A_1 +4A_2(N-2)}  d_{m,N}^N d_{m,N-2}^{N-2}. \label{TBRPD}
\end{equation}

$P_{gen}$ given by eq. (\ref{TBRPD}) is plotted in Figure \ref{TBRPlots}.  As expected the the sign of $\sigma$ determines whether the perturbation vertically shrinks or stretches the particle distribution.  We also see that the system is very sensitive to three-body recombination terms, as a coupling constant of order $0.001$ causes significant changes to the particle distribution.

Once again, setting $P^{(1)}$ equal to the last term in eq. (\ref{TBRPD}), we can use eqs. (\ref{pertEnt}) and  (\ref{incEnt}) to study the perturbed entanglement.  The choice of the sign of $\sigma$ to increase the entanglement again depends on the values of $\theta$, $A_1$ and $A_2$.

\begin{figure}[t]
\psfrag{a}{\textit{a})}
\psfrag{b}{\textit{b})}
\psfrag{c}{\textit{c})}
\psfrag{m}{$m$}
\psfrag{P}{$P$}
\centering
\includegraphics[width=2in, height=2in]{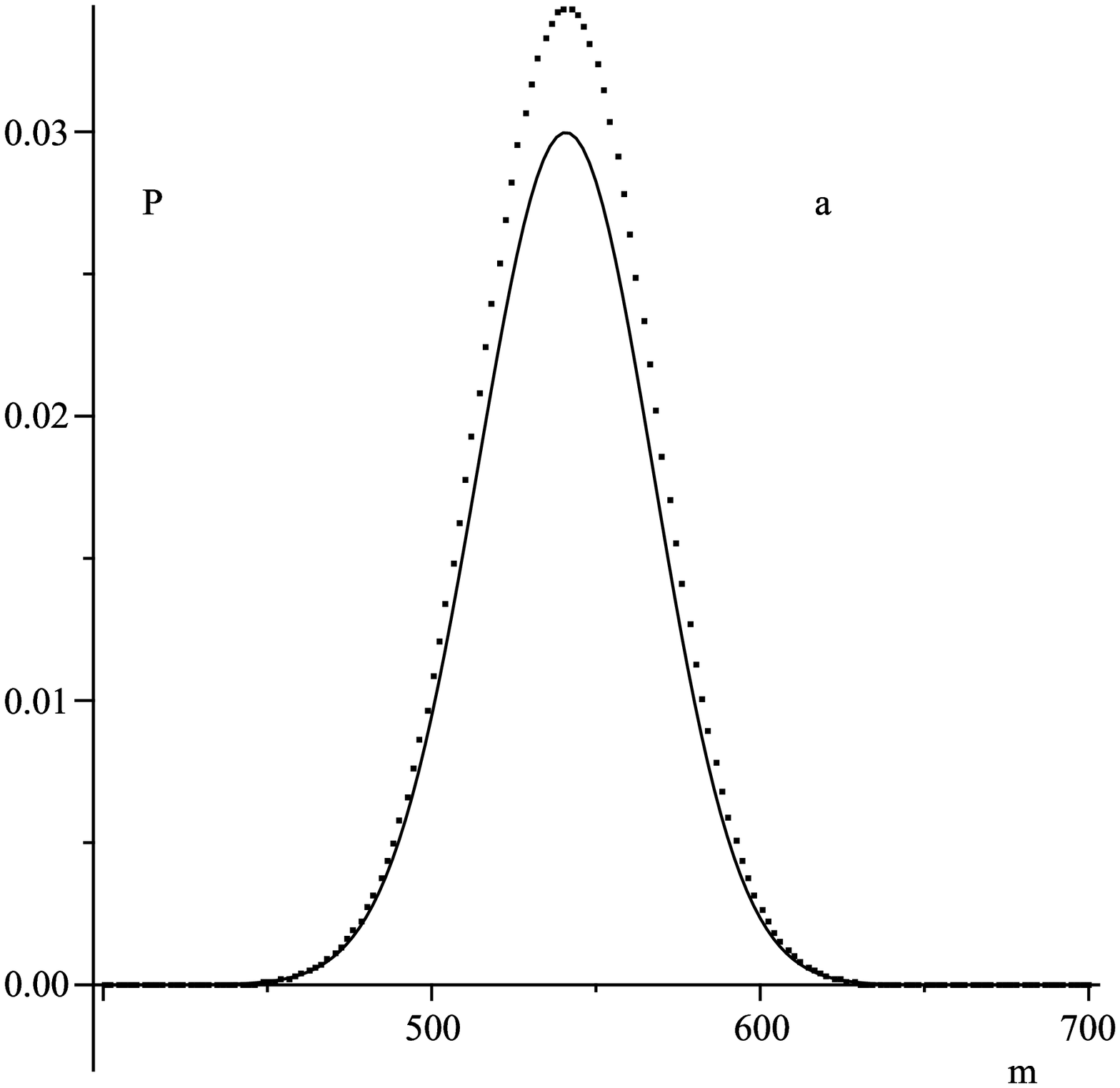}
\hspace{0.4in}
\includegraphics[width=2in, height=2in]{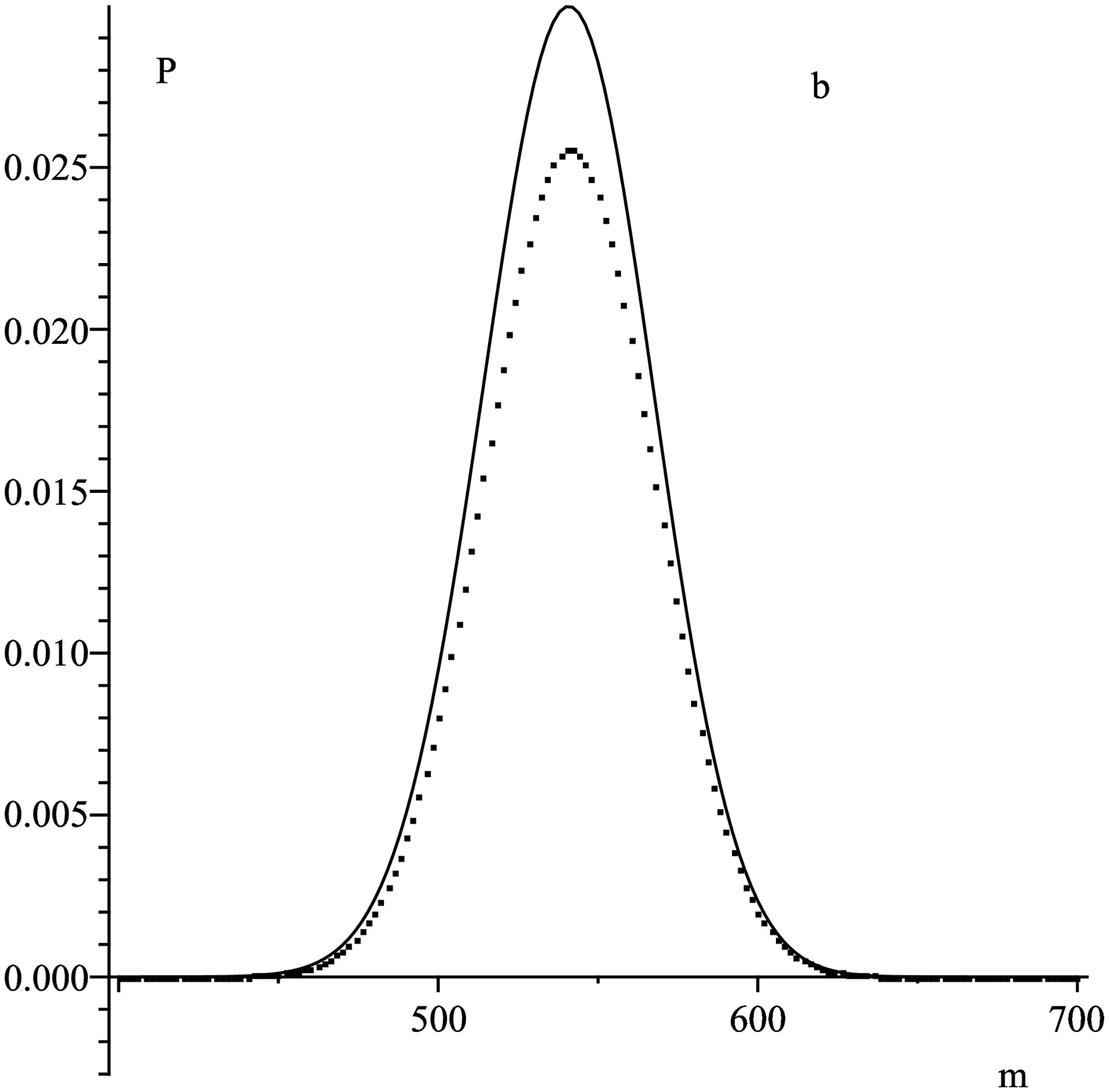}
\hspace{0.4in}
\includegraphics[width=2in, height=2in]{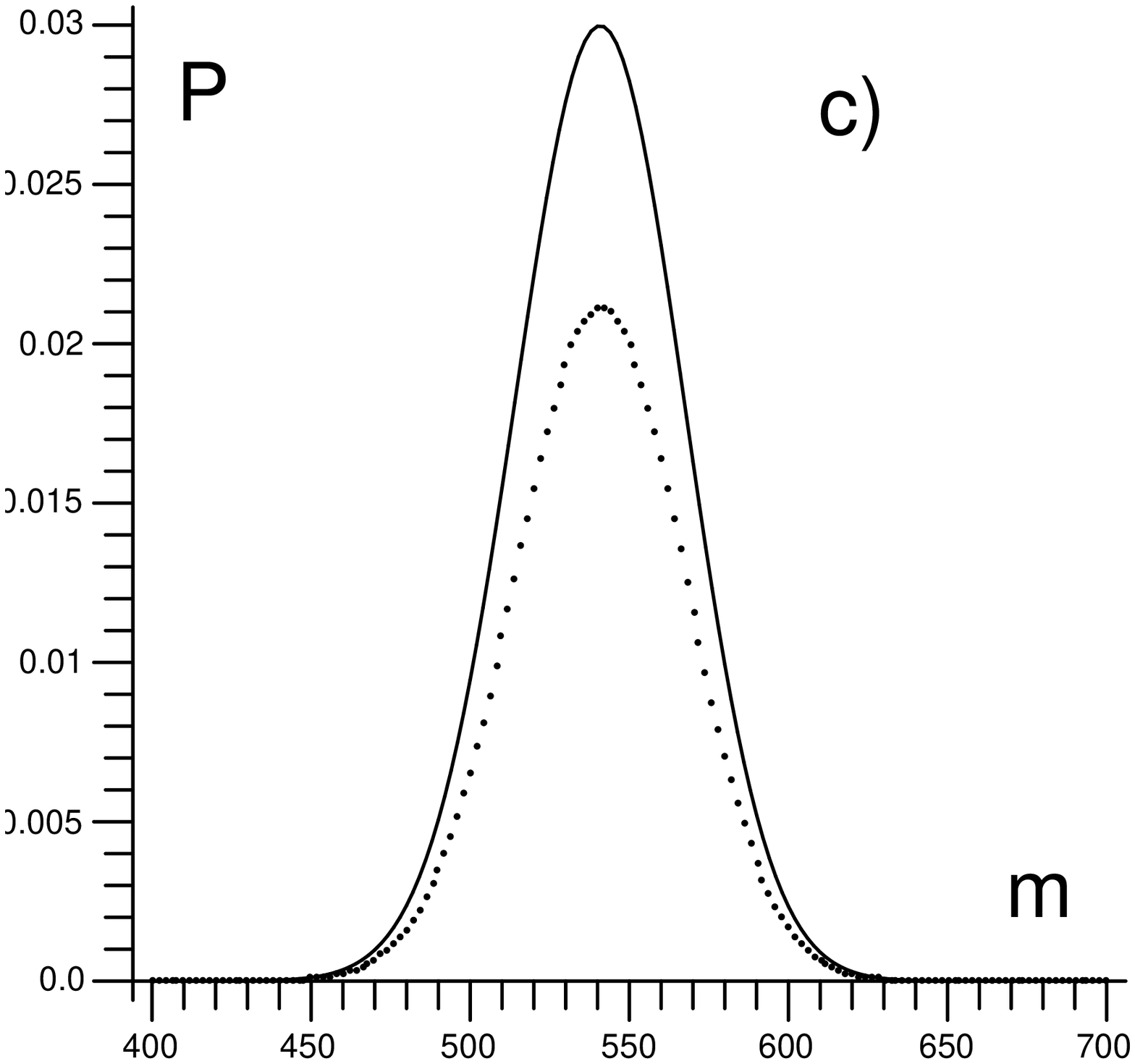}
\caption{The unperturbed (solid) and perturbed (dotted) particle distributions, the latter given by eq. (\ref{TBRPD}), for $N=1000$ and $\theta=1$ with \textit{a)} $\sigma= -\frac{1}{1000}$,  \textit{b)} $\sigma= \frac{1}{1000}$ and \textit{c)} $\sigma= \frac{2}{1000}$.}
\label{TBRPlots}
\end{figure}

\begin{figure}[t]
\centering
\includegraphics[width=2in, height=2in]{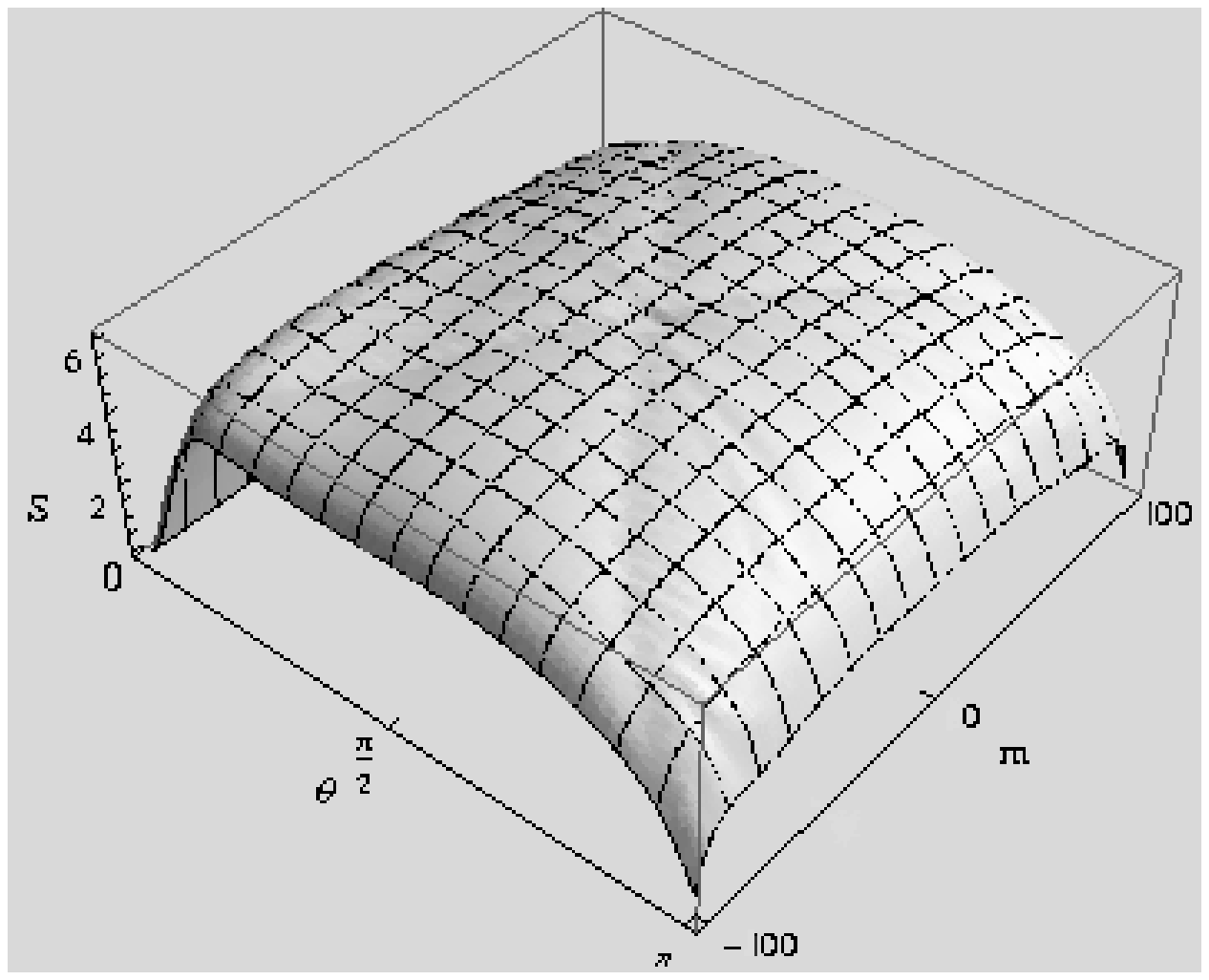}
\hspace{0.4in}
\includegraphics[width=2in, height=2in]{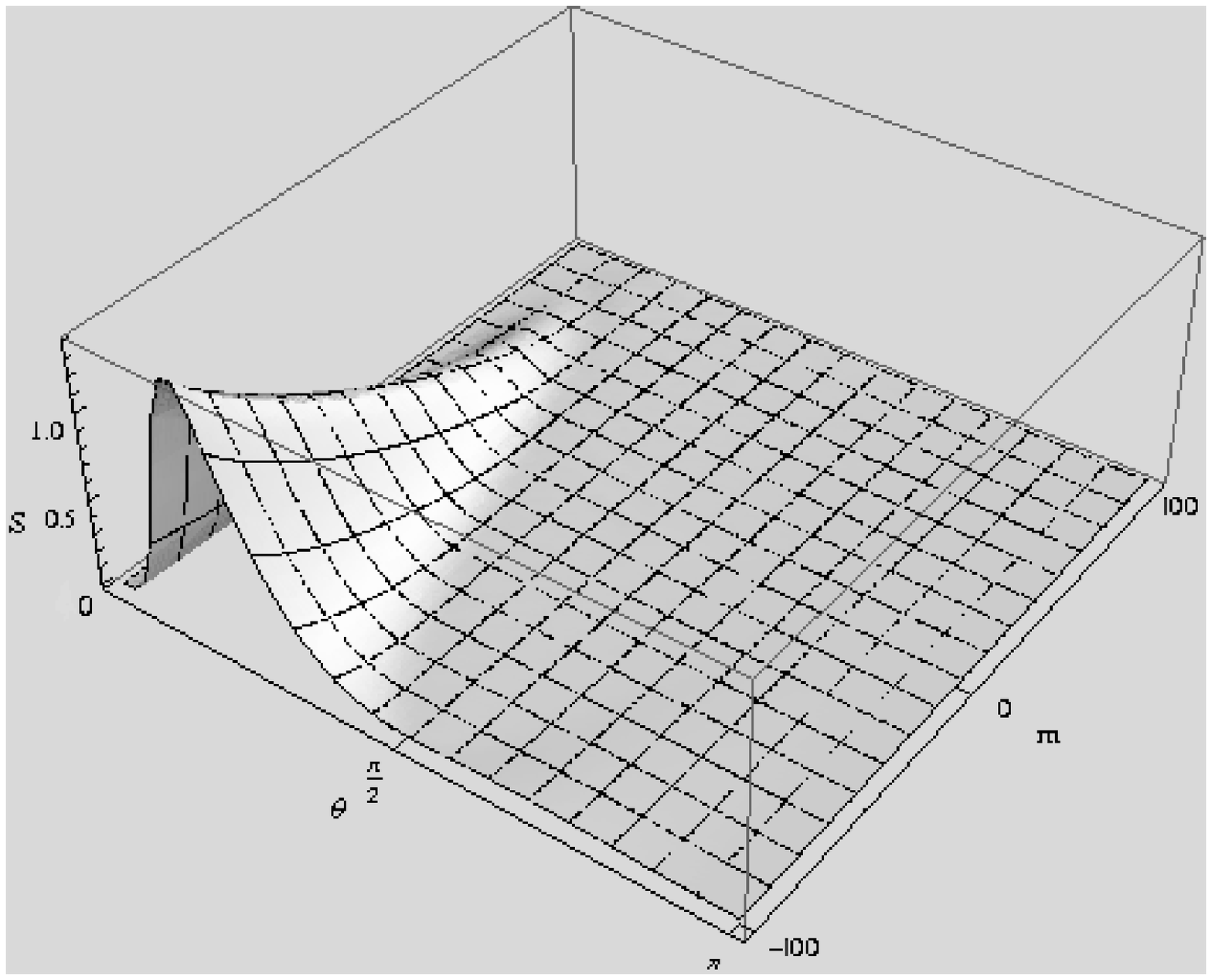}
\hspace{0.4in}
\caption{The figures show the \textit{a)} perturbed entanglement for $\sigma=0.5$ as a function of $m_0$ and $\theta$ for $N=100$ and \textit{b)} the difference $S^{(0+1)}-S^{(0)}$ for $N=100$ and $\sigma=0.5$.}
\label{ent_TBR}
\end{figure}

\section{Discussion}

We have successfully studied the effects of a number of perturbations to the two-mode BEC model considered in \cite{fuent1}.  We have found the corrections to the condensate wave functions, which in turn allowed the determination of the corrections to the particle distribution, time-evolution of the relative number operator and the entanglement.  In the non-degenerate case, we have shown that the model of \cite{fuent1}, in which the coupling constants of the Hamiltonian are constrained, is robust to perturbations in these constants.  The system is most sensitive to perturbations in the elastic scattering length $\mathcal{U}$ and in the mode-exchange parameters $\Lambda$ and $\mu$.  In each of the parameter perturbations the entanglement of the coherent states $\vert N, N\rangle$ and $\vert N, -N\rangle$ is most affected.  Each parameter perturbation was observed to increase the asymmetry in the particle distributions. We also observed that for specific values of $m_0$ and $\theta$, the latter corresponding to certain parameter strengths, the entanglement perturbations are especially large.  It was found that when the condensate is degenerate (because of specific choices of $A_1$ and $A_2$) it is much more sensitive to perturbations, both in terms of particle distributions and entanglement.  The effects on the entanglement are qualitatively different than in the non-degenerate case.  In particular, the perturbations to the entanglement are mainly present only in the regions in which $m_0$ is close to $-N$ which corresponds to a condensate with small negative scattering length.

We have also extended the formalism to include the analysis of interactions involving particle loss.  From these we can predict corrections to the particle distribution, entanglement and evolution of the relative population.  This provides a new class of possible experiments that will allow the model here to be tested. Indeed, interactions involving particle loss have been limited experimentally, partially because they create instabilities in the system.  With our results these interactions could be allowed to occur, and the results compared with the predictions contained above. The changes induced by both three-body recombination and background collisions are qualitatively different than those induced by parameter perturbations.  We find that the system is very sensitive to these external perturbations, which induce large changes in the particle distribution and entanglement from relatively small external coupling strengths, as compared to the induced changes from parameter perturbations.  As with the degenerate parameter perturbation, we found that the perturbative effects on the entanglement become negligible as $m_0$ approaches $N$, i.e. when the scattering length is small and positive. In general, we can conclude that higher collision rates make the condensate more stable to perturbations.

Our results promise to be useful in the experimental realization of two-mode Bose-Einstein condensates which are stable to parameter perturbations and particle loss.  We are planning to extend our analysis to include many-body interactions which are present at cooler stages of the condensate and study the role of such interactions in the stability of the condensate.

\section*{Appendix: The Schwinger $\mathfrak{su}(2)$ Boson Representation}

Let $J_z$ and $J_{\pm}=J_x \pm i J_y$ be the usual generators of the Lie algebra $\mathfrak{su}(2)$, satisfying
\begin{equation}
[J_z ,J_{\pm} ]=\pm J_{\pm}, \;\;\;\; [J_+ , J_- ]=2J_z.
\end{equation}
The eigenstates are labeled as $\vert j,m_{ang} \rangle$ where
\[
J_z \vert j, m_{ang} \rangle= m_{ang} \vert j , m_{ang} \rangle, \;\;\; J^2\vert j, m_{ang} \rangle = j\left(j+1\right)\vert j, m_{ang} \rangle.
\]
Note that for fixed $j$, $m_{ang}$ may take any of the $2j+1$ values $-j, -j +\frac{1}{2}, \dots, j-\frac{1}{2}, j$.  We also have $J_{\pm} \vert j, m_{ang} \rangle= \sqrt{j(j+1)-m_{ang}(m_{ang} \pm 1)} \vert j, m_{ang} \pm 1 \rangle$.  The Schwinger representation of $\mathfrak{su}(2)$ defines a Lie algebra homomorphism between the angular momentum representation (generated by $J_{\pm}$ and $J_z$) and a bosonic representation.  Consider two bose operators $a$ and $b$ with $\left[ a, a^{\dagger} \right] =1 =\left[b,b^{\dagger}\right]$ and all other pairs having vanishing commutator. Defining the mapping
\begin{equation}
J_+ \mapsto a^{\dagger}b, \;\;\; J_- \mapsto ab^{\dagger}, \;\;\; J_z \mapsto \frac{1}{2}\left(a^{\dagger}a - b^{\dagger}b \right)
\end{equation}
and extending linearly to the rest of $\mathfrak{su}(2)$ then gives the desired homomorphism \cite{schwi}.  A short calculation then shows that $J^2$ is mapped to $\frac{1}{4}\hat{N}^2 + \frac{1}{2}\hat{N}$ where we have defined $\hat{N}\equiv a^{\dagger}a +b^{\dagger}b$.  We label the basis states in the bosonic representation as $\left| \frac{N}{2}, \frac{m}{2} \right>$ with $N$ the eigenvalue of $\hat{N}$ and $m$ the eigenvalue of $\hat{m}\equiv a^{\dagger}a-b^{\dagger}b$; this is in complete analogy with the label $\vert j, m_{ang} \rangle$.  Note that $\frac{m}{2}$ may take the $2\left(\frac{N}{2}\right) +1 =N+1$ values $-\frac{N}{2},\; -\frac{N}{2}+1, \dots , \frac{N}{2}-1,\; \frac{N}{2}$.  From the definition of the homomorphism it follows that $a^{\dagger}b \left| \frac{N}{2}, \frac{m}{2} \right> = \frac{1}{2}\sqrt{N(N+2)-m(m+2)}\left| \frac{N}{2}, \frac{m}{2} +1 \right>$.  If we then rescale the state $\left| \frac{N}{2}, \frac{m}{2} \right>$ to $\left| N, m \right>$  the above identity is rewritten as
\begin{equation}
a^{\dagger}b \left| N, m \right> = \frac{1}{2}\sqrt{N(N+2)-m(m+2)}\left| N, m+2 \right>.
\end{equation}
That $\left| N, m \right>$ is mapped to a multiple of $\left| N, m+2 \right>$ can be seen directly from the form of the operator $a^{\dagger}b$, which annihilates a particle in mode $b$ while creating one in mode $a$.  Similarly we have
\begin{equation}
ab^{\dagger} \left| N, m \right> = \frac{1}{2}\sqrt{N(N+2)-m(m-2)}\left| N, m-2 \right>.
\end{equation}
Now rescaled, $m$ may take on the $N+1$ values $-N,\; -N+2, \dots,\; N-2, \; N$.

\section*{Appendix: Counting Terms in the $n$-Model Hamiltonian}

It is of interest to quantify the generality of the $n$-model Hamiltonian under study.  We begin this below by first counting the number of terms in the most general Hamiltonian that would be of interest to us.  To do this we must define precisely the Hamiltonians that are of interest to our study of two-mode BECs.  First, we limit ourselves to Hamiltonians that are polynomials in the bose operators.  The interactions under consideration consist of any total number preserving operations.  Operators corresponding to such interactions must then commute with the total number operator $\hat{N}$.  To restrict the class of interactions, we consider only those that have no intermediate interactions, such as the spontaneous creation and annihilation of a particle.  This imposes the restriction that the Hamiltonian be normal ordered.  As usual, we also require self-adjointness of the Hamiltonian.  Finally, we may decompose the Hamiltonian into its homogeneous parts, i.e. terms of degree $1,2,3,\dots$.  So, it is sufficient to first consider only homogeneous Hamiltonians, and then construct more general Hamiltonians from these.  With these assumptions we prove the following proposition.

\begin{Prop}
Let $H$ be a homogeneous, self-adjoint, normal-ordered polynomial in the bose operators $a$ and $b$ and their adjoints.  Furthermore, assume that each term in $H$ commutes with the total number operator $\hat{N}=a^{\dagger}a+b^{\dagger}b$. Put $\deg (H) = n \geq 0$.  Then the number of terms in $H$ is at most $\frac{(n+2)(n+4)}{8}$.
\end{Prop}
\begin{proof}
Let $\chi (n)$ denote the number of terms in the Hamiltonian.  Observe that the requirement that each term in $H$ commute with $\hat{N}$ ensures that $n=2k$ for some $k\in\mathbb{N}$.  Consider now a monomial of $H$ with $2k-p$ of the operators being either $a$ or $a^{\dagger}$ while the remaining $p$ being either $b$ or $b^{\dagger}$.  It is easy to check that, assuming self-adjointness, normal ordering and vanishing commutator with $\hat{N}$, there are $\lfloor \frac{p}{2} \rfloor + 1$ such terms, where for $q\in\mathbb{R}$, $\lfloor q \rfloor$ is the greatest integer less than or equal to $q$. We may group the monomials of $H$ into three groups according to whether there are more than, less than, or the same number of $a$ mode terms as $b$ mode terms. Doing so, we find that the maximal number of terms in $H$ is
\[
\chi(n=2k) = 2\sum_{p=0}^{k-1} \left( \lfloor \frac{p}{2} \rfloor +1 \right) + \lfloor \frac{k}{2} \rfloor +1.
\]
Since $\lfloor \frac{k-1}{2} \rfloor + \lfloor \frac{k}{2} \rfloor = k-1$ for all $k\in\mathbb{Z}$ we obtain the recursion relation $\chi(2k)=\chi(2(k-1))+k+1$.  Repeated application of this relation yields $\chi(2k)= \chi(2) + \sum_{j=3}^{k+1} j$, which we may rewrite as $\displaystyle \chi (2k)= \sum_{j=1}^{k+1} j$ since $\chi(2)=3$, from which the proposition follows.
\end{proof}

Since we can decompose a general polynomial in terms of its monomials, the above proposition is sufficient to count the maximal number of terms in the most general Hamiltonian described above.

\begin{Cor}
Let $H$ be as above without the assumption of homogeneity with $\deg (H)= n \geq 0$.  Then the number of terms in $H$ is at most $\frac{1}{48}n^3+\frac{1}{4}n^2+\frac{11}{12}n+1$.
\end{Cor}
\begin{proof}
Again, we can write $n=2k$ for some integer $k\geq 0$.  The number of terms in the Hamiltonian is the sum of the number of terms in each of its homogeneous parts, i.e.  $\sum_{j=0}^k \chi(2j)$. After some algebra the corollary follows.
\end{proof}

As noted in the proofs of the above proposition and corollary, the assumptions on the Hamiltonian imply $deg(H)=n$ is even. We see that the total number of terms in the most general $n$-model Hamiltonian grows like $n^3$.  We would like to compare this result to the number of terms in the $n$-model of \cite{fuent1,fuent2}.  Unfortunately, we have not been able to successfully count the number of terms in the $n$-model. In order to do so, one would first need to conjugate
\[
H_{0,n}=\sum_{i=0}^n A_i \left( a^{\dagger}a-b^{\dagger}b \right)^i
\]
by the displacement operator $U (\xi)$.  The number of terms in this calculation grows quickly with $n$ which makes the conjecture of a formula for the number of terms difficult.  Moreover, general arguments to count the number of terms, as in the proposition above, seem difficult to make in this case.  Regardless, the results for $n=1,2$ and $3$, summarized in Table \ref{count_tab}, suggest that there are a significant number of terms missed by the $n$-model. Whether or not  these terms are important in an experimental setting is another question.

\begin{table}[ht]
\caption{Comparison of the number of terms in $n$-model and the most general model}
\centering
\begin{tabular}{c c c c}
\hline\hline
$n$ & Terms in $n$-model & Terms in general model & Missed Terms  \\ [0.5ex]

\hline % inserts single horizontal line
0 & 1 & 1 & 0 \\ % inserting body of the table
1 & 3 & 4 & 1 \\
2 & 6 & 10 & 4 \\
3 &	13 & 20 & 7\\
 [1ex] % [1ex] adds vertical space
\hline %inserts single line
\end{tabular}
\label{count_tab}
\end{table}

\section*{Acknowledgments}

This work was supported by the Natural Sciences \& Engineering Research
Council of Canada. M.B.Y. was also partially supported by a Renaissance Technologies fellowship, and would like to thank Andrew Louca and Paul McGrath for discussions.
I. F-S was supported by the Alexander von Humboldt Foundation and would like to thank Tobias Brandes, Carsten Henkel and Martin Wilkens for their hospitality and P. Barberis-Blostein for useful comments and discussions.

\end{document}